\documentclass[10pt, conference, letterpaper]{IEEEtran}

\title{Source Localization in Networks: Trees and Beyond}
\author{\IEEEauthorblockN{Kai Zhu and Lei Ying}\\
\IEEEauthorblockA{School of Electrical, Computer and Energy Engineering\\ Arizona State University\\
Tempe, AZ, United States, 85287\\
Email: kzhu17@asu.edu, lei.ying.2@asu.edu}}
\date{}
\usepackage{enumerate} 
\usepackage[cmex10]{amsmath} 
\allowdisplaybreaks[4] 
\usepackage{amssymb} 
\usepackage{amsthm} 
\usepackage{hyperref} 
\usepackage[sc,osf,noBBpl]{mathpazo} 
\usepackage{graphicx}
\usepackage{fancybox}
\usepackage{epstopdf}
\usepackage[ruled,vlined]{algorithm2e}
\usepackage{xcolor}
\usepackage{slashbox,pict2e}
\usepackage{tablefootnote}

\usepackage{subfig}

\newtheorem{thm}{Theorem}
\newtheorem{lemma}{Lemma}

\newcommand{\obs}{{\cal O}}
\newcommand{\tree}{T}
\newcommand{\treeSet}{{\cal T}}
\newcommand{\obsTime}{t}
\newcommand{\ecce}{e}
\newcommand{\edges}{{\cal E}}
\newcommand{\nodes}{{\cal V}}
\newcommand{\infObs}{{\cal I}}
\newcommand{\healthObs}{{\cal H}}
\newcommand{\dist}{d}
\newcommand{\boundary}{{\cal B}}
\newcommand{\degree}{{\hbox{deg}}}
\newcommand{\source}{s}
\newcommand{\levelset}{{\cal L}}
\newcommand{\graph}{g}
\newcommand{\collisionset}{{\cal R}}
\newcommand{\collisionsize}{R}
\newcommand{\level}{l}
\newcommand{\offspring}{\Phi}
\newcommand{\offspringsize}{\phi}

\begin{document}
\maketitle
\begin{abstract}
Information diffusion in networks can be used to model many real-world phenomena, including rumor spreading on online social networks, epidemics in human beings, and malware on the Internet. Informally speaking, the source localization problem is to identify a node in the network that provides the best explanation of the observed diffusion. Despite significant efforts and successes over last few years, theoretical guarantees of source localization algorithms were established only for tree networks due to the complexity of the problem. This paper presents a new source localization algorithm, called the Short-Fat Tree (SFT) algorithm. Loosely speaking, the algorithm selects the node such that the breadth-first search (BFS) tree from the node has the minimum depth but the maximum number of leaf nodes. Performance guarantees of SFT under the independent cascade (IC) model are established for both tree networks and the Erdos-Renyi (ER) random graph. On tree networks, SFT is the maximum a posterior (MAP) estimator. On the ER random graph, the following fundamental limits have been obtained: $(i)$ when the infection duration $<\frac{2}{3}t_u,$  SFT identifies the source with probability one asymptotically, where $t_u=\left\lceil\frac{\log n}{\log \mu}\right\rceil+2$ and $\mu$ is the average node degree, $(ii)$ when the infection duration $>t_u,$ the probability of identifying the source approaches zero asymptotically under any algorithm; and $(iii)$ when infection duration $<t_u,$ the BFS tree starting from the source is a fat tree. Numerical experiments on tree networks, the ER random graphs and real world networks with different evaluation metrics show that the SFT algorithm outperforms existing algorithms.
\end{abstract}

\section{Introduction}\label{sec:introduction}
The information source detection problem (or called rumor source detection problem) is to identify the source of information diffusion in networks based on available observations like the states of the nodes and the timestamps at which nodes adopted the information (or called infected). The solution of the problem can be used to answer a wide range of important questions. For example, in epidemiology, the knowledge of the epidemic source has been used to understand the transmission media of the disease \cite{John_1854}.  For a computer virus spreading on the Internet, tracing the source helps locate the virus creator. For the news over the social media, locating the sources helps users verify the credibility of the news.

Because of its wide range of applications, the problem has gained a lot of attention in the last few years since the seminal work by Shah and Zaman \cite{ShaZam_11}. A number of effective information source detection algorithms have been proposed under different diffusion models. Despite significant efforts and successes, theoretical guarantees have been established only for tree networks due to the complexity of the problem in non-tree networks. In this paper, we first develop a new information source detection algorithm, called the Short-Fat Tree algorithm, and then present a comprehensive performance analysis of the algorithm under the IC model for both tree networks and the ER random graph. To the best of our knowledge, SFT is the first algorithm that has provable performance guarantees on both tree networks and the ER random graph \cite{ErdRen_59} (non-tree networks).

The fundamental possibility and impossibility results are summarized as follows.
\begin{enumerate}
	\item For tree networks, we prove that the Jordan infection center with the maximum weighted boundary node degree (WBND) is the MAP estimator of the source under the heterogeneous IC model.
	Based on the derivation, we propose an algorithm called the Short-Fat Tree (SFT) algorithm which is applicable to both tree and general networks.
\item We analyze the performance of the SFT algorithm on the ER random graph. Under some mild conditions on the average node degree, we establish the following three results:
\begin{itemize}
\item[(i)] Assume the infection duration $<\frac{\log n}{(1+\alpha)\log\mu}$ for some $\alpha>0.5,$ SFT identifies the source with probability 1 (w.p.1) asymptotically (as network size increases to infinity).

\item[(ii)] Assume the infection duration $\geq \left\lceil\frac{\log n}{\log \mu}\right\rceil+2,$ the probability of identifying the source approaches zero asymptotically under any information source detection algorithm, i.e., it is impossible to detect the source with non-zero probability.

\item[(iii)] Assume the infection duration $<\frac{\log n}{(1+\alpha)\log\mu}$ for some $\alpha>0,$ asymptotically, at least $1-\delta$ fraction of the nodes on the BFS-tree starting from the source are leaf-nodes, where $\delta>3\sqrt{\frac{\log n}{\mu}}$. This result does not provide any guarantee on the probability of correctly localizing the source, but states that the BFS-tree starting from the true source is a ``fat'' tree, which further justifies the SFT algorithm.
\end{itemize}
The results are summarized in Figure \ref{fig:PerformanceSummary}.
\begin{figure}
        \centering
 		\includegraphics[width=0.8\columnwidth]{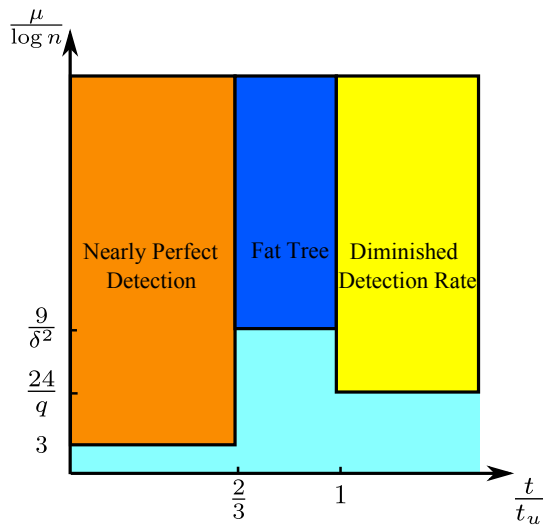}
        \caption{Summary of the main results. This figure summarizes the key results in terms of $t,$ the infection time, and $\mu,$ the average node degree. In the figure, $q$ is the lower bound on the infection probability; ``fat tree'' means that there are $1-\delta$ fraction of nodes are boundary nodes on the BFS tree rooted at the source; and $t_u=\left\lceil\frac{\log n}{\log \mu}\right\rceil+2,$ which is the lower bound of the observation time (we proved that all algorithms fail when $t>t_u.$)}\label{fig:PerformanceSummary}
\end{figure}
We remark that results (i) and (iii) are highly nontrivial because a subgraph of the ER random graph is a tree with high probability only when the diameter is $\frac{\log n}{2\log \mu}$, and (i) and (iii) deal with subgraphs that are not trees. To the best of our knowledge, these are the first theoretical results on information source detection on non-tree networks under probabilistic diffusion models.

\item One drawback of the WBND tie-breaking is that it requires the infection probabilities of all edges in the IC model. We simplify WBND to BND by using the boundary node degree in SFT. As shown in Section \ref{sec:performance-evaluation}, the performance of BND tie-breaking is very close to WBND tie-breaking. We conducted extensive simulations on trees, ER random graphs and real world networks. SFT outperforms existing algorithms by having a higher detection rate and being closer to the actual source. We further evaluated the scalability of the algorithm by measuring the running time. Our results demonstrate that SFT achieves a better performance with a reasonably short execution time.
\end{enumerate}

\subsection{Related Work}\label{sec:relatedWork}
\cite{ShaZam_11} is one of the first papers that study the information source detection problem, in which a new graph centrality called rumor centrality was proposed and proved to be the maximum likelihood estimator (MLE) on regular trees under the susceptible-infected (SI) model. In addition, the detection probability (the probability that the estimator is the source) for regular trees was proved to be greater than zero and the detection probability for geometric trees approaches one asymptotically as the increase of the spreading time. Later, \cite{ShaZam_12} quantified the detection probability of the rumor centrality on general random trees.

The rumor centrality has been further studied under different scenarios: 1) \cite{LuoTayLen_13} extended the rumor centrality to multiple sources and showed that the detection probability goes to one as the number of infected nodes increases for geometric trees when there are at most two sources; 2) \cite{KarFra_13} proved a similar performance guarantee for the single source case when only a subset of infected nodes are observed; 3) \cite{DonZhaTan_13} studied the detection probability when the prior knowledge of suspect nodes is available in the single source detection problem for  trees; 4) \cite{WanDonZha_14} analyzed the detection probability of the rumor centrality for tree networks when there are multiple observations of independent diffusion processes from the same source.

\cite{ZhuYin_14_2} proposed the sample path based approach for the single source detection problem. Define the infection eccentricity of a node to be the maximum distance between the node and the infected nodes. \cite{ZhuYin_14_2} proved that on tree networks, under the homogeneous susceptible-infected-recovered (SIR) model,  the root of the most likely sample path is a node with the minimum infection eccentricity (a Jordan infection center), which is within a constant distance to the actual source with a high probability. The approach has been extended to several directions: 1) \cite{ZhuYin_14} extended the approach to the case with partial observations and under the heterogeneous SIR model; 2) \cite{CheZhuYin_14} extended the analysis to multiple sources under the SIR model and proved that the distance between the estimator and its closest actual source is bounded by a constant with a high probability in tree networks; 3) \cite{LuoTay_13_2,LuoTay_13} proved that the Jordan infection centers are the optimal sample path estimators under the SI model \cite{LuoTay_13_2} and the susceptible-infected-susceptible (SIS) model \cite{LuoTay_13} for tree networks, respectively.

Besides the rumor centrality and the Jordan infection center,  several other heuristic algorithms based on a single snapshot of the network have been proposed in the literature: 1) \cite{LapTerGun_10} studied a similar problem under the independent cascade (IC) model \cite{GolLibMul_01} to minimize the l1 distance between the expected states and observed states of the nodes. A dynamic programming algorithm was proposed to solve the problem for tree networks and a Steiner tree heuristic was used for general networks; 2) \cite{PraVreFal_12} proposed an algorithm called NETSLEUTH which ranks the nodes according to an eigen vector based metric under the SI model. The algorithm was designed based on the Minimum Description Length principle; 3) \cite{LokMezOht_14} proposed a dynamic message passing algorithm based on the mean field approximation of the maximum likelihood estimation (MLE) of the source.

In addition, there exist several other algorithms which tackled the problem under the assumption that a subset of the infection timestamps are known: 1) \cite{PinThiVet_12} solved the MLE problem with partial timestamps for tree networks and extended the algorithm to general networks using a BFS tree heuristic; 2) \cite{ZhuCheYin_15} proposed two rank based algorithms using a modified BFS tree heuristic for general graphs; 3) \cite{AgaLu_13} proposed a simulation based Monte Carlo algorithm which utilizes the states of the sparsely placed observers within a fixed time window; 4) \cite{ZejGomSin_13} obtained sufficient conditions on the number of timestamps needed to locate the source correctly under the deterministic slotted SI models.

Our paper establishes possibility and impossibility results of SFT beyond tree networks, which differs it from the existing work mentioned above. The rest of the paper is organized as follows. In Section \ref{sec:model}, we first introduce the IC model and formulate a MAP problem for information source detection and SFT will be presented in Section \ref{sec:algorithm}. Section \ref{sec:mainresult} summarizes the main theoretical results of the paper including the analysis on both tree networks and the ER random graph.  The simulation based performance evaluation will be presented in Section \ref{sec:performance-evaluation}. All the proofs are provided in  the appendices.

\section{Model and Algorithm}\label{sec:model-algorithm-mainresult}

\subsection{Model}\label{sec:model}

Given an undirected graph $g,$ denote by $\edges(g)$ the set of edges in $g$ and denote by $\nodes(g)$ the set of nodes in $g.$ We consider the IC model \cite{GolLibMul_01} for information diffusion and assume a time-slotted system.  Each node has two possible states: active (or called infected) and inactive (or called susceptible). At time slot $0,$ all nodes are inactive except the source. At the beginning of each time slot, an active node attempts to activate its inactive neighbors. If an attempt is successful, the corresponding node becomes active at next time slot; otherwise, the node remains inactive. The weight of each edge represents the success probability of the attempt, called the \emph{infection probability} of the edge and each attempt is independent of others. Each active node only attempts to activate each of its inactive neighbors once. Denote by $q_{uv}$ the infection probability of edge $(u,v)$ and we assume $q_{uv}=q_{vu}$ throughout the paper since the graph is undirected. We assume that a complete snapshot $\obs=\{\infObs,\healthObs\}$ of the network at time $t$ (called the \emph{observation time}) is given,  where $\infObs$ is the set of active nodes and $\healthObs$ is the set of inactive nodes. Based on $\obs,$ we want to detect the source. We further assume the observation time $\obsTime$ is unknown. The problem can be formulated as a MAP problem as follows,
\[
\arg\max_{v\in\nodes(g)} \Pr(v|\obs).
\]
where $\Pr(v|\obs)$ is the probability that $v$ is the source given the snapshot $\obs.$ The infected nodes form a connected component under the IC model, called the \emph{infection subgraph} and denoted by $g_i.$ Since the source must be an infected node, the MAP problem can be simplified to
\[
\arg\max_{v\in\infObs} \Pr(v|\obs),
\]
and the search of the information source can be restricted to the infection subgraph. We assume the observation time $\obsTime,$ which itself is a random variable, is independent of the source node.

\subsection{The Short-Fat Tree Algorithm}\label{sec:algorithm}
In this section, we first present the SFT algorithm. We will show in Theorem \ref{thm:treepartialMAP} that the algorithm outputs the MAP estimator for tree networks, which motivates the algorithm. The performance on the ER random graph is studied in Theorems \ref{thm:sourceIsJordanER} and \ref{thm:approximateRatio}.

We first introduce several necessary definitions. Denote by $\dist^g_{uv}$ the distance from node $u$ to node $v$ in graph $g,$ where the distance is the minimum number of hops between two nodes.  Define the \emph{infection eccentricity} of an infected node to be the maximum distance from the node to all infected nodes on the infection subgraph $g_i,$ denote by $\ecce(v,\infObs),$
\[
\ecce(v,\infObs)=\max_{u\in \infObs}\dist^{g_i}_{uv}.
\]
Recall that the \emph{Jordan infection centers} of a graph are the nodes with the minimum infection eccentricity \cite{ZhuYin_14_2}.

Consider a BFS tree $T_v$ rooted at node $v$ on the infection subgraph $g_i.$ Denote by $\hbox{par}_v(u)$ the parent of node $u$ in $T_v.$
Define the set of \emph{boundary nodes} of $T_v$ to be
\[
\boundary(v,\infObs)=\{w\in \infObs|\dist^{T_v}_{vw}=\ecce(v,\infObs)\},
\]which are the set of active nodes furthest away from node $v$ in the infection subgraph.

The weighted boundary node degree (WBND) with respect to node $v$ is defined to be
\begin{align}
\sum_{(u,w)\in{{\cal F}'_v}} |\log (1-q_{uw})|,\label{eqn:boundaryDegreeCountNew}
\end{align}
where
\begin{align}
{\cal F}'_v=\{(u,w)|(u,w)\in\edges(g),w\neq \hbox{par}_v(u), u\in \boundary(v,\infObs)\}.\label{eqn:FvDefinition}
\end{align}

The SFT algorithm, presented in Algorithm \ref{alg:RI}, identifies the source based on the BFS trees on the infection subgraph.  The algorithm is called the \emph{Short-Fat Tree} algorithm because (1) it first identifies the \emph{shortest} BFS tree; and (2) the shortest BFS tree that maximizes the WBND is then selected in tie-breaking, which is usually the tree with a large number of leaf-nodes, i.e., a \emph{fat} tree. The pseudo codes of the algorithms are presented in Algorithm \ref{alg:RI} and \ref{alg:WBNDC}, which can be executed in a parallel fashion.

A simple example is presented in Figure \ref{fig:AlgorithmExample} to illustrate algorithm. Each node has a unique node ID. The red nodes are infected and the white nodes are healthy. For simplicity, we assume the weights of all edges equal to $|\log(0.5)|$. The vector next to each infected node records the distance from it to all infected nodes. Initially at Iteration 0, each infected node only knows the distance to itself. For example, $[0$ $*$ $*$ $*]$ next to node 1 means that the distance from node 1 to itself is 0 and the distance from node 1 to node 2 is unknown. At Iteration 1, each infected node broadcasts its ID to its neighbors in next iteration. Upon receiving the node ID from node 1, node 2 updates its vector to $[1$ $0$ $*$ $*],$ and broadcasts node 1's ID to its neighbors. The figure in the middle shows the updated vectors after all node ID exchanges occur at Iteration 1. At Iteration 2, node 1 and 2 do not receive any new node IDs. Therefore, node 1 and node 2 report themselves as the Jordan infection centers which are circled with blue in Figure \ref{fig:AlgorithmExample}. The boundary nodes of the BFS tree rooted at node 1 are 2,3,4. The WBND of node 1 is $13|\log(0.5)|$. Similarly, the boundary nodes of the BFS tree rooted at node 2 are 1,3,4 and the WBND is $9|\log(0.5)|.$ Therefore, node 1 has a larger WBND and is chosen to be our estimator of the information source.

\begin{figure*}
        \centering
		\includegraphics[width=0.75\textwidth]{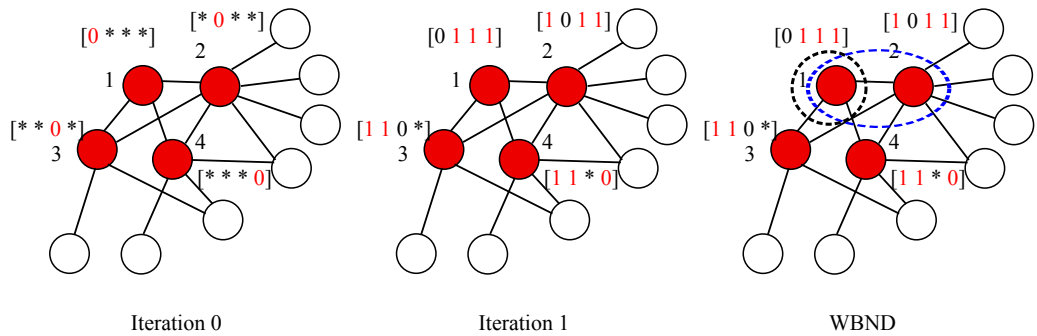}
        \caption{An example of the Short-Fat Tree algorithm}\label{fig:AlgorithmExample}
\end{figure*}

\begin{algorithm}
\SetAlgoLined

 \KwIn{$\infObs, g$;}
 \KwOut{$v^\dag$ (the estimator of information source)}
 Set subgraph $g_i$ to be a subgraph of $g$ induced by node set $\infObs.$

 \For {$v \in \infObs$}{
 Initialize an empty dictionary $D_v$ associating with node $v.$

 Set $D_v[v] = 0.$
 }
 Each node receives its own node ID at time slot $0.$

 Set time slot $t=1.$

\SetKwRepeat{doWhile}{do}{while}
 \doWhile{No node receives $|\infObs|$ distinct node IDs}{
	\For{$v\in \infObs$}{
		\If{$v$ received new node IDs in $t-1$ time slot, where ``new'' IDs means node $v$ did not receive them before time slot $t-1$}{
			$v$ broadcasts the new node IDs to its neighbors in $g_i$.
		}

	}
	\For{$v\in \infObs$}{
			\If{$v$ receives a new node ID $u$ which is not in $D_v.$}{
			Set $D_v[u]=t.$
			}
	}
	
	$t=t+1.$
 }

 Set ${\cal S}$ to be the set of nodes who receive $|{\cal I}|$ distinct node IDs.

 \For{$v\in {\cal S}$}
 {
	 Compute WBND of $\tree_v$ using Algorithm \ref{alg:WBNDC}.
 }

 \Return $v^\dag\in{\cal S}$ with the maximum WBND.
 \caption{The Short-Fat Tree Algorithm}\label{alg:RI}
\end{algorithm}

\begin{algorithm}
\SetAlgoLined

 \KwIn{$v$,$D_v$ (Dictionary of distance from $v$ to other nodes), $g$, $\infObs$, $t$;}
 \KwOut{WBND$(v)$ }

 Set ${\cal B}$ to be empty.

 \For {$u$ in the keys of $D_v$}
 {
  \If {$D_v[u]=t$}{ Add $u$ to ${\cal B}.$}
 }
 Set $x=0;$

 \For {$w\in{\cal B}$ }
 {
	 Find the neighbor $u$ of $w$ such that $D_v[u] = t-1.$
	
	 Set $x = x + \sum_{y\in \hbox{neighbors}(w)}|\log(1-q_{wy})|-|\log(1-q_{wu})|.$
 }

 \Return $x.$
 \caption{The WBND Algorithm}\label{alg:WBNDC}
\end{algorithm}

{\bf Remark:} Note Equation (\ref{eqn:boundaryDegreeCountNew}) requires the infection probabilities of all edges in the network which could be hard to obtain in practice. When the infection probabilities are not available, we can assume each edge has the same infection probability $q$ and WBND becomes,
\begin{align*}
\left(\sum_{u\in\boundary(v,\infObs)}\degree(u)-|\boundary(v,\infObs)|\right)|\log(1-q)|,
\end{align*}
where $\degree(u)$ is the degree of node $u.$

Define the boundary node degree (BND) of node $v$ to be
\begin{align}
\sum_{u\in\boundary(v,\infObs)}\degree(u)-|\boundary(v,\infObs)|\label{eqn:boundaryNodeDegreeCountHomo}
\end{align}
which is only related to the degree of the boundary nodes and can be used to replace WBND as the tie-breaking among the Jordan infection center in SFT when the infection probabilities are unknown. As shown in Section \ref{sec:performance-evaluation}, the performance using BND and WBND are similar. To differentiate the two algorithms, we call the algorithm which uses WBND as wSFT and the one which uses BND as SFT. Next, we analyze the complexity of the algorithm.
\begin{thm}\label{thm:RunningTime}
The worst case computational complexity of the SFT algorithm is $O(|\infObs|\degree(\infObs))$ where $\degree(\infObs)$ is the total degree of nodes in $\infObs$ in graph $g.$
\end{thm}
The detailed proof can be found in Appendix \ref{proof:RunningTime}.

\section{Main Results}\label{sec:mainresult}
In this section, we summarize the main results of the paper and present the intuitions of the proofs.
\subsection{Main Result 1 (The MAP Estimator on Tree Networks)}
On tree networks, the Jordan infection center of the infection subgraph with the maximum WBND is a MAP estimator.

\begin{thm}\label{thm:treepartialMAP}
Consider a tree network. Assume the following conditions hold.
\begin{itemize}
\item The probability distribution of the observation time satisfies $\Pr(\obsTime)\geq \Pr(\obsTime+1)$ for all $\obsTime.$
\item The source is uniformly and randomly selected, i.e., $\Pr(u)=\Pr(v).$
\end{itemize}
Denote by ${\cal J}$ the set of Jordan infection centers of the infection subgraph $g_i$. We have
\begin{align}
\arg\max_{u\in{\cal J}} \sum_{(v,w)\in{{\cal F}'_u}}|\log(1-q_{vw})|\subset \arg\max_u\Pr(u|\obs).\label{eqn:boundaryDegreeCount}
\end{align}
where
${\cal F}'_u$ is defined in Equation (\ref{eqn:FvDefinition}).
\end{thm}

The detailed proof can be found in Appendix \ref{sec:TreeJordanIsMAP}. The theorem has been proved in two steps: 1) We show that one of the Jordan infection centers maximizes the posterior probability on tree networks following similar arguments in \cite{ZhuYin_14_2}. In particular, for two neighboring nodes, we show the one with smaller infection eccentricity has a larger posterior probability of being the source. Since there exists a path from any node to a Jordan infection center on the infection subgraph, along which the infection eccentricity strictly decreases, we conclude that a MAP estimator of the source must be a Jordan infection center; 2) Consider the case where the tree network has more than one (at most two according to \cite{Har_91}) Jordan infection centers. When the observation time is larger than the infection eccentricity of the Jordan infection center, the probability of having the observed infected subgraph from any Jordan infection center is the same. When the observation time equals the infection eccentricity, we prove that the probability for a Jordan infection center to be the source is an increasing function of WBND of the BFS tree starting from it.

\subsection{Main Result 2 (Detection with Probability One on the ER Random Graph)}
 Denote by $n$ the number of nodes in the ER random graph and $p$ the wiring probability of the ER random graph. Let $\mu=np.$ Recall that $t$ is the observation time. We show that the Jordan infection center is the actual source in the ER random graph with probability one asymptotically when $t<\frac{\log n}{(1+\alpha)\log \mu},$ which implies that SFT can locate the source w.p.1 asymptotically.
\begin{thm}\label{thm:sourceIsJordanER}
If the following conditions hold, source $\source$ is the only Jordan infection center on the infection subgraph with probability one asymptotically.
\begin{itemize}
\item $\mu> 3\log n.$
\item $\obsTime\leq \frac{\log n}{(1+\alpha)\log\mu},$ for some  $\alpha\in(\frac{1}{2},1).$
\end{itemize}
\end{thm}

We present a brief overview of the proof and the details can be found in Appendix \ref{sec:ERJordanIsSource}. Note the infection eccentricity of the actual source is no larger than the observation time $t.$ We show in the proof that the infection eccentricity of an infected node other than the source is larger than $t.$ Consider the BFS tree $\tree^\dag$ rooted at the actual source $\source.$ A node is said to be on level $i$ if its distance to the source is $i.$ Consider another infected node $\source'.$ Denote by $a(\source')$ the ancestor of $\source'$ on level $1$ of $\tree^\dag.$ As shown in Figure \ref{fig:onebyoneexploration}, the yellow area shows the level $t$ infected nodes on subtree $T_u^{-s},$ which is the subtree of $T^\dag$ rooted at node $u,$ and the distance from $\source'$ to a node in the yellow area is larger than $t$ if any path between the two nodes can only traverse the edges on tree $\tree^\dag$. If $\source'$ has an infection eccentricity no larger than $t,$ there must exist a path from $\source'$ to each node in the yellow area with length no larger than $t.$ Such a path must contain edges that are not in $\tree^\dag$ (we call these edges \emph{collision edges}). We show in the proof that the number of nodes that are within $t$ hops from $\source'$ via collision edges are strictly less than the number of nodes in the yellow area. Therefore, the infection eccentricity of $\source'$ must be larger than $t$, which implies that $\source$ is the only Jordan infection center.

Existing theoretical results in the literature on information source detection problems are only for tree networks. As shown in the proof of Theorem \ref{thm:sourceIsJordanER}, the infection subgraph of the ER random graph is not a tree when $t>\frac{\log n}{2\log\mu}.$ From the best of our knowledge, this result is the first one on non-tree networks.

\begin{figure}[h!]
  \centering
    \includegraphics[width=0.8\columnwidth]{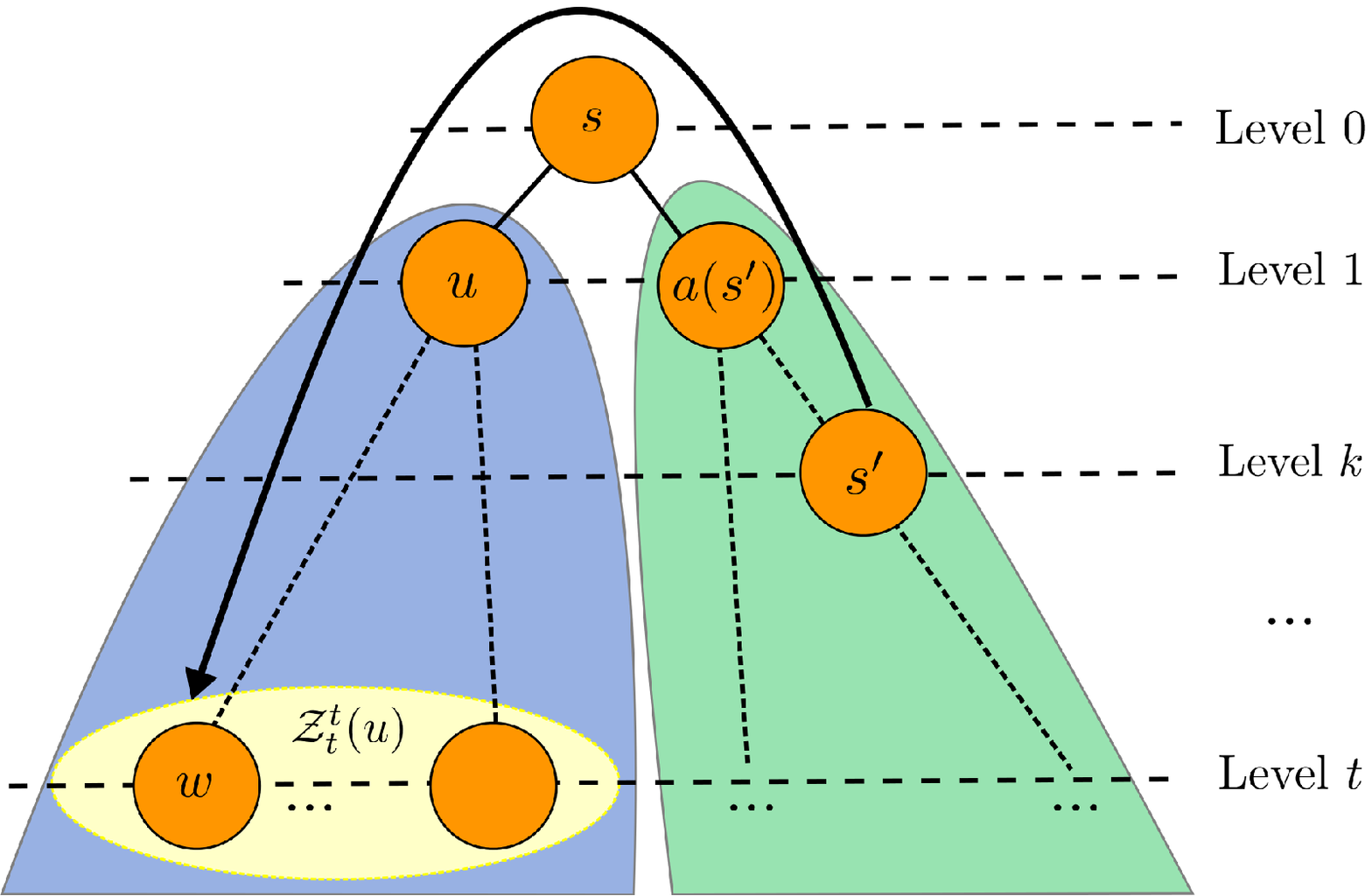}
  \caption{A pictorial example of ${\cal Z}^\obsTime_\obsTime(u)$ in BFS tree $\tree^\dag$}\label{fig:onebyoneexploration}
\end{figure}

\subsection{Main Result 3 (The Fat Tree Result on the ER Random Graph) }

\begin{thm}\label{thm:approximateRatio}
If the following conditions hold,
\begin{itemize}
\item $\mu> \frac{9}{\delta^2}\log n$.
\item $\obsTime\leq \frac{\log n}{(1+\alpha)\log\mu},$ for some $\alpha\in(0,1)$.
\end{itemize}
 the leaf-nodes of the BFS tree starting from the actual source consists of at least $1-\delta$ fraction of the BFS tree asymptotically.
\end{thm}

The detailed proof can be found in Appendix \ref{sec:approximateRate}. Consider the BFS tree from the source $\source$ in graph $g.$ The boundary nodes are the nodes at level $t$ and all boundary nodes must be infected at time $t.$ If we ignore the presence of collision edges, the number of infected nodes roughly increases by a factor of $q\mu$ at each level where $q=\min_{(u,v)\in \edges(g)}q_{uv}.$ Due to this exponential growth nature, the total number of infected nodes is dominated by those infected at the last time slot. We show this property holds with the presence of collision edges. Theorem \ref{thm:approximateRatio} suggests that the BFS tree rooted at the actual source is a ``fat" tree and the BND of the actual source is large. Hence, in the tie breaking, the SFT algorithm has a good chance to select the actual source, which suggests that BND is a good tie breaking rule for the ER random graph.

\subsection{Main Result 4 (The Impossibility Result on the ER Random Graph)}
We next present the threshold of $t$ after which it is impossible for any algorithm to find the actual source with a non-zero probability asymptotically.  The result is based on the analysis of the diameter of an ER random graph in Theorem 4.2 in \cite{DraMas_10}. For clarity purpose, we rephrase that theorem with our notation in the following lemma.
\begin{lemma}\label{thm:diameterER}
If $24\log n<np<<\sqrt{n},$ we have
\[
\lim_{n\rightarrow \infty} \Pr(\hbox{Diameter}(g)\leq D+2)=1,
\]
where $D=\lceil\frac{\log n}{\log np}\rceil.$
\end{lemma}
We remark that in \cite{DraMas_10}, the condition is $\log n<<(n-1)p<<\sqrt{n}.$ We explicitly calculated the lower bound according to the proof in \cite{DraMas_10}. For the sake of completeness, we present the proof in Appendix \ref{sec:proofImpossibility}.

Based on Lemma \ref{thm:diameterER}, we obtain the following impossibility result.
\begin{thm}\label{thm:impossibility}
If $24\log n<q\mu <<\sqrt{n}$ and $q>0$ is a constant,
\[
\lim_{n\rightarrow \infty}\Pr(\infObs=\nodes(g))=1
\]
when the observation time
\begin{align}
t\geq\left\lceil\frac{\log n}{\log \mu+\log q}\right\rceil+2 \triangleq t_u.\label{eqn:tRange}
\end{align}
In other words the entire network is infected. In such a case, asymptotically, the probability of any node being the source is $1/n.$
\end{thm}

The process to generate the ER random graph and the process of the information diffusion under the IC model can be viewed as a combined process. In this combined process, an edge exits only when the edge exists in the ER random graph and is live in the IC model. The detailed definition of the live edge could be found in Appendix \ref{sec:TreeJordanIsMAP}. Loosely speaking, an edge $(u,v)$ is said to be live if node $v$ is infected by node $u$ under the IC model. When the observation time is larger than or equal to the diameter of the coupled ER random graph, all nodes in the network are infected. In such a case, the probability of a node being the source is $1/n$ as the source was uniformly chosen. Based on Lemma \ref{thm:diameterER}, the diameter of the combine network is smaller than $\lceil\frac{\log n}{\log q+\log \mu}\rceil+2$ w.p.1 asymptotically.

{\bf Remark 1:} We compare $t_u$ in Equation (\ref{eqn:tRange}) and the upper bound in Theorem \ref{thm:sourceIsJordanER}. Since $q$ is a constant, the ratio between $t_u$ and the upper bound becomes $\frac{1}{1+\alpha}$ asymptotically. Since $\alpha$ can be arbitrarily close to $\frac{1}{2},$ the ratio becomes $\frac{2}{3}.$ Therefore, the Jordan infection center is the actual source when the observation time is in the range of $(0,\frac{2}{3}t_u)$ and it is impossible to locate the source when the observation time is $(t_u,\infty).$

{\bf Remark 2:} We compare $t_u$ and the upper bound in Theorem \ref{thm:approximateRatio} and asymptotically the ratio between $t_u$ and the upper bound becomes $\frac{1}{1+\alpha}$ where $\alpha \in (0,1).$ Since $\alpha$ can be arbitrarily close to $0$ and the ratio are close to $1$ which means the BFS tree from  the source has large BND before it becomes impossible to locate the source. While the theorem does not provide any guarantee on the detection rate, it justifies the tie-breaking using BND and WBND.

\section{Performance Evaluation}\label{sec:performance-evaluation}

In this section, we compare the proposed algorithms with existing algorithms on different networks such as tree networks, the ER random graphs and real world networks.

\subsection{Algorithms}
Among all the existing algorithms discussed in Section \ref{sec:introduction}, we choose the algorithms which require only a single snapshot of the network but not the infection probabilities which could be difficult to obtain in practice. We compared SFT and wSFT with the algorithms summarized as follows.
\begin{figure*}
        \centering
          \subfloat[Detection rate\label{fig:binomialDetectionRate}]{%
                           \includegraphics[width=0.3\textwidth]{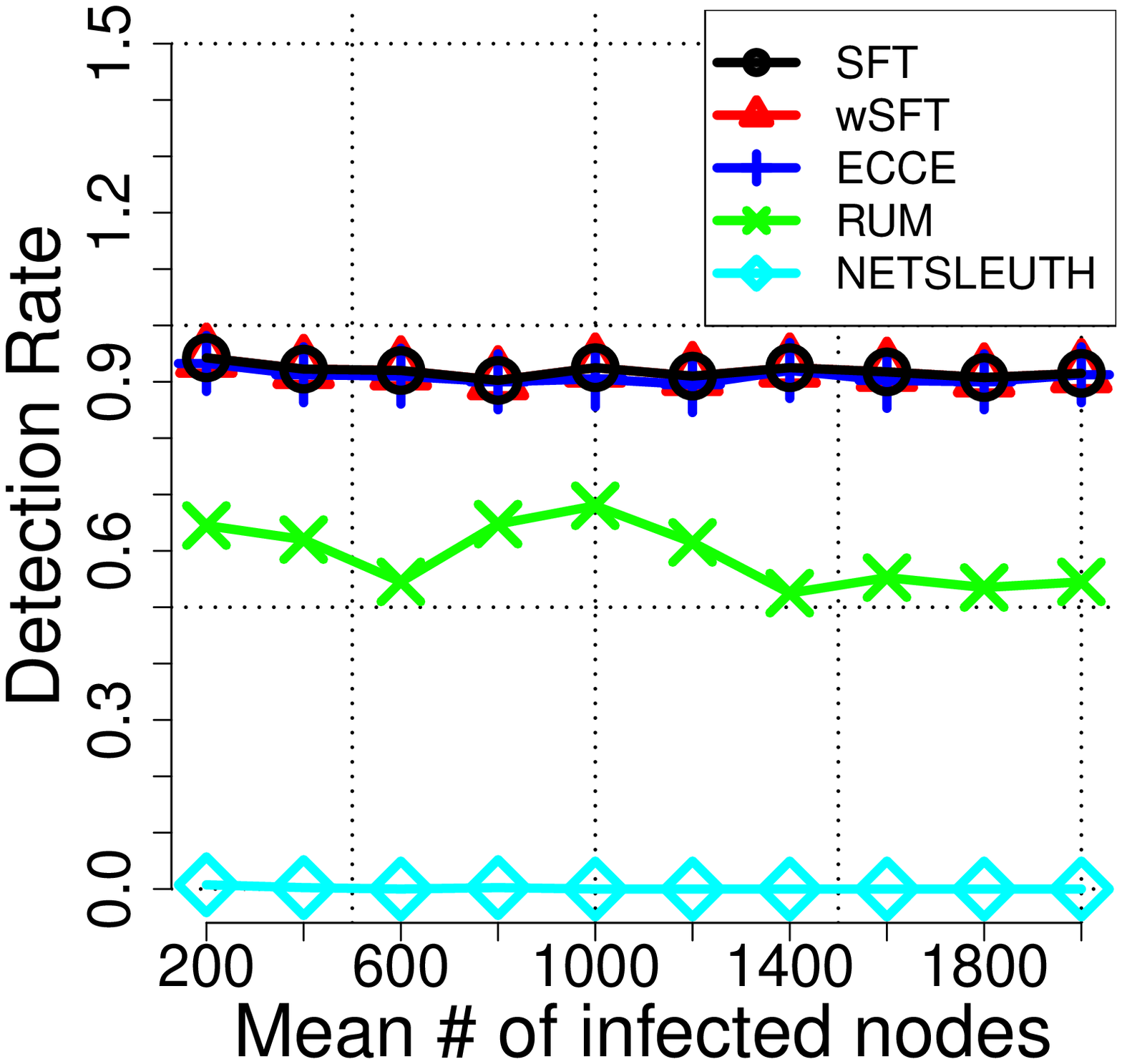} }
~
          \subfloat[Distance to the source\label{fig:binomialDist}]{%
                           \includegraphics[width=0.3\textwidth]{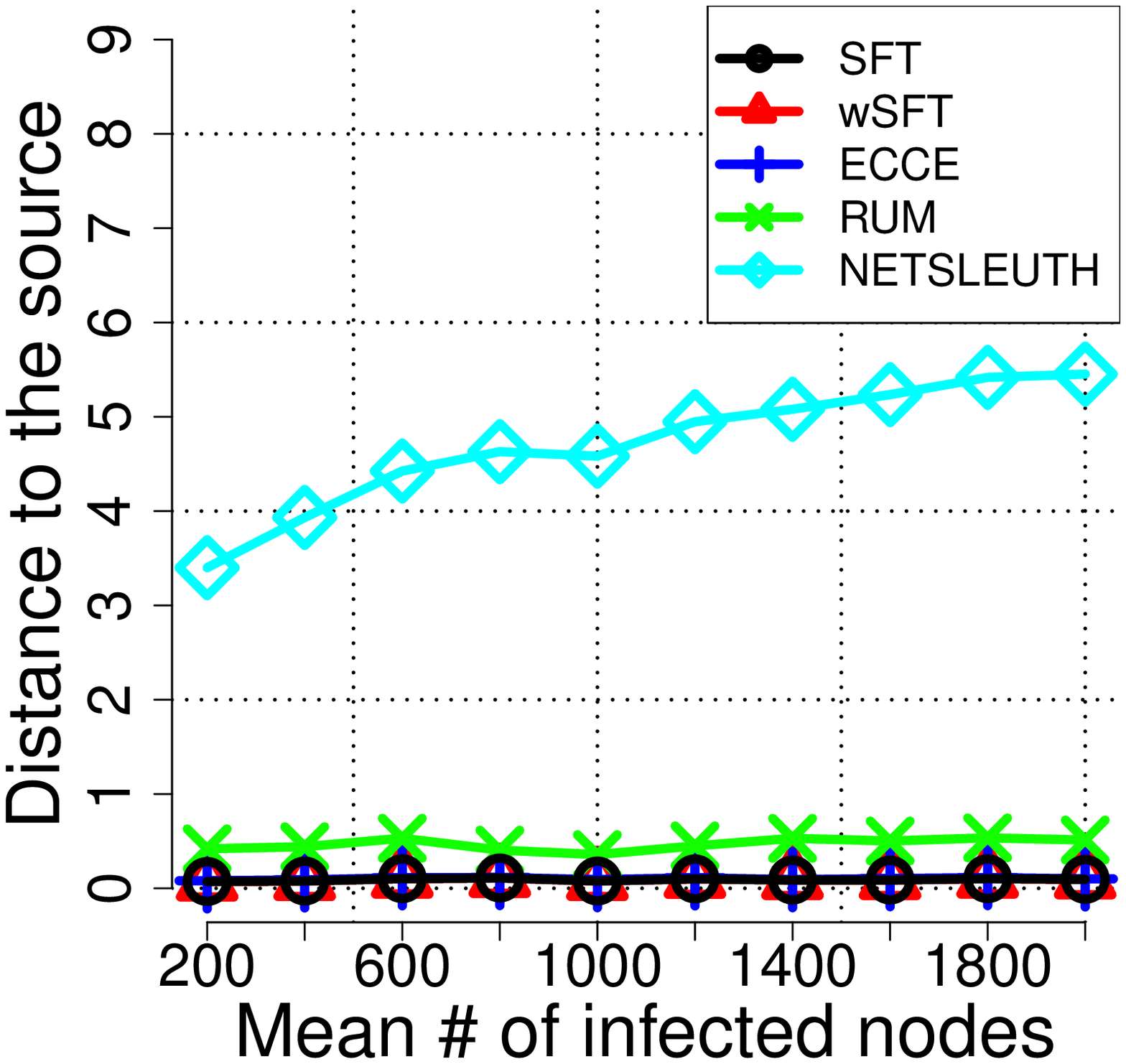}}
~        \subfloat[$\gamma\%$-accuracy\label{fig:binomialcdf}]{%
                         \includegraphics[width=0.3\textwidth]{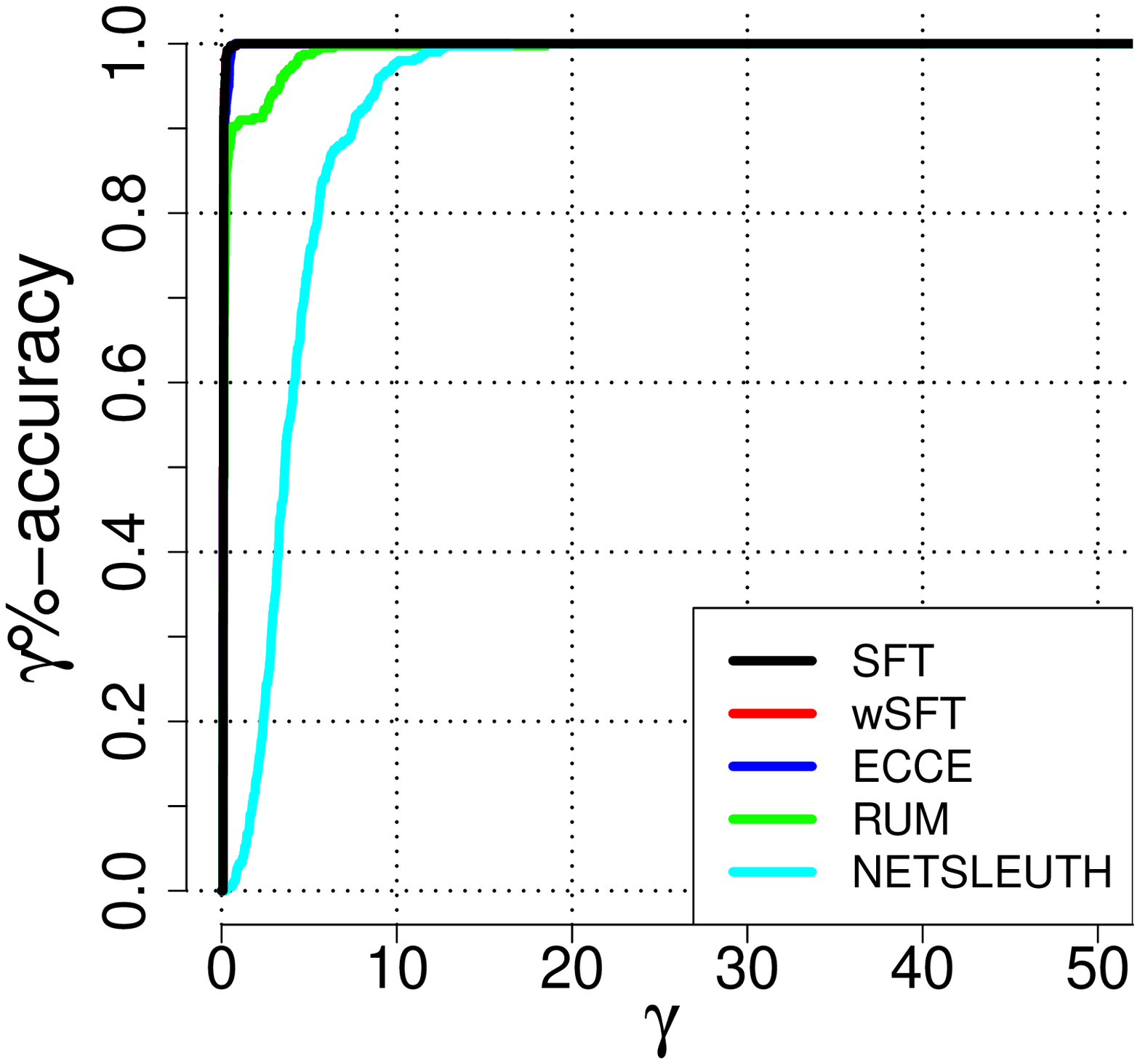}}
        \centering
        \caption{Performance in the binomial trees}\label{fig:binomial}
\end{figure*}

\begin{itemize}
\item {\bf ECCE:} Select the node with minimum infection eccentricity. Ties are breaking randomly. Recall the definition of the infection eccentricity is the maximum distance from the node to all infected nodes. \cite{ZhuYin_14_2} showed that the optimal sample path estimator on tree networks is the Jordan infection center of the graph under the SIR model.
\item {\bf RUM:} Select the node with maximum rumor centrality proposed in \cite{ShaZam_11}. The rumor centrality was proved to be the maximum likelihood estimator on regular trees under the continuous time SI model in which the infection time follows exponential distribution.
\item {\bf NETSLEUTH:} Select the node with maximum value in the eigenvector corresponding to the largest eigenvalue of a submatrix which is constructed from the infected nodes based on the graph Laplacian matrix. The algorithm was proposed in \cite{PraVreFal_12}.
\end{itemize}
Among the selected algorithms, only wSFT requires the infection probabilities. We included wSFT to evaluate the importance of the knowledge of edge weights to our algorithm.  We will see that the performance of SFT is almost identical to wSFT, so the infection probabilities are not important for our detection algorithm.

\subsection{Evaluation Metrics}
We evaluated the performance of the algorithms with three different metrics.
\begin{itemize}
\item Detection rate is the probability that the node identified by the algorithm is the actual source. A desired goal of the information source detection is to have a high detection rate.

\item Distance is the number of hops from the source estimator to the actual source. The distance is an often used metric for information source detection.
\item $\gamma\%$-accuracy is the probability with which the source is ranked among top $\gamma$ percent. Note that besides providing a source estimator, an information source algorithm can also be used to rank the infected nodes according to their likelihood to be the source. For example, SFT can rank the nodes in an ascendant order according to their infection eccentricity and then breaks ties using BND. Other algorithms can be used to rank nodes as well. $\gamma\%$-accuracy is a less ambitious alternation to the detection rate. When the detection rates of all algorithms are low, it is useful to compare $\gamma\%$-accuracy as a high $\gamma\%$-accuracy guarantee that the actual source is among the top ranked nodes with a high probability.
\end{itemize}

\begin{figure*}
        \centering
          \subfloat[Detection rate\label{fig:erDetectionRate}]{%
                           \includegraphics[width=0.3\textwidth]{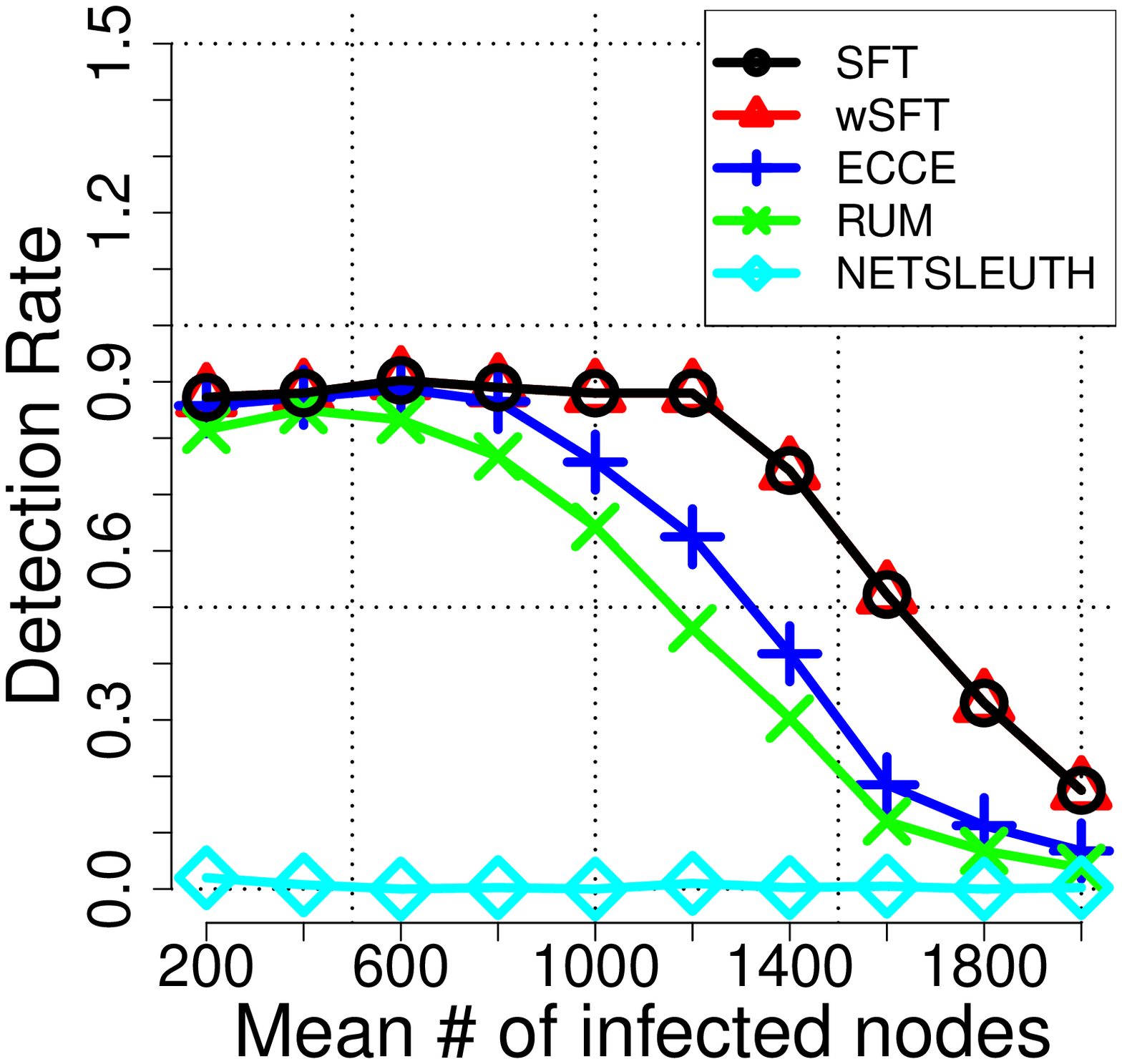} }
~
          \subfloat[Distance to the source\label{fig:erDist}]{%
                           \includegraphics[width=0.3\textwidth]{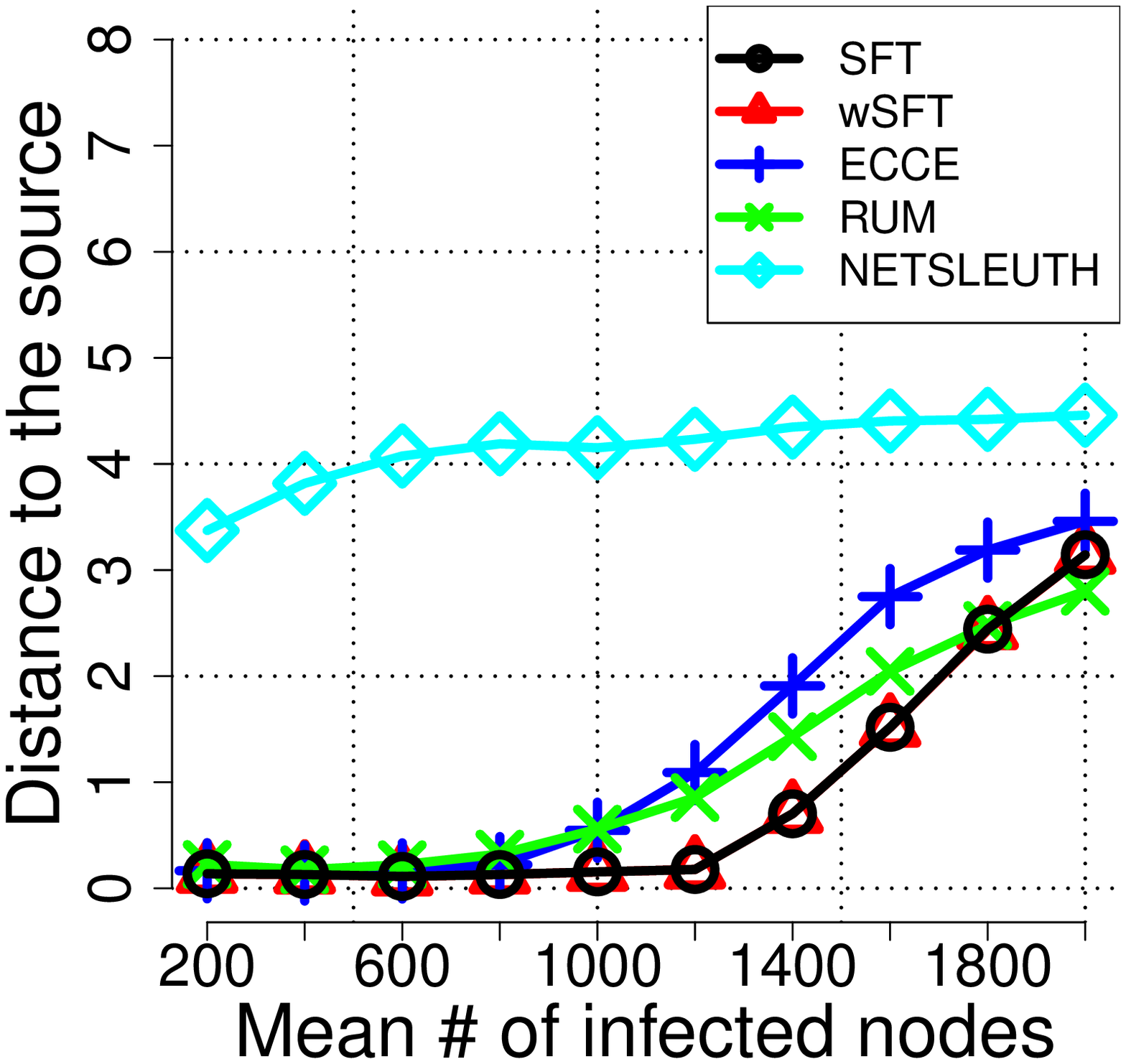}}
~        \subfloat[$\gamma\%$-accuracy\label{fig:ercdf1000}]{%
                         \includegraphics[width=0.3\textwidth]{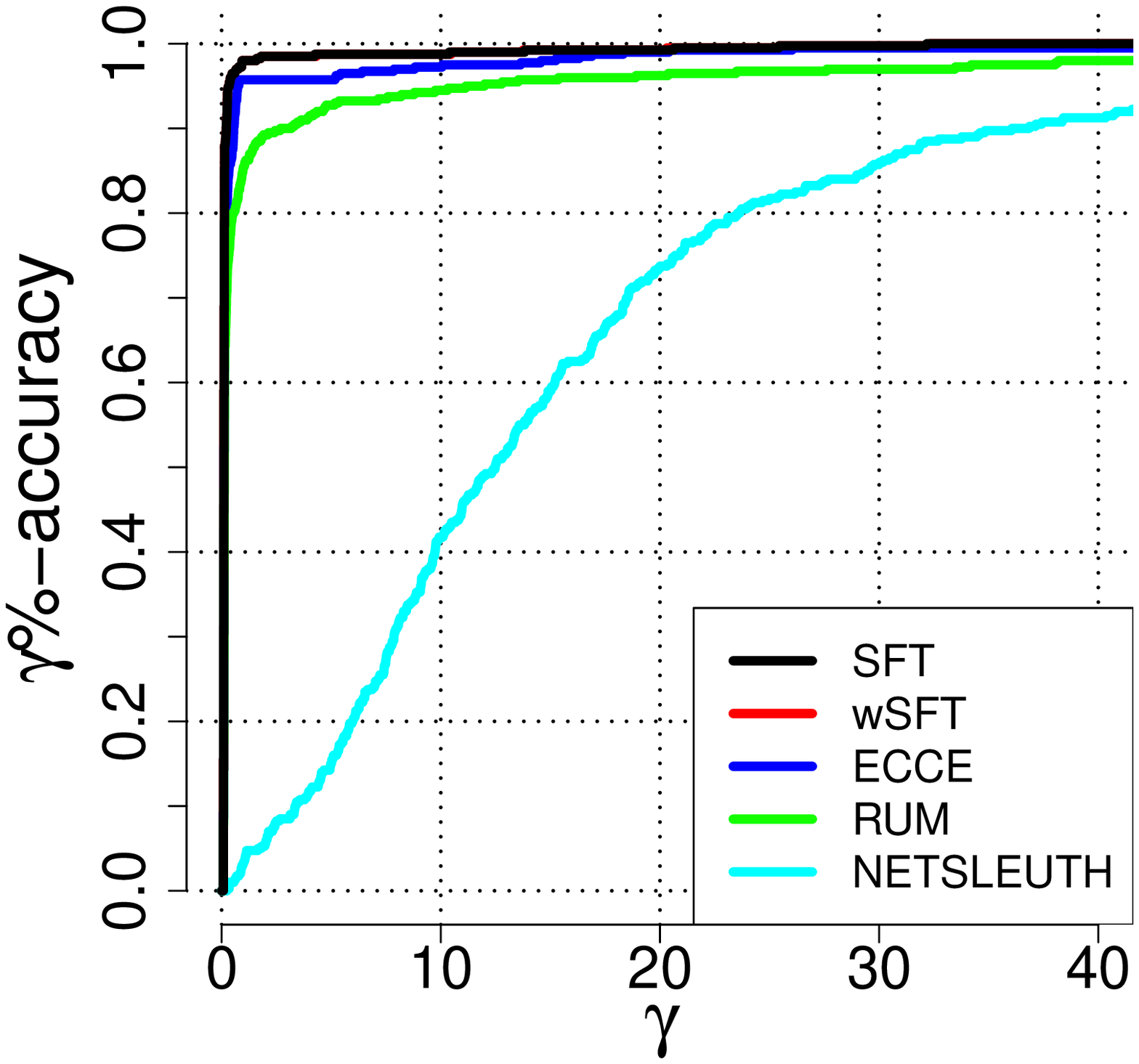}}
        \centering
        \caption{Performance in the ER random graph}\label{fig:er}
\end{figure*}

\subsection{Binomial Trees}
In this section, we evaluate the algorithms on binomial trees. Denote by $\hbox{Bi}(m,\beta)$ the binomial distribution with $m$ number of trials and each trial succeeds with probability $\beta.$ A binomial tree is a tree where the number of children of each node follows a binomial distribution $\hbox{Bi}(m,\beta).$ In the experiments, we set $m=20$ and $\beta=0.5.$ We adopted the IC model where the infection probability of each edge is assigned with a uniform distribution in $(0.2,0.5).$ The lower bound on the infection probability is set to be $0.2$ to prevent the diffusion process dies out quickly. We evaluated the performance for different infection size $x.$ Under a discrete infection model, it is hard to obtain the diffusion snapshots with exact $x$ infected nodes. Therefore, for each infection size $x,$ we generate the diffusion samples where the number of infected nodes are in range $[0.75x, 1.25x].$ The source was chosen uniformly at random among all nodes in the network. We varied $x$ from $200$ to $2000$ with a step size $200.$ For each infection size, we generate $400$ diffusion samples.

 Figure \ref{fig:binomialDetectionRate} shows the detection rates for different infection sizes. The detection rates of ECCE, SFT and wSFT do not change for different infection sizes since the structure of the binomial tree is simple. SFT, wSFT and ECCE have the highest detection rate (more than 0.9) while the detection rate of RUM and NETSLEUTH are much lower.

The distance results are shown in Figure \ref{fig:binomialDist}. As expected, SFT, wSFT and ECCE outperform RUM, which are all much better than NETSLEUTH.

Figure \ref{fig:binomialcdf} shows the $\gamma\%$-accuracy versus the rank percentage $\gamma.$ We picked infection size 1,000. As shown in Figure \ref{fig:binomialcdf}, all three algorithms based on infection eccentricity (ECCE, SFT, wSFT) have better performance than RUM and NETSLEUTH. Recall that the node identified by wSFT is a MAP estimator of the actual source.

\subsection{The ER random graph}\label{subsec:ERsimulation}
In this section, we compared the performance of the algorithms on the ER random graph. In the experiments, we generated the ER random graph with $n=5,000$ and wiring probability $p = 0.002.$ We again varied the infection network size from $200$ to $2,000$. The infection probability of each edge is assigned with a uniform distribution in $(0.2,0.5).$ We generated $400$ diffusion samples.

 Figure \ref{fig:erDetectionRate} shows the detection rate versus the infection size. The detection rate decreases as the infection size increases. SFT and wSFT have higher detection rates compared to other algorithms. Figure \ref{fig:erDist} shows the results on distance. As we expected, SFT and wSFT outperform other algorithms when the infection size is less than 1,600 nodes. As the size of the infected nodes increase, SFT and wSFT become close to RUM in term of distance to the source. However, the detection rate of both algorithms are still much higher than that of RUM.  Another observation is that SFT and wSFT have identical performance which indicates that the performance of SFT is robust to edge weights.

 Figure \ref{fig:ercdf1000} shows the $\gamma\%$-accuracy versus the rank percentage $\gamma$ with 1000 infected nodes. SFT and wSFT have similar or better performance compared to all other algorithms.

 Although the performance of ECCE and SFT algorithms are similar in tree networks, SFT outperforms ECCE significantly on the ER random graphs. The observation indicates that BND is an effective tie breaking rule and increases the detection accuracy.

 \begin{figure*}
         \centering
           \subfloat[Detection rate\label{fig:iasDetectionRate}]{%
                            \includegraphics[width=0.3\textwidth]{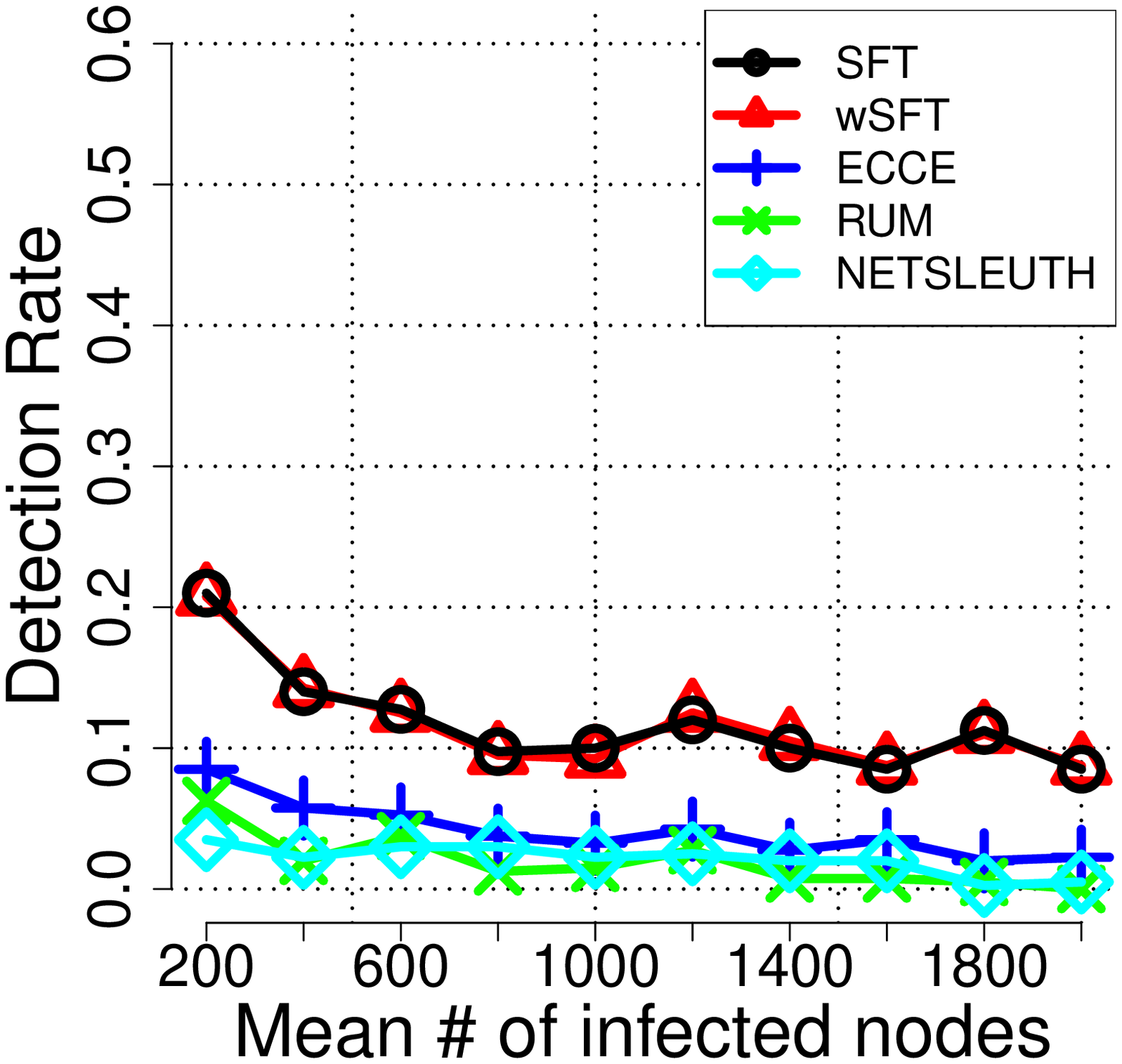} }
 ~		\subfloat[Distance to the source\label{fig:iasDist}]{%
                            \includegraphics[width=0.3\textwidth]{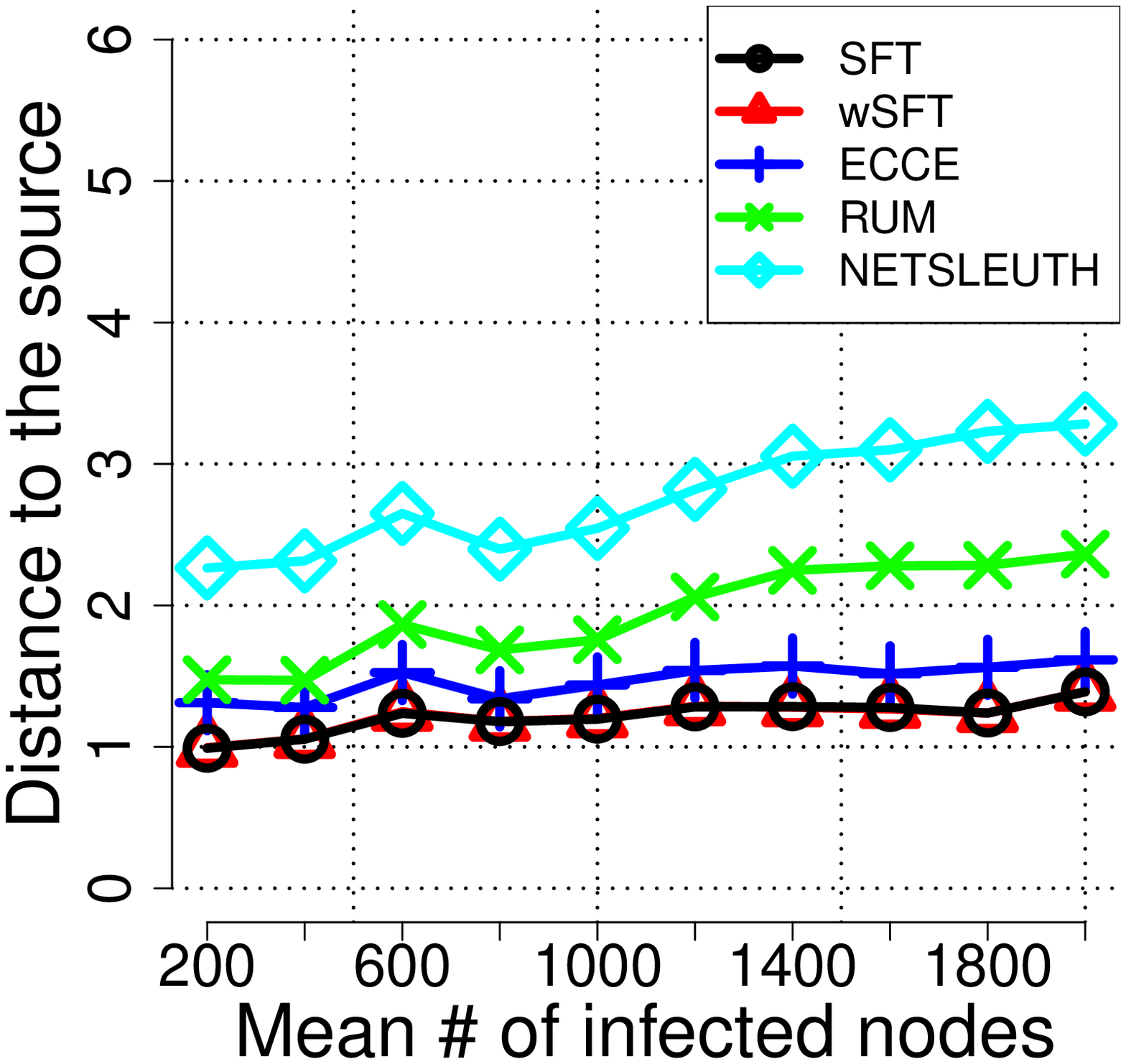}}
 ~			
 		\subfloat[$\gamma\%$-accuracy\label{fig:iascdf1000}]{%
                          \includegraphics[width=0.3\textwidth]{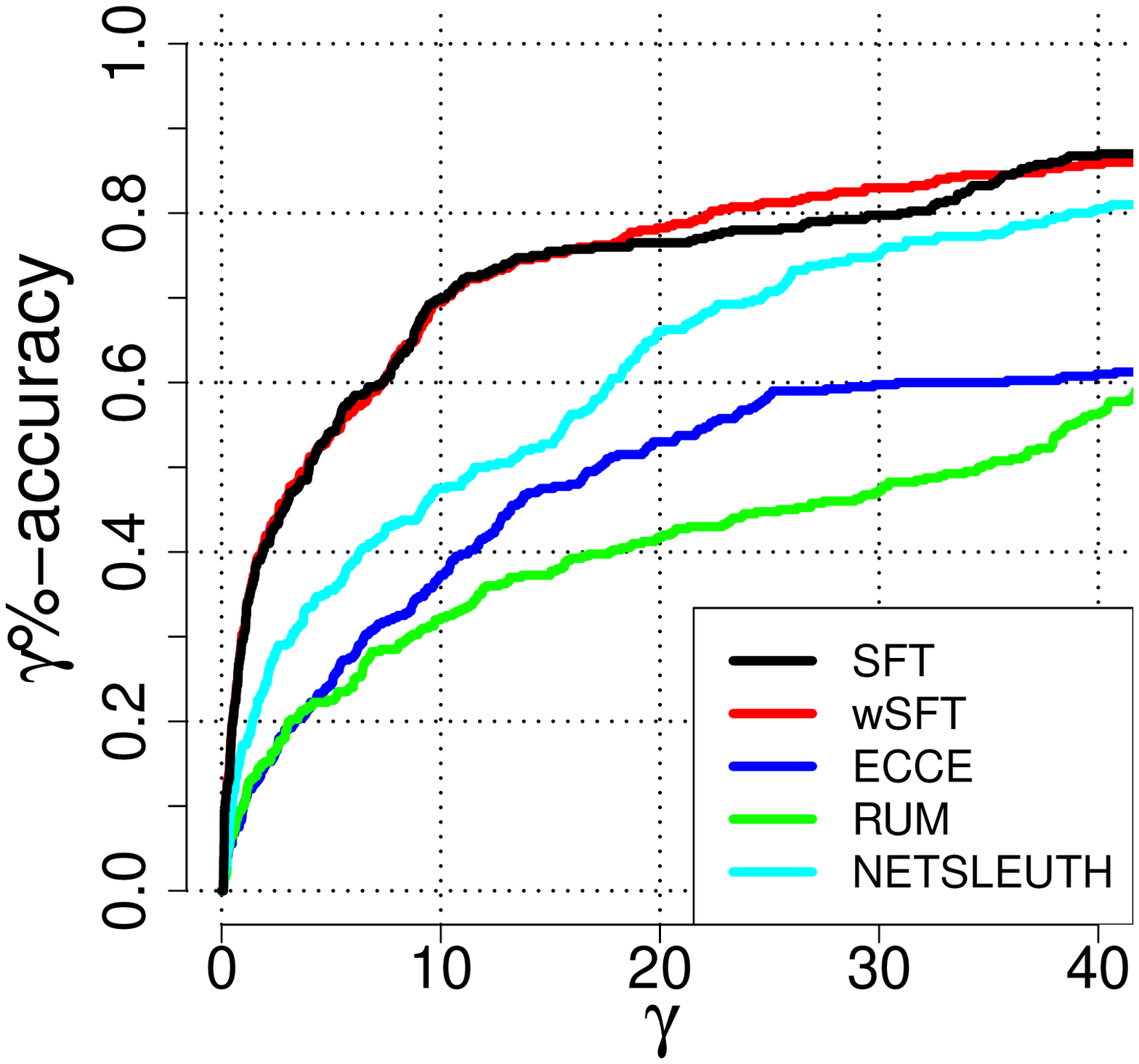}}
         \centering
         \caption{Performance in the IAS graph}\label{fig:ias}
 \end{figure*}
\subsection{The Internet Autonomous System Network }

The Internet autonomous systems (IAS) network \footnote{Available at
\url{http://snap.stanford.edu/data/index.html}} is the Internet autonomous system from Oregon route-views on March, 31st, 2001 with 10,670 nodes and 22,002 edges. The IAS network is a small world network. We adopted similar settings as in Section \ref{subsec:ERsimulation}.

The detection rates are shown in Figure \ref{fig:iasDetectionRate}. The detection rate of ECCE is low since the IAS graph is a small world network and there are multiple Jordan infection centers due to the small diameter of the network. With the tie breaking rule BND, the detection rate doubles in most cases which demonstrates the effectiveness of BND. While the detection rate of SFT is only $10\%$ when the infection size is 1,000, the distance to the actual source is slightly more than one-hop away as shown in Figure \ref{fig:iasDist}. In addition, the $\gamma\%$-accuracy versus $\gamma$ for 1,000 infection size is shown in Figure \ref{fig:ercdf1000}. The $10\%$-accuracies of SFT and wSFT are close to $70\%$ which are significantly higher than that of other algorithms.

\subsection{Running Time vs Performance}
In this section, we evaluated the scalability of the algorithms by comparing the running time. The experiments were conducted on an Intel Core i5-3210M CPU with four cores and 8G RAM with a Windows 7 Professional 64 bit system. All algorithms were implemented with python 2.7. The ER random graphs with 5,000 nodes and $p=0.002$ edge generation probability were used in the experiments. The infection probability of each edge is uniformly distributed over $(0.2,0.5).$ We generated 100 diffusion samples for the experiments. Figure \ref{fig:TimeVsDetection} show the average running time versus the detection rate. The infection size is chosen to be 1,000. SFT and wSFT took 1.11 seconds and achieves 0.87 detection rate while NETSLEUTH took 0.62 seconds with 0 detection rate and RUM took 14.86 seconds with 0.7 detection rate. The detection rate of SFT is much higher than NETSLEUTH and SFT is 14 times faster than RUM.

\section{Conclusions}
In this paper, we derived the MAP estimator of the information source on tree networks under the IC model. Based on that, the SFT algorithm has been proposed. We proved that the SFT algorithm identifies the information source with probability one asymptotically in the ER random graph when the observation time $t\leq \frac{2}{3}t_u,$ which is the first theoretical guarantee on non-tree networks to our best knowledge. We evaluated the performance of SFT on tree networks, the ER random graph and the IAS network.

\section*{Acknowledgement}
This work was supported in part by the U.S. Army Research Laboratory's Army Research Office (ARO Grant No. W911NF1310279).

\begin{figure}
        \centering
		\includegraphics[width=0.35\textwidth]{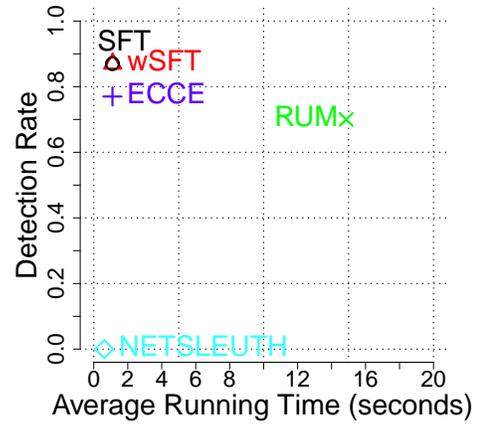}
        \caption{Detection rate versus running time in the ER random graph}\label{fig:TimeVsDetection}
\end{figure}

\bibliographystyle{ieeetr}
\bibliography{ERcitation}

\begin{thebibliography}{10}

\bibitem{John_1854}
J.~Snow, ``The cholera near {G}olden-square, and at {D}eptford,'' {\em Medical
  Times and Gazette}, 1854.

\bibitem{ShaZam_11}
D.~Shah and T.~Zaman, ``Rumors in a network: {W}ho's the culprit?,'' {\em IEEE
  Trans. Inf. Theory}, vol.~57, pp.~5163--5181, Aug. 2011.

\bibitem{ErdRen_59}
P.~Erdos and A.~Renyi, ``On random graphs {I},'' {\em Publ. Math. Debrecen},
  vol.~6, pp.~290--297, 1959.

\bibitem{ShaZam_12}
D.~Shah and T.~Zaman, ``Rumor centrality: a universal source detector,'' in
  {\em Proc. Ann. ACM SIGMETRICS Conf.}, (London, England, UK), pp.~199--210,
  2012.

\bibitem{LuoTayLen_13}
W.~Luo, W.~P. Tay, and M.~Leng, ``Identifying infection sources and regions in
  large networks,'' {\em IEEE Trans. Signal Process.}, vol.~61, pp.~2850--2865,
  2013.

\bibitem{KarFra_13}
N.~Karamchandani and M.~Franceschetti, ``Rumor source detection under
  probabilistic sampling,'' in {\em Proc. IEEE Int. Symp. Information Theory
  (ISIT)}, (Istanbul, Turkey), July 2013.

\bibitem{DonZhaTan_13}
W.~Dong, W.~Zhang, and C.~W. Tan, ``Rooting out the rumor culprit from
  suspects,'' in {\em Proc. IEEE Int. Symp. Information Theory (ISIT)},
  (Istanbul, Turkey), pp.~2671--2675, 2013.

\bibitem{WanDonZha_14}
Z.~Wang, W.~Dong, W.~Zhang, and C.~W. Tan, ``Rumor source detection with
  multiple observations: fundamental limits and algorithms,'' in {\em Proc.
  Ann. ACM SIGMETRICS Conf.}, (Austin, TX), 2014.

\bibitem{ZhuYin_14_2}
K.~Zhu and L.~Ying, ``Information source detection in the {SIR} model: {A}
  sample path based approach,'' {\em IEEE/ACM Trans. Netw.}, Nov. 2014.
\newblock DOI: 10.1109/TNET.2014.2364972.

\bibitem{ZhuYin_14}
K.~Zhu and L.~Ying, ``A robust information source estimator with sparse
  observations,'' in {\em Proc. IEEE Int. Conf. Computer Communications
  (INFOCOM)}, (Toronto, Canada), April-May 2014.

\bibitem{CheZhuYin_14}
Z.~Chen, K.~Zhu, and L.~Ying, ``Detecting multiple information sources in
  networks under the {SIR} model,'' in {\em Proc. IEEE Conf. Information
  Sciences and Systems (CISS)}, (Princeton, NJ), 2014.

\bibitem{LuoTay_13_2}
W.~Luo and W.~P. Tay, ``Estimating infection sources in a network with
  incomplete observations,'' in {\em Proc. IEEE Global Conference on Signal and
  Information Processing (GlobalSIP)}, (Austin, TX), pp.~301--304, 2013.

\bibitem{LuoTay_13}
W.~Luo and W.~P. Tay, ``Finding an infection source under the {SIS} model,'' in
  {\em Proc. IEEE Int. Conf. Acoustics, Speech, and Signal Processing
  (ICASSP)}, (Vancouver, BC), May 2013.

\bibitem{LapTerGun_10}
T.~Lappas, E.~Terzi, D.~Gunopulos, and H.~Mannila, ``Finding effectors in
  social networks,'' in {\em Proc. Ann. ACM SIGKDD Conf. Knowledge Discovery
  and Data Mining (KDD)}, pp.~1059--1068, 2010.

\bibitem{GolLibMul_01}
J.~Goldenberg, B.~Libai, and E.~Muller, ``Talk of the network: A complex
  systems look at the underlying process of word-of-mouth,'' {\em Marketing
  Letters}, vol.~12, no.~3, pp.~211--223, 2001.

\bibitem{PraVreFal_12}
B.~A. Prakash, J.~Vreeken, and C.~Faloutsos, ``Spotting culprits in epidemics:
  {H}ow many and which ones?,'' in {\em IEEE Int. Conf. Data Mining (ICDM)},
  (Brussels, Belgium), pp.~11--20, 2012.

\bibitem{LokMezOht_14}
A.~Y. Lokhov, M.~M\'ezard, H.~Ohta, and L.~Zdeborov\'a, ``Inferring the origin
  of an epidemic with a dynamic message-passing algorithm,'' {\em Phys. Rev.
  E}, vol.~90, p.~012801, Jul 2014.

\bibitem{PinThiVet_12}
P.~C. Pinto, P.~Thiran, and M.~Vetterli, ``Locating the source of diffusion in
  large-scale networks,'' {\em Phys. Rev. Lett.}, vol.~109, no.~6, p.~068702,
  2012.

\bibitem{ZhuCheYin_15}
K.~Zhu, Z.~Chen, and L.~Ying, ``Locating the contagion source in networks with
  partial timestamps,'' {\em Data Mining and Knowledge Discovery}, 2015.

\bibitem{AgaLu_13}
A.~Agaskar and Y.~M. Lu, ``A fast {M}onte {C}arlo algorithm for source
  localization on graphs,'' in {\em SPIE Optical Engineering and Applications},
  2013.

\bibitem{ZejGomSin_13}
S.~Zejnilovic, J.~Gomes, and B.~Sinopoli, ``Network observability and
  localization of the source of diffusion based on a subset of nodes,'' in {\em
  Proc. Annu. Allerton Conf. Communication, Control and Computing},
  (Monticello, IL), 2013.

\bibitem{Har_91}
F.~Harary, {\em Graph theory}.
\newblock Addison-Wesley, 1991.

\bibitem{DraMas_10}
M.~Draief and L.~Massouli, {\em Epidemics and rumours in complex networks}.
\newblock Cambridge University Press, 2010.

\bibitem{KemKleTar_03}
D.~Kempe, J.~Kleinberg, and E.~Tardos, ``Maximizing the spread of influence
  through a social network,'' in {\em Proc. Ann. ACM SIGKDD Conf. Knowledge
  Discovery and Data Mining (KDD)}, (Washington DC), pp.~137--146, 2003.

\bibitem{MitUpf_05}
M.~Mitzenmacher and E.~Upfal, {\em Probability and Computing: Randomized
  Algorithms and Probabilistic Analysis}.
\newblock Cambridge: Cambridge University Press, 2005.

\end{thebibliography}

\appendices

\section{Proof of Theorem \ref{thm:RunningTime}}\label{proof:RunningTime}
\begin{proof}
Follow the argument in \cite{ZhuYin_14_2}, the computational complexity of the node ID broadcasting phase of Algorithm \ref{alg:RI} is $O(|\nodes(g_i)||\edges(g_i)|)$ since the node IDs are passed on the subgraph $g_i.$ The complexity of Algorithm \ref{alg:WBNDC} is $O(\degree(\infObs))$ since the number of boundary nodes is bounded by $|{\cal I}|.$ In addition, Algorithms \ref{alg:WBNDC} are called at most $|{\cal S}|$ times in Algorithm \ref{alg:RI} and $|{\cal S}|\leq |\nodes(g_i)|.$ Therefore, the complexity of Algorithm \ref{alg:RI} is
\[
O(|\nodes(g_i)||\edges(g_i)|+|\nodes(g_i)|\degree(\infObs))
\]
Note $\nodes(g_i) = \infObs$ and $|\edges(g_i)|\leq \degree(\infObs).$ The complexity becomes
\[
O(|\infObs|\degree(\infObs)).
\]

\end{proof}
\section{Proof of Theorem \ref{thm:treepartialMAP}}\label{sec:TreeJordanIsMAP}

First, we prove the following lemma for neighboring nodes.

\begin{lemma}\label{cor:MAPneighbors}
{\bf Neighboring nodes inequality}
Consider nodes $u,v$ on tree $\tilde{\tree}$ satisfying the following conditions:
\begin{itemize}
\item  $(u,v)\in\edges(\tilde{\tree}).$
\item The observation time follows a distribution such that $\Pr(\obsTime)\geq \Pr(\obsTime+1)$ for all $\obsTime.$
\item The source is uniformly chosen among all nodes, i.e., $\Pr(u)=\Pr(v).$
\item $\ecce(v,\infObs)>\ecce(u,\infObs).$
\end{itemize}
We have
\[
\Pr(v|\obs)\leq \Pr(u|\obs)
\]
\end{lemma}
\begin{proof}
Consider nodes $u,v$ on a tree $\tilde{\tree}$ where $(u,v)\in\edges(\tilde{\tree})$. Let $\obsTime_u,\obsTime_v$ be the observation times associated with $u,v.$ We will show that when $\ecce(v,\infObs)>\ecce(u,\infObs),$
\begin{align}
\Pr(\obs|v,\obsTime_v = t+1)\leq \Pr(\obs|u,\obsTime_u = t),\label{eqn:neighborInequality}
\end{align}
where $\Pr(\obs|v,\obsTime_v = t)$ is the probability of the snapshot $\obs$ given $v$ is the source and the observation time is at time slot $t.$ .

We adopt an equivalent view of the IC model called \emph{live edge} model \cite{KemKleTar_03}. In the IC model, after $u$ becomes infected, it attempts to infect its neighbor $w$ with probability $q_{uw}$ once.  Therefore, we can assume that a biased coin with parameter $q_{uw}$ is flipped for edge $(u,w)\in \edges(g)$ when $u$ tries to infect $w$ in the IC model. Note that the probability of node $w$ is infected by node $u$ remains the same whether the coin is flipped at the moment when node $u$ attempts to infect $w$ or prior to the infection but is revealed for the attempt. Assume the coins of all edges are flipped at the beginning of the infection process. When one node attempts to infect one of its neighbors, we check the stored coin realization to determine whether the infection succeeds. This process is called live edge model and it is equivalent to the IC model since we only change the time of the coin flippings, so the probability of an infection trace remains the same. In the live edge model, the infection process consists of two steps. First, each edge $(u,w)$ flips a biased coin with probability $q_{uw}$ to be a \emph{live edge} prior to the infection starts. After all coin flippings, the graph formed by live edges is called the \emph{live edge graph}. In the second step, the infection spreads over all live edges deterministically, starting from the source. We now analyze SFT under the live-edge model.

Denote by $\treeSet$ the set of all live edge graphs of $\tilde{\tree}$, i.e.,
\[
\treeSet=\{\tree|\edges(\tree)\subset\edges(\tilde{\tree}),\nodes(\tree)=\nodes(\tilde{\tree})\}
\]
Note there are no loops in $\tree \in \treeSet$ since $\tree$ is a subgraph of tree $\tilde{\tree}.$

 Denote by ${\cal K}(\obs,v,\obsTime_v)$ the set of all live edge graphs on which the observation $\obs$ is feasible if the source is $v$ and the observation time is $\obsTime_v.$ All infected nodes must be within $t_v$ hops from the source and all the observed healthy nodes must be more than $t_v$ hops away from the source in a feasible live edge graph. Formally, we have
\[
{\cal K}(\obs,v,\obsTime_v)=\{\tree\in \treeSet| \forall w \in \infObs, \dist^\tree_{vw}\leq \obsTime_v, \forall w \in \healthObs, \dist^\tree_{vw}>\obsTime_v\}.
\]

$\Pr(\obs| v, \obsTime_v)$ equals to the probability a live edge graph is in set ${\cal K}(\obs,v,\obsTime_v)$ due to the equivalence between the IC model and the live-edge model. The probability of a specific live edge graph is the product of edge live/dead probabilities.

Hence we have
\[
\Pr(\obs|v,\obsTime_v)=\sum_{\tree\in {\cal K}(\obs,v,\obsTime_v)}\Pr(\tree),
\]
To prove the lemma, we will prove the following claim,
\[
{\cal K}(\obs,v,\obsTime_v = t+1)\subset{\cal K}(\obs,u,\obsTime_u = t).
\]
To simplify notation, we next assume $t_u=t_v-1=t\geq e(u, {\cal I}),$ and ignore $t$ in the equations. Note that we only consider $t_u\geq e(u, {\cal I})$ because $\Pr({\cal O}|u, t_u) =0$ otherwise.

Consider $\tree\in {\cal K}(\obs,v,\obsTime_v).$ Denote by $\tree^{-u}_v$ the tree rooted at $v$ without the subtree starting from $u$ where $(u,v)\in\edges(\tree).$ Since $\ecce(v,\infObs)>\ecce(u,\infObs),$ according to Lemma 2 in \cite{ZhuYin_14_2}, there exists $w^\dag\in\infObs\cap\tilde{\tree}^{-v}_{u}$ which has $\dist^{\tilde{\tree}}_{vw^\dag}=\dist^{\tilde{\tree}}_{uw^\dag}+1=\ecce(v,\infObs).$ If $(v,u)\not \in \edges(\tree),$ we have $\dist^{T}_{vw^\dag}=\infty>t_v$ which is a contradiction to $\tree\in {\cal K}(\obs,v,\obsTime_v).$ Hence, we have $(v,u)\in \edges(\tree).$
\begin{itemize}
\item Consider $\tilde{\tree}^{-v}_u$ part.
\begin{itemize}
\item For any $w \in \tilde{\tree}^{-v}_u \cap {\cal I},$ we have
\[
\dist^{\tree}_{vw}\leq \ecce(v, {\cal I})\leq  \obsTime_v.
\]
Since $(v,u)\in\edges(g)$ and $\tree$ have no loops, we have
\[
\dist^{\tree}_{uw}=\dist^{\tree}_{vw}-1\leq \obsTime_v-1=\obsTime_u.
\]

\item For any $w \in \tilde{\tree}^{-v}_u \cap \healthObs,$ we have
\[
\dist^{\tree}_{vw} > e(v, {\cal I})\geq  \obsTime_v.
\]
Since $(v,u)\in\edges(g)$ and $\tree$ has no loops, we have
\[
\dist^{\tree}_{uw}=\dist^{\tree}_{vw}-1> \obsTime_v-1=\obsTime_u.
\]
\end{itemize}

\item Consider $\tilde{\tree}^{-u}_v$ part.
\begin{itemize}
\item For any $w \in \tilde{\tree}^{-u}_v \cap {\cal I},$ based on the proof of Lemma 2 in \cite{ZhuYin_14_2}, we have
\[
\dist^{\tilde{\tree}}_{vw}\leq \ecce(u,\infObs)-1.
\]
There is only one path ${\cal P}_{vw}$ from $v$ to $w$ in tree $\tilde{\tree}.$ If ${\cal P}_{vw}\not\subset \edges(\tree),$ $v,w$ are disconnected in $\tree$ which contradicts the fact that $\dist^\tree_{vw}\leq \obsTime_v.$ Hence ${\cal P}_{vw}\subset \edges(\tree).$ In addition, we have $(v,u)\subset \edges(\tree).$
Hence
\[\dist^{\tree}_{uw}=\dist^{\tree}_{vw}+1=\dist^{\tilde{\tree}}_{vw}+1\leq \ecce(u,\infObs)\leq \obsTime_u.
\]

\item For any $w \in \tilde{\tree}^{-u}_v \cap {\cal H},$
there is only one path ${\cal P}_{vw}$ from $v$ to $w$ in tree $\tilde{\tree}.$ If ${\cal P}_{vw}\not\subset \edges(\tree),$ we have
\[
\dist^\tree_{uw}=\infty>t_u.
\]
If ${\cal P}_{vw}\subset \edges(\tree),$
\[
\dist^{\tree}_{uw}=\dist^{\tilde{\tree}}_{uw}=\dist^{\tilde{\tree}}_{vw}+1=\dist^{\tree}_{vw}+1> \obsTime_v+1 > \obsTime_u
\]
where $d_{vw}^\tree >t_v$ because $\tree\in {\cal K}({\cal O}, v, t_v).$
\end{itemize}
\end{itemize}
As a summary, for any $\tree \in {\cal K}(\obs,v,\obsTime_v),$ we have
\[
\forall w \in \infObs, \dist^\tree_{uw}\leq \obsTime_u, \forall w \in \healthObs, \dist^\tree_{uw}>\obsTime_u
\]
Therefore, $\tree\in {\cal K}(\obs,u,\obsTime_u).$

Hence, we proved
\[
{\cal K}(\obs,v,\obsTime_v = t+1)\subset{\cal K}(\obs,u,\obsTime_u = t)
\]
which implies
\begin{align*}
\Pr(\obs|v,\obsTime_v = t+1)= & \sum_{\tree\in {\cal K}(\obs,v,\obsTime_v = t+1)}\Pr(\tree)\\
\leq & \sum_{\tree\in {\cal K}(\obs,u,\obsTime_u = t)}\Pr(\tree)\\
= & \Pr(\obs|u,\obsTime_u = t).
\end{align*}
Hence, we proved Inequality (\ref{eqn:neighborInequality}).

Denote by $\Pr(v)$ the probability that $v$ is the source, $\Pr(\obsTime)$ the probability that the observation time is $\obsTime,$ and $\Pr(\obs|v,\obsTime)$ is the probability of snapshot $\obs$ given $v$ is the source and $\obsTime$ is the observation time. Since the observation time $\obsTime$ is independent of the source node, we obtain
\begin{align*}
\Pr(v|\obs)&=\frac{1}{\Pr(\obs)}\Pr(v,\obs)\\
&=\frac{1}{\Pr(\obs)}\sum_{\obsTime\geq \ecce(v,\infObs)} \Pr(v,\obsTime,\obs)\\
&=\frac{1}{\Pr(\obs)}\sum_{\obsTime\geq \ecce(v,\infObs)}\Pr(\obs|v,\obsTime)\Pr(v,\obsTime)\\
&=\frac{\Pr(v)}{\Pr(\obs)} \sum_{\obsTime\geq \ecce(v,\infObs)}\Pr(\obs|v,\obsTime)\Pr(\obsTime)\\
&\leq_{\hbox{(a)}} \frac{\Pr(v)}{\Pr(\obs)} \sum_{\obsTime\geq \ecce(v,\infObs)}\Pr(\obs|u,\obsTime-1)\Pr(\obsTime)\\
&=_{\hbox{(b)}}\frac{\Pr(u)}{\Pr(\obs)} \sum_{\obsTime\geq \ecce(u,\infObs)}\Pr(\obs|u,\obsTime)\Pr(\obsTime+1)\\
&\leq_{\hbox{(c)}} \frac{\Pr(u)}{\Pr(\obs)} \sum_{\obsTime\geq \ecce(u,\infObs)}\Pr(\obs|u,\obsTime)\Pr(\obsTime) \\
&=\Pr(u|\obs)
\end{align*}
 (a) is due to Inequality (\ref{eqn:neighborInequality}), (b) is based on $\Pr(u)=\Pr(v)$ and $e(v, {\cal I})=e(u, {\cal I})+1,$  and (c) is based on $\Pr(\obsTime)\geq \Pr(\obsTime+1).$
\end{proof}

Based on the proof of Theorem 4 in \cite{ZhuYin_14_2}, there exists a path from any node to a Jordan infection center in the tree network such that the infection eccentricity strictly decreases along the path. By repeatedly applying Lemma \ref{cor:MAPneighbors}, we conclude that a MAP estimator must be a Jordan infection center.

Recall ${\cal J}$ is the set of Jordan infection centers. Next, we will show that the MAP estimator , say node $v,$  has the maximum $\sum_{(v,w)\in{{\cal F}'_u}}|\log(1-q_{vw})|$ among all nodes in ${\cal J}.$

Define an edge set
\[
{\cal F}=\{(v,w)|(v,w)\in\edges(\tilde{\tree}),v\in \infObs,w\in \healthObs\}.
\]
We call the edges in ${\cal F}$ the frontier edges since they are the edges between $\infObs$ and $\healthObs.$

Define another edge set
\[
{\cal B}=\{(v,w)|(v,w)\in \edges(\tilde{\tree}), v,w\in \infObs\}.
\]
The edges in ${\cal B}$ are the edges between infected nodes.

In addition, for any $u\in{\cal J}$ define
\[
{\cal F}_u(\obsTime_u)=\{(v,w)|(v,w)\in\edges(\tilde{\tree}),v\in \infObs,w\in \healthObs,\dist_{uw}\leq \obsTime_u\}.
\]
${\cal F}_u(\obsTime_u)$ is set of edges which cannot be live edges when $u$ is the source and $\obsTime_u$ is the observation time.

For a complete observation, we have
\begin{align}
\Pr(\obs|u,\obsTime_u)=\prod_{(v,w)\in {\cal B}}q_{vw}\prod_{(v,w)\in{\cal F}_u(\obsTime_u)}(1-q_{vw})\label{eqn:AP}
\end{align}

Denote by $e^*$ the minimum infection eccentricity, i.e.,
\[
\forall v \in {\cal J}, \ecce(v,\infObs) = e^*
\]
 Intuitively, when $t_u>e^*,$ none of frontier edges should be a live edge in a feasible live edge graph to make sure healthy nodes are not infected. So when $\obsTime_{u}>e^*,$ we have
\[
{\cal F}_u(\obsTime_u)={\cal F}.
\]
Hence,
\[
\Pr(\obs|u,\obsTime_u)=\prod_{(v,w)\in {\cal B}}q_{vw}\prod_{(v,w)\in{\cal F}}(1-q_{vw})\triangleq C,
\]
which is not a function of either $t_{u}$ or $u.$ Substituting into Equation (\ref{eqn:AP}), we have
\begin{align}
\Pr(\obs|u,e^*)=\frac{C}{\prod_{(v,w)\in{\cal F}\backslash{\cal F}_u(e^*)}(1-q_{vw})}\label{eqn:newAP}
\end{align}

Follow a similar procedure in Lemma \ref{cor:MAPneighbors}, for an Jordan infection center $u,$ we have
\begin{align*}
\Pr(u|\obs)&=\frac{\Pr(u)}{\Pr(\obs)} \sum_{\obsTime}\Pr(\obs|u,\obsTime)\Pr(\obsTime)\\
&=\frac{\Pr(u)}{\Pr(\obs)}\bigg( \Pr(\obs|u,e^*)\Pr(t = e^*)\\
&+\sum_{\obsTime>e^*}\Pr(\obs|u,\obsTime)\Pr(\obsTime)\bigg)\\
&=\frac{\Pr(u)}{\Pr(\obs)}\left( \Pr(\obs|u,e^*)\Pr(t = e^*)+\Pr(\obsTime>e^*)C\right)
\end{align*}
Therefore,
\begin{align}
&\arg\max_{u}\Pr(u|\obs)\\
&=\arg\max_{u\in {\cal J}}\Pr(u|\obs)\\
&=\arg\max_{u\in {\cal J}}\frac{\Pr(u)}{\Pr(\obs)}\big( \Pr(\obs|u,e^*)\Pr(t = e^*)\\
&+\Pr(\obsTime>e^*)C\big)\\
&=\arg\max_{u\in {\cal J}}\Pr(\obs|u,e^*)\label{eqn:MAP}\\
&=\arg \min_{u\in {\cal J}}\prod_{(v,w)\in{\cal F}\backslash{\cal F}_u(e^*)}(1-q_{vw}).\label{eqn:MLE}
\end{align}

Note we have
\begin{align*}
&{\cal F}\backslash{\cal F}_u(e^*)\\
=&\{(v,w)|(v,w)\in\edges(\tilde{\tree}),v\in \infObs,w\in \healthObs,\dist_{vw}>e^*\}
\end{align*}
Since $e^*$ is the minimum eccentricity and $u$ is the Jordan infection center, we have $\dist_{uv}\leq e^*$ for all $v\in \infObs.$ Hence for all $w\in \healthObs$ which have at least one edge to the infected nodes, we have $\dist_{uw}\leq e^*+1.$ Therefore, we have
\begin{align*}
&{\cal F}\backslash{\cal F}_u(e^*)\\
=&\{(v,w)|(v,w)\in\edges(\tilde{\tree}),v\in \infObs,w\in \healthObs,\dist_{vw}=e^*+1\}\\
=&{\cal F}'_u.
\end{align*}
Based on equations \ref{eqn:MAP} and \ref{eqn:MLE}, we conclude
\begin{align*}
&\arg\min_{u\in{\cal J}} \prod_{(v,w)\in{{\cal F}'_u}}(1-q_{vw})\\
=&\arg\max_{u\in{\cal J}} \sum_{(v,w)\in{{\cal F}'_u}}|\log(1-q_{vw})|\\
=&\arg\max_u\Pr(u|\obs).
\end{align*}

{\bf Remark:} Theorem \ref{thm:treepartialMAP} contains two important properties for the MAP estimator on tree networks: 1) the MAP estimator is a Jordan infection center; 2) the Jordan infection center with minimum $\prod_{(v,w)\in{{\cal F}'_u}}(1-q_{vw})$ is the MAP estimator. The short-fat tree algorithm is designed based on these properties, which identifies the Jordan infection centers first and then selects the one with maximum $\sum_{(v,w)\in{{\cal F}'_u}}|\log(1-q_{vw})|.$

\section{Proof of Theorem \ref{thm:sourceIsJordanER}}\label{sec:ERJordanIsSource}

We first introduce and recall some necessary notations. Consider an ER random graph $\graph.$

\begin{itemize}
\item Denote by $\source$ the actual source.
\item A node $v$ is said to locate on level $k$ if $\dist_{\source v}=k.$ Denote by $\levelset_k$ the set of nodes from level $0$ to level $k$ and $\level_k=|\levelset_k|.$
\item The \emph{descendants} of node $v$ in a tree are all the nodes in the subtree rooted at $v$ and $v$ is the \emph{ancestor} of these nodes.
\item The offsprings of a node on level $k$ (say $v$) the set of the nodes which are on level $k+1$ and have edges to $v.$ Denote by $\offspring(v)$ the offspring set of $v$ and $\offspringsize(v)=|\offspring(v)|.$
\item Denote by $p$ the wiring probability in the ER random graph.
\item Denote by $n$ the total number of nodes.
\item Denote by $\mu=np.$
\item Recall that $\hbox{Bi}(n,p)$ is the binomial distribution with $n$ number of trials and each trial succeeds with probability $p.$
\item Denote by $q$ the minimum infection probability of all the edges, i.e., $q=\min_{e\in\edges(\graph)} q_e.$
\end{itemize}
For simplicity, we use $\dist_{vu}=\dist^\graph_{vu}.$

We first elaborate the construction of the BFS tree.  Denote by $v_{ij}$ the $j$th saturated node on level $i$ of the BFS tree  from the source. $v_{01} = \source$ is the first node on level zero. Denote by $b_i$ the number of nodes on level $i$ of the BFS tree starting from the source. We start with an empty graph $\tree^\dag.$ Initially, we add $v_{01}$ to the tree. Starting from $v_{01},$ we explore all neighbors $v_{11},v_{12},\cdots,v_{1b_1}$ of $v_{01},$ mark $v_{01}$ as saturated and add the edges from $v_{01}$ to  $v_{11},v_{12},\cdots,v_{1b_1}$ to $\tree^\dag$. Then we explore all neighbors $v_{21},v_{22},\cdots,v_{2r_1}$ of $v_{11}$ in the set $\nodes(g)\backslash\{v_{01},v_{11},\cdots,v_{1b_1}\},$ mark $v_{11}$ as saturated and add the edges from $v_{11}$ to $v_{21},v_{22},\cdots,v_{2r_1}$ to $\tree^\dag$. Then we explore all neighbors $v_{2r_1+1},v_{2r_1+2},\cdots,v_{2r_1+r_2}$ of $v_{12}$ in the set $\nodes(g)\backslash\{v_{01},v_{11},\cdots,v_{1b_1},v_{21},v_{22},\cdots,v_{2r_1}\}$ and add the corresponding edges. Only after all nodes on level $i$ are saturated, we explore nodes on level $i+1.$ The exploration terminates after all nodes on level $t-1$ are saturated. The resulting tree $\tree^\dag$ is the BFS tree.

We further introduce some notations for the BFS tree.
\begin{itemize}
\item Denote by $\offspring'(v)$ the set of offsprings of node $v$ on $\tree^\dag$ and $\offspringsize'(v)=|\offspring'(v)|.$
\item Denote by $\graph_\obsTime$ the subgraph induced by all nodes within $\obsTime$ hops from $s$ on the ER graph.  The \emph{collision edges} are the edges which are not in $\tree^\dag$ but in $\graph_\obsTime,$ i.e., $e\in \edges(\graph_\obsTime)\backslash\edges(\tree^\dag).$ A node who is an end node of a collision edge is called a \emph{collision node}.  Denote by $\collisionset_k$ the set of collision edges whose end nodes are on level $0$ to level $k$ and $\collisionsize_k=|\collisionset_k|.$
\item Denote by ${\cal Z}^{i}_j(v)$ the set of nodes that are infected at time slot $i,$ on level $j$ and the descendants of node $v$ in the BFS tree $\tree^\dag$. In addition, denote by $Z^i_j(v)=|{\cal Z}^i_j(v)|$. We often use ${\cal Z}^{i}_j={\cal Z}^{i}_j(\source)$ and $Z^i_j=Z^{i}_j(\source)$ for simplicity.
\end{itemize}

We first define the probability space of the problem. Define the sample space $\Omega$ to be the set of live edge subgraphs of all ER graphs. The probability measure of a live-edge graph is defined by edge generations. Edge $(v,w)$ exists in a live edge subgraph with probability $pq_{vw}.$

To prove that $\source$ is the \emph{only} Jordan infection center, we consider the following asymptotically high probability events.

\begin{itemize}
\item {\bf Offsprings of each node.} Define
\[
E_1=\{\forall v \in \levelset_{\obsTime-1}, \phi'(v)\in((1-\delta)\mu,(1+\delta)\mu)\}.
\]
$E_1$, when occurs, provides upper and lower bounds for the number of offsprings of each node in $\levelset_{\obsTime-1}.$

\item {\bf Collision edges.} We define event $E_2$ when the following upper bound on the collision edges holds
\[
\collisionsize_j
\begin{cases} =0 & \mbox{if } 0<j\leq \lfloor m^-\rfloor, \\
\leq 8\mu & \mbox{if } \lfloor m^-\rfloor<j<\lceil m^+\rceil,\\
\leq \frac{4[(1+\delta)\mu]^{2j+1}}{n} & \mbox{if } \lceil m^+\rceil\leq j\leq \frac{\log n}{(1+\alpha)\log \mu}.
\end{cases}
\]
where $m^+=\frac{\log n}{2\log[(1+\delta)\mu]}$ and $m^-=\frac{\log n-2\log \mu-\log 8}{2\log[(1+\delta)\mu]}.$ $E_2$ provides the upper bounds for collision edges at different levels. Note that a subgraph with diameter $\leq m^-$ is a tree with high probability since there is no collision edge.

\item {\bf Infected nodes.} Define
\begin{align*}
E_3&=\{Z^1_1\geq (1-\delta)^2\mu q\}\\
&\cap \{\forall v\in{\cal Z}^1_1, \cap_{i=2}^t Z^{i}_i(v)\geq (1-\delta)^2\mu qZ^{i-1}_{i-1}(v)\}.
\end{align*}
Level $1$ has at least $(1-\delta)^2\mu q$ infected nodes and the number of nodes grows exponentially by each level with a factor of $(1-\delta)^2\mu q.$ One immediate consequence of event $E_3$ is that
\[
\forall v \in {\cal Z}^1_1, Z^{\obsTime}_{\obsTime}(v)\geq [(1-\delta)^2\mu q]^{\obsTime-1},
\]
i.e., there are at least $[(1-\delta)^2\mu q]^{\obsTime-1}$ infected descendants on level $t$ in $\tree^\dag$ for each infected node on level $1.$

\end{itemize}

 Based on Lemma  \ref{lem:ERAllOffspring}, \ref{lem:E1} and \ref{lem:E3}, for any $\epsilon>0,$ with the union bound, we have that when $t\leq \frac{\log n}{(1+\alpha)\log \mu}$ and $n$ is sufficiently large,
 \begin{align*}
 &\Pr(E_1\cap E_2\cap E_3)\\
 \geq & \Pr(E_1)\left(1-\Pr(\bar{E}_2|E_1)-\Pr(\bar{E}_3|E_1)\right)\\
 \geq &1-\epsilon
 \end{align*}

Next, we show that $\source$ is the only Jordan infection center when $E_1,E_2,E_3$ occur.

For $t\leq\lfloor m^-\rfloor,$ the nodes within $t$ hops from the source form a tree because there is no collision edge (due to event $E_2$). When event $E_3$ occurs, we have $\forall v \in {\cal Z}^1_1, Z^{\obsTime}_{\obsTime}(v)\geq [(1-\delta)^2\mu q]^{\obsTime-1}$ which means  there exists at least one observed infected node on $t$ level for each subtree rooted on level $1.$ Consider infected node $\source'.$ Recall that $a(\source')$ is the ancestor of $\source'$ on level $1$ of $\tree^\dag.$ Consider node $u\in {\cal Z}^1_1$ such that $u\neq a(\source')$ and node $w\in {\cal Z}^t_t(u).$ We have $d^{\tree^\dag}_{\source' w} = d^{\tree^\dag}_{\source' \source}+d^{\tree^\dag}_{\source w}> t.$ Hence the infection eccentricity of $\source'$ is larger than $t.$ Therefore, $\source$ is the only Jordan infection center. The positions of $\source,\source',a(\source'),u$ and $w$ are illustrated in Figure \ref{fig:onebyoneexploration}.

Consider the case $t>\lfloor m^-\rfloor$ and an infected node $\source'$ on level $k\in[1,t].$ In the rest of the proof, we show that there exists node $v\in\infObs$ such that $\dist_{\source' v}>t,$  which means that $s'$ cannot be the Jordan infection center.

\begin{figure}[h!]
  \centering
    \includegraphics[width=\columnwidth]{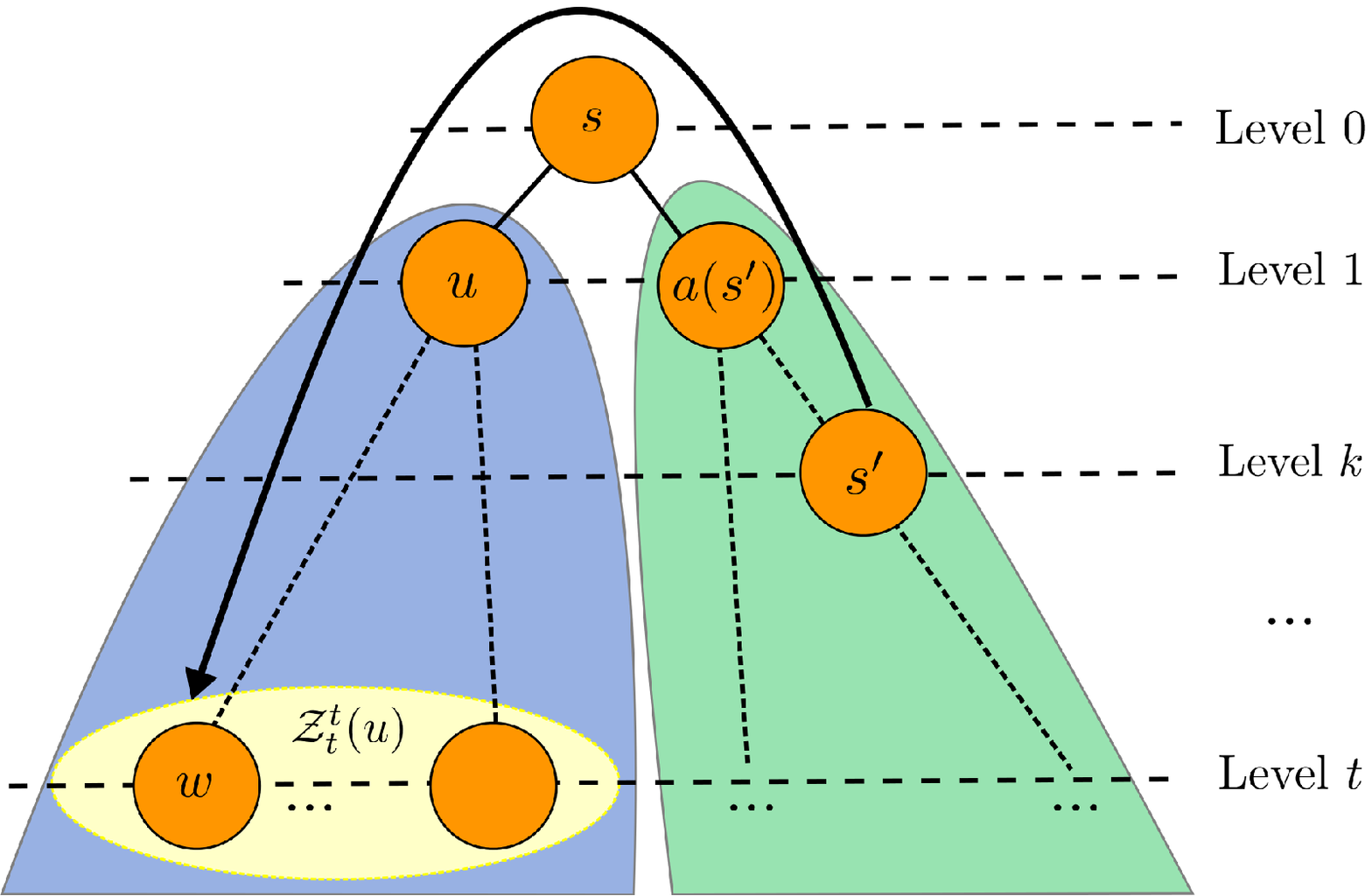}
  \caption{A pictorial example of ${\cal Z}^\obsTime_\obsTime(u)$ in BFS tree $\tree^\dag$}\label{fig:onebyoneexplorationCopy}
\end{figure}
Consider node $u\in {\cal Z}_1^1, u\neq a(\source')$ (the existence of $u$ is guaranteed since $Z^1_1\geq (1-\delta)\mu q\geq 2$).  For the convenience of the reader, we copied Figure \ref{fig:onebyoneexploration} as Figure \ref{fig:onebyoneexplorationCopy} which shows the relative positions of $\source',a(\source'),u,$ and ${\cal Z}^\obsTime_\obsTime(u).$ The distance between a node in ${\cal Z}^\obsTime_\obsTime(u)$ and $\source'$ on the tree $\tree^\dag$ is $k+t.$ Therefore, if $\source'$ is the Jordan infection center, there exists at least one collision node on the path between $\source'$ and each node in ${\cal Z}^\obsTime_\obsTime(u)$ to make the distance $\leq t.$

Define $H$ to be the total number of nodes each of which has the shortest path to $\source'$ within $t$ hops and containing at least one collision node on $g_t$. If $H< Z^\obsTime_\obsTime(u),$ there exists a node $v\in{\cal Z}^\obsTime_\obsTime(u)$ such that $\dist_{\source' v}>t.$ Therefore, $\source'$ can not be the Jordan infection center and the theorem is proved.

In the rest of the proof, we will show that $H< Z^\obsTime_\obsTime(u).$ We first have the lower bound on $Z^\obsTime_\obsTime(u)$ according to $E_3,$
\begin{align}
Z^\obsTime_\obsTime(u)\geq [(1-\delta)^2\mu q]^{t-1}\label{eqn:ztt_upperbound}
\end{align}
The upper bound of $H$ is computed in Lemma \ref{prop:collisionNodesUpperBounds}.
\[
H\leq c[(1+\delta)\mu]^{\frac{3}{4}t+\frac{1}{2}}+c[(1+\delta)\mu]^{(\frac{5}{4}-\frac{\alpha}{2})t+2},
\]

 Since $\frac{1}{2}<\alpha<1,$ we have $\alpha=\frac{1}{2}+\alpha'$ where $0<\alpha'<\frac{1}{2}$ is a constant.
Based on Lemma \ref{prop:collisionNodesUpperBounds}, we have
\begin{align*}
\frac{H}{Z^t_t(u)}&\leq
\frac{c[(1+\delta)\mu]^{\frac{3}{4}t+\frac{1}{2}}+c[(1+\delta)\mu]^{(\frac{3}{2}-\alpha)t+2}}{ [(1-\delta)^2\mu q]^{t-1}}\\
&\leq \frac{c[(1+\delta)\mu]^{\frac{3}{4}t+\frac{1}{2}}}{ [(1-\delta)^2\mu q]^{t-1}}+\frac{c[(1+\delta)\mu]^{(\frac{5}{4}-\frac{\alpha}{2})t+2}}{ [(1-\delta)^2\mu q]^{t-1}}\\
&= \frac{c}{\mu}\left(\frac{(1+\delta)^{\frac{3}{4}+\frac{1}{2t}}}{[(1-\delta)^2q]^{1-\frac{1}{t}}\mu^{\frac{1}{4}-\frac{5}{2t}}}\right)^t\\
&+\frac{c}{\mu}\left(\frac{(1+\delta)^{\frac{5}{4}-\frac{\alpha}{2}+\frac{2}{t}}}{[(1-\delta)^2q]^{1-\frac{1}{t}}\mu^{\frac{\alpha}{2}-\frac{1}{4}-\frac{4}{t}}}\right)^t\\
&\leq \frac{c}{\mu}\left(\frac{(1+\delta)^{\frac{3}{4}+\frac{1}{2t}}}{[(1-\delta)^2q]^{1-\frac{1}{t}}\mu^{\frac{1}{4}-\frac{4}{t}}}\right)^t\\
&+\frac{c}{\mu}\left(\frac{(1+\delta)^{1-\frac{\alpha'}{2}+\frac{2}{t}}}{[(1-\delta)^2q]^{1-\frac{1}{t}}\mu^{\frac{\alpha'}{2}-\frac{4}{t}}}\right)^t
\end{align*}
For $t>16/\alpha'$ we have
\[
\frac{H}{Z^t_t(u)}\leq \frac{2c}{\mu}\left(\frac{(1+\delta)}{[(1-\delta)^2q]\mu^{\frac{\alpha'}{4}}}\right)^t.
\]
Since $\mu>3\log n$ and $\delta,q,\alpha'$ are constants, we have
\[
\frac{(1+\delta)}{[(1-\delta)^2q]\mu^{\frac{\alpha'}{4}}}<1
\]
when
\[
n> \exp\left(\frac{1}{2}\left(\frac{(1+\delta)}{(1-\delta)^2q}\right)^{\frac{4}{\alpha'}}\right).
\]
Therefore, we have
\[
\frac{H}{Z^t_t(u)}\leq \frac{2c}{\mu}\leq \epsilon',
\]
where $\epsilon'\in(0,1)$ is a constant and the inequality holds for sufficiently large $n.$
Therefore, there are at least $(1-\epsilon')Z^t_t(u)$ nodes which cannot be reached from $\source'$ on level $k$ with time $t.$ Hence we have $\ecce(\source',\obs)>t,\forall \source'\neq \source.$

\subsection{Bounds on the Number of Offsprings of Each Node}

\begin{lemma}\label{lem:ERAllOffspring}
Assume the conditions in Theorem \ref{thm:sourceIsJordanER} hold, for any $\epsilon>0,$ we have
\[
\Pr\left(E_1\right)\geq 1-\epsilon
\]
for sufficient large $n.$
\end{lemma}
\begin{proof}
Consider $\delta\in(0,1).$ Since $\obsTime\leq\frac{\log n}{(1+\alpha)\log\mu},$ we have for sufficiently large $n,$
\[
\sum_{i=0}^\obsTime[(1+\delta)\mu]^i\leq 2[(1+\delta)\mu]^t\leq \delta'n,
\]
where $\delta'\in(0,1)$ is a constant which can be arbitrarily close to $0$. This condition shows that the $\obsTime$ hop neighborhood of node $\source$ includes at most a constant fraction of the total number of nodes.

Denote by $\edges(\nodes_1, \nodes_2)$ the set of edges between node set $\nodes_1$ and $\nodes_2.$ Recall that $v_{ij}$ is the $j$th nodes on level $i$ to be explored in the BFS tree starting from the source and $b_i$ is the number of nodes on level $i.$

Define the edge set from $v_{01}$ to all other nodes in the ER graph to be $$\Psi(v_{01}) = \edges(\{v_{01}\}, \nodes(g)\backslash\{v_{01}\}),$$ which is the set of edges between $v_{01}$ and all other nodes in the graph.

Define $$\Phi'(v_{01})=\{v|(v,v_{01})\in \Psi(v_{01})\}.$$

Define $$\Psi(v_{01},v_{11}) =  \edges(\{v_{11}\}, \nodes(g)\backslash(\Phi'(v_{01})\cup \{v_{01}\})),$$ which is the set of edges from node $v_{11}$ to all nodes that are not already included in the BFS tree and $$\Phi'(v_{01},v_{11}) = \{v|(v,v_{11})\in \Psi(v_{01},v_{11})\},$$
which is the set of offsprings of $v_{11}.$

For simplicity, we use $\Psi(v_{ij})$ to denote
\[
\Psi(v_{01},v_{11},\cdots, v_{1 b_1},\cdots, v_{i1},\cdots,v_{ij}).
\]
and use $\Phi'(v_{ij})$ to denote
\[
\Phi'(v_{01},v_{11},\cdots, v_{1 b_1},\cdots, v_{i1},\cdots,v_{ij}).
\]

Iteratively, we define
\begin{align*}
&\Psi(v_{ij})\\
 \triangleq &\edges(\{v_{ij}\}, \nodes(g)\backslash(\{v_{01}\}\cup \Phi'(v_{01})\cup \cdots \cup \Phi'(v_{i j-1})))
\end{align*}
and
\begin{align*}
&\Phi'(v_{ij}) \triangleq\{v|(v,v_{ij})\in\Psi(v_{ij})\}
\end{align*}
which is the set of offsprings of node $v_{ij}$ in the BFS tree from the source.
Define $\phi'(v_{ij}) = |\Phi'(v_{ij})|$ and $\psi(v_{ij}) = |\Psi(v_{ij})|.$

Note that $\Psi(v_{ij})$ uniquely determines $\Phi'(v_{ij})$ and vice versa. In addition, according to the definition, $\Psi(v_{ij})$ for any $i$ and $j$ are pairwise disjoint.

Define
\[
\Lambda(v_{ij}) = \{\Phi'(v_{ij})|\phi'(v_{ij})\in(\mu(1-\delta),\mu(1+\delta))\}.
\]
which is the set of $\Phi'(v_{ij})$ which satisfies the given bounds on the number of offsprings.

Therefore, we have
\begin{align*}
&\Pr(E_1)\\
&=\Pr\left(\forall v \in \levelset_{\obsTime-1}, \phi'(v)\in(\mu(1-\delta),\mu(1+\delta))\right)\\
&= \sum_{\Phi'(v_{01})\in \Lambda(v_{01})}\Pr\Big(\Phi'(v_{01}), \\
& \phi'(v)\in(\mu(1-\delta),\mu(1+\delta)),\forall v \in \levelset_{\obsTime-1}\backslash\{v_{01}\}\Big)\\
& = \sum_{\Phi'(v_{01})\in \Lambda(v_{01})}\Pr(\Phi'(v_{01}))\Pr\big(\forall v \in \levelset_{\obsTime-1}\backslash\{v_{01}\},\\
& \phi'(v)\in(\mu(1-\delta),\mu(1+\delta))|\Phi'(v_{01})\big)
\end{align*}
Given $\Phi'(v_{01}),$ the order of the nodes to be explored during the construction of the BFS tree on the next level is determined and we have
\begin{align*}
&\Pr\left(\forall v \in \levelset_{\obsTime-1}, \phi'(v)\in(\mu(1-\delta),\mu(1+\delta))\right)\\
& = \sum_{\Phi'(v_{01})\in \Lambda(v_{01})}\Pr(\Phi'(v_{01}),\Phi'(v_{11}))\\
&\times\Pr\big(\forall v \in \levelset_{\obsTime-1}\backslash\{v_{01} , v_{11}\},\\
& \phi'(v)\in(\mu(1-\delta),\mu(1+\delta))|\Phi'(v_{01}) ,\Phi'(v_{11}) \big)
\end{align*}
Iteratively, we have
\begin{align}
&\Pr\left(\forall v \in \levelset_{\obsTime-1}, \phi'(v)\in(\mu(1-\delta),\mu(1+\delta))\right)\\
& = \sum_{\Phi'(v_{01})\in \Lambda(v_{01})}\cdots \sum_{\Phi'(v_{t-1 b_{t-1}-1})\in \Lambda(v_{t-1 b_{t-1}-1})}\label{eqn:iterativeForNumberOfOffsprings1}\\
&\Pr(\Phi'(v_{01}),\cdots,\Phi'(v_{t-1 b_{t-1}-1}))\\
&\times\Pr\Big(\phi'(v_{t-1 b_{t-1}})\in(\mu(1-\delta),\mu(1+\delta))\\
&|\Phi'(v_{01}),\cdots,\Phi'(v_{t-1 b_{t-1}-1})\Big)\label{eqn:iterativeForNumberOfOffsprings2}
\end{align}
Next, we focus on the last term in Equation (\ref{eqn:iterativeForNumberOfOffsprings2}). Note, $\Psi(v_{ij})$ uniquely determines $\Phi'(v_{ij})$ and vice versa.
Therefore,
\begin{align}
\Pr\Big(&\phi'(v_{t-1 b_{t-1}})\in(\mu(1-\delta),\mu(1+\delta))\\
&|\Phi'(v_{01}),\cdots,\Phi'(v_{t-1 b_{t-1}-1})\Big)\\
= \Pr\Big(&\phi'(v_{t-1 b_{t-1}})\in(\mu(1-\delta),\mu(1+\delta))\\
&|\Psi(v_{01}),\cdots,\Psi(v_{t-1 b_{t-1}-1})\Big)
\end{align}
Since $\Psi(v_{t-1 b_{t-1}})$ is disjoint with $\Psi(v_{01}),\cdots,\Psi(v_{t-1 b_{t-1}-1})$ and each edge is generated independently in the ER graph. Therefore, conditioned on $\Psi(v_{01}),\cdots,\Psi(v_{t-1 b_{t-1}-1}),$ we have
$\phi'(v_{t-1 b_{t-1}})$ follows
\[
\hbox{Bi}(n - \sum_{i=0}^{t-2}\sum_{j=1}^{b_{i}}\phi'(v_{ij})-\sum_{j=1}^{b_{t-1}-1}\phi'(v_{t-1j})-1, p).
\]

Note, $\phi(v_{01}),\cdots,\phi(v_{t-1 b_{t-1}-1})$ are in $(\mu(1-\delta),\mu(1+\delta))$ according to the condition in Equation (\ref{eqn:iterativeForNumberOfOffsprings1}). Hence
\[
\sum_{i=0}^{t-2}\sum_{j=1}^{b_{i}}\phi'(v_{ij})+\sum_{j=1}^{b_{t-1}-1}\phi'(v_{t-1j})+1\leq\sum_{i = 0}^{t}[\mu(1+\delta))]^i
\]

Therefore, $\phi'(v_{t-1 b_{t-1}})$ stochastically dominates $\hbox{Bi}(n - \sum_{i = 0}^{t}[\mu(1+\delta))]^i, p)$ and is stochastically dominated by $\hbox{Bi}(n,p)$ which implies

\begin{align*}
&\Pr\big(\phi'(v_{t-1 b_{t-1}})\in(\mu(1-\delta),\mu(1+\delta))\\
&|\Phi'(v_{01}),\cdots,\Phi'(v_{t-1 b_{t-1}-1})\big)\\
&\geq 1-\Pr\left(\hbox{Bi}\left(n-\sum_{i = 0}^{t}[\mu(1+\delta))]^i,p\right)\leq (1-\delta)\mu\right)\\
&-\Pr\left(\hbox{Bi}\left(n,p\right)\geq \mu(1+\delta )\right)
\end{align*}
Note  $\sum_{i=0}^\obsTime[(1+\delta )\mu]^i\leq \delta' n.$ Therefore, we have
\begin{align}
&\Pr\big(\phi'(v_{t-1 b_{t-1}})\in(\mu(1-\delta),\mu(1+\delta))\\
&|\Phi'(v_{01}),\cdots,\Phi'(v_{t-1 b_{t-1}-1})\big)\\
&\geq 1-\Pr\left(\hbox{Bi}\left((1-\delta')n,p\right)\leq (1-\delta)\mu\right)\\
&-\Pr\left(\hbox{Bi}\left(n,p\right)\geq \mu(1+\delta )\right)\label{eqn:singleCase}
\end{align}

By using the Chernoff bound in Lemma \ref{lem:chernoff}, we have
\[
\Pr\left(\hbox{Bi}\left((1-\delta')n,p\right)\leq \mu(1-\delta)\right)\leq \exp\left(-\frac{(\delta-\delta')^2\mu}{2(1-\delta')}\right),
\]
and
\[
\Pr\left(\hbox{Bi}\left(n,p\right)\geq \mu(1+\delta)\right)\leq  \exp\left(-\frac{\delta ^2\mu}{2+\delta }\right)
\]
Substitute into Inequality (\ref{eqn:singleCase}), we obtain
\begin{align}
&\Pr\big(\phi'(v_{t-1 b_{t-1}})\in(\mu(1-\delta),\mu(1+\delta))\\
&|\Phi'(v_{01}),\cdots,\Phi'(v_{t-1 b_{t-1}-1})\big)\\
&\geq 1-\exp\left(-\frac{(\delta-\delta')^2\mu}{2(1-\delta')}\right) -\exp\left(-\frac{\delta ^2\mu}{2+\delta }\right)\triangleq \Delta\label{eqn:singleCaseFinal}
\end{align}
Substitute Inequality (\ref{eqn:singleCaseFinal}) into Equation (\ref{eqn:iterativeForNumberOfOffsprings2}), we obtain
\begin{align}
&\Pr\left(\forall v \in \levelset_{\obsTime-1}, \phi'(v)\in(\mu(1-\delta),\mu(1+\delta))\right)\\
&\geq \sum_{\Phi'(v_{01})\in \Lambda(v_{01})}\cdots \sum_{\Phi'(v_{t-1 b_{t-1}-1})\in \Lambda(v_{t-1 b_{t-1}-1})}\\
&\Pr(\Phi'(v_{01}),\cdots,\Phi'(v_{t-1 b_{t-1}-1}))\times\Delta\\
& = \Delta\sum_{\Phi'(v_{01})\in \Lambda(v_{01})}\cdots \sum_{\Phi'(v_{t-1 b_{t-1}-2})\in \Lambda(v_{t-1 b_{t-1}-2})}\\
&\Big( \sum_{\Phi'(v_{t-1 b_{t-1}-1})\in \Lambda(v_{t-1 b_{t-1}-1})}\Pr(\Phi'(v_{01}),\cdots,\Phi'(v_{t-1 b_{t-1}-1})) \Big)\\
& = \Delta\sum_{\Phi'(v_{01})\in \Lambda(v_{01})}\cdots \sum_{\Phi'(v_{t-1 b_{t-1}-2})\in \Lambda(v_{t-1 b_{t-1}-2})}\\
&\Pr(\Phi'(v_{01}),\cdots,\Phi'(v_{t-1 b_{t-1}-2}))\\
& \Pr\big(\phi'(v_{t-1 b_{t-1}-1})\in(\mu(1-\delta),\mu(1+\delta))\\
&|\Phi'(v_{01}),\cdots,\Phi'(v_{t-1 b_{t-1}-2})\big)\\
& = \Delta^2\sum_{\Phi'(v_{01})\in \Lambda(v_{01})}\cdots \sum_{\Phi'(v_{t-1 b_{t-1}-2})\in \Lambda(v_{t-1 b_{t-1}-2})}\\
&\Pr(\Phi'(v_{01}),\cdots,\Phi'(v_{t-1 b_{t-1}-2}))\label{eqn:IterativeApply}
\end{align}
Applying Equation (\ref{eqn:IterativeApply}) iteratively, we have
\begin{align*}
&\Pr\left(\forall v \in \levelset_{\obsTime-1}, \phi'(v)\in(\mu(1-\delta),\mu(1+\delta))\right)\\
\geq & \Delta^{\sum_{i=0}^{\obsTime-1}[(1+\delta )\mu]^i}\\
\geq & \bigg(1-\exp\left(-\frac{(\delta-\delta')^2\mu}{2(1-\delta')}\right)-\exp\left(-\frac{\delta ^2\mu}{2+\delta }\right)\bigg)^{\sum_{i=0}^{\obsTime-1}[(1+\delta )\mu]^i}\\
\end{align*}
 When $\delta'\rightarrow 0,$ we have
\[
\frac{(\delta-\delta')^2\mu}{2(1-\delta')}\rightarrow \frac{\delta^2}{2}>\frac{\delta^2}{2+\delta}
\]
Therefore, we can choose a sufficiently small $\delta'$ such that
\begin{align*}
&\Pr\left(\forall v \in \levelset_{\obsTime-1}, \phi'(v)\in(\mu(1-\delta),\mu(1+\delta))\right)\\
\geq & \left(1-2\exp\left(-\frac{\delta^2\mu}{2+\delta}\right)\right)^{\sum_{i=0}^{\obsTime-1}[(1+\delta )\mu]^i}\\
\geq & \left(1-2\exp\left(-\frac{\delta^2\mu}{2+\delta}\right)\right)^{2[(1+\delta )\mu]^{\obsTime-1}}\\
\geq_{(a)}& \exp\left(-8[(1+\delta )\mu]^{\obsTime-1}\exp\left(-\frac{\delta^2\mu}{2+\delta}\right)\right)\\
\geq & \exp\left(-8\exp\left(-\frac{\delta^2\mu}{2+\delta}+(\obsTime-1)\log[(1+\delta)\mu]\right)\right),
\end{align*}
where $(a)$ is based on Lemma \ref{lem:exponentialBounds} and holds when $\mu$ is sufficiently large (i.e., when $n$ is sufficiently large).
To make the above bound greater than $1-\epsilon,$ we need
\[
\obsTime\leq \frac{\frac{\delta^2\mu}{2+\delta}-\log8+\log\log\left(\frac{1}{1-\epsilon}\right)}{\log(1+\delta)+\log \mu}+1.
\]
When $\mu> \frac{2+\delta}{\delta^2}\log n$, we have
\begin{align*}
\obsTime &\leq\frac{\log n}{(1+\alpha)\log\mu}\\
&<\frac{\log n-\log8+\log\log\left(\frac{1}{1-\epsilon}\right)}{\log(1+\delta)+\log\mu}+1.
\end{align*}
for sufficiently large $n.$

Note $ \frac{2+\delta}{\delta^2}\rightarrow 3$ when $\delta \rightarrow 1$ which matches the condition that $\mu>3\log n.$ Therefore, we prove the lemma.

\end{proof}

\subsection{Bounds on the Number of Collision Edges}

Next, we analyze the number of collision edges on different levels. We have the following lemma.

\begin{lemma}\label{lem:E1}
 If the conditions in Theorem \ref{thm:sourceIsJordanER} hold, for any $\epsilon>0,$
\[
\Pr(E_2|E_1)\geq 1-\epsilon
\]
for sufficiently large $n.$
\end{lemma}
\begin{proof}
We have
\begin{align*}
\Pr(E_2|E_1)&\geq 1-\Pr(R_{\lfloor m^-\rfloor}\neq 0|E_1)\\
&-\sum_{j=\lfloor m^-\rfloor+1}^{\lceil m^+\rceil-1}\Pr(R_j>8\mu|E_1)\\
&-\sum_{j=\lceil m^+\rceil}^t\Pr\left(R_j>\frac{4[(1+\delta)\mu]^{2j+1}}{n}|E_1\right)
\end{align*}
\begin{itemize}
\item {\bf No collision edge at the first $\lfloor m^- \rfloor$ levels.}

We will show that
\[
\Pr\left(\collisionsize_{\lfloor m^-\rfloor}\neq 0|E_1\right)\leq 1-\exp\left(-\frac{1}{\mu}\right)\leq \frac{1}{\mu}
\]
when $n$ is sufficiently large.

Conditioning on $E_1,$ we have $\collisionsize_j$ is stochastically dominated by $\hbox{Bi}(l_j^2,p).$ Since $l_j\leq 2[(1+\delta)\mu]^j,$ $\collisionsize_j$ is stochastically dominated by $\hbox{Bi}(4[(1+\delta)\mu]^{2j},p).$ We have for sufficiently large $n,$
\begin{align*}
\Pr\left(\collisionsize_j=0\big|E_1\right)\geq& \left(1-p\right)^{\left(2 [(1+\delta)\mu]^j\right)^2}\\
= & \left(1-\frac{\mu}{n}\right)^{4 [(1+\delta)\mu]^{2j}}\\
\geq_{(a)} & \exp\left(-8[(1+\delta)\mu]^{2j}\frac{\mu}{n}\right)\\
\geq_{(b)}&\exp\left(-\frac{1}{\mu}\right)
\end{align*}
Inequality $(a)$ is based on Lemma \ref{lem:exponentialBounds}.
To obtain Inequality $(b),$ note $$j\leq m^-=\frac{\log n-2\log \mu-\log 8}{2\log[(1+\delta)\mu]}.$$
 We have
\[
8[(1+\delta)\mu]^{2j}\frac{\mu}{n}\leq \frac{1}{\mu}
\]
which explains $(b).$
\item {\bf The number of collision edges at levels between $\lfloor m^-\rfloor+1$ and $\lceil m^+\rceil -1$.}

We will show
\[
\Pr\left(\collisionsize_j> \frac{4[(1+\delta)\mu]^{2j+1}}{n}\bigg|E_1 \right)\leq\exp\left(-\frac{4\delta^2}{2+\delta}\mu\right)
\]
when $n$ is sufficiently large.

Let
\[
\delta'=\frac{2n}{[(1+\delta)\mu]^{2j}}-1
\]
Since $j\leq m^+ = \frac{\log n}{2\log[(1+\delta)\mu]},$ we have $n\geq[(1+\delta)\mu]^{2j}.$
Hence
\[
\delta'\geq 1
\]
Conditioned on event $E_1,$ $\collisionsize_j$ is stochastically dominated by $\hbox{Bi}(4[(1+\delta)\mu]^{2j},p).$
Using the Chernoff bound in Lemma \ref{lem:chernoff}, we have,
\begin{align*}
&\Pr\left(\collisionsize_j\leq (1+\delta')4[(1+\delta)\mu]^{2j}p \big|E_1\right)\\
&\geq \Pr\left(\hbox{Bi}(4[(1+\delta)\mu]^{2j},p)\leq (1+\delta')4[(1+\delta)\mu]^{2j}p\right)\\
&\geq 1-\exp\left(-\frac{\delta'^2}{2+\delta'} 4[(1+\delta)\mu]^{2j}\frac{\mu}{n}\right)\\
&\geq 1-\exp\left(-\frac{\delta'}{2+\delta'} 4(2n-[(1+\delta)\mu]^{2j})\frac{\mu}{n}\right)\\
&\geq_{(a)}1-\exp\left(-\frac{\delta'}{2+\delta'} 4\mu\right)\\
&\geq_{(b)}1-\exp\left(-\frac{4}{3}\mu\right)\\
\end{align*}
Note $(a)$ is due to $n>[(1+\delta)\mu]^{2j}$ and $(b)$ is due to $\delta'\geq 1.$
Note,
\begin{align*}
&(1+\delta')4[(1+\delta)\mu]^{2j}p\\
=&\left(1+\frac{2n}{[(1+\delta)\mu]^{2j}}-1\right)4[(1+\delta)\mu]^{2j}p\\
=&\frac{2n}{[(1+\delta)\mu]^{2j}}4[(1+\delta)\mu]^{2j}\frac{\mu}{n}\\
=& 8 \mu
\end{align*}

\item {\bf The number of collision edges at levels between $\lceil m^+\rceil$ and $\frac{\log n}{(1+\alpha) \log \mu}.$ }

We will show that
\[
\Pr\left(\collisionsize_{j}> 8\mu|E_1\right)\leq \exp\left(-\frac{4}{3}\mu\right)
\]
when $n$ is sufficiently large.

Conditioned on event $E_1,$ $\collisionsize_j$ is stochastically dominated  by $\hbox{Bi}(l_j^2,p).$ Since $l_j\leq 2[(1+\delta)\mu]^j,$ $\collisionsize_j$ is stochastically dominated by $\hbox{Bi}(4[(1+\delta)\mu]^{2j},p).$
Then
\begin{align*}
&\Pr\left(\collisionsize_j\leq \frac{4 [(1+\delta)\mu]^{2j+1}}{n} \big|E_1\right)\\
&\geq \Pr\left(\hbox{Bi}(4[(1+\delta)\mu]^{2j},p)\leq \frac{4 [(1+\delta)\mu]^{2j+1}}{n}\right)\\
&\geq \Pr\left(\hbox{Bi}(4[(1+\delta)\mu]^{2j},p)\leq (1+\delta)4 [(1+\delta)\mu]^{2j}p\right)\\
&\geq 1-\exp\left(-\frac{\delta^2}{2+\delta} 4[(1+\delta)\mu]^{2j}\frac{\mu}{n}\right)\\
&\geq_{(a)} 1-\exp\left(-\frac{\delta^2}{2+\delta}4\mu\right)
\end{align*}
From $j\geq \lceil m^+\rceil\geq \frac{\log n}{2\log[(1+\delta)\mu]},$ we obtain
\begin{align}
n\leq [(1+\delta)\mu]^{2j}.\label{eqn:LowLevelCollision}
\end{align}
we obtain Inequality $(a)$ by substituting Inequality (\ref{eqn:LowLevelCollision}).

\end{itemize}

Since $m^+-m^-< 2$  we have
\begin{align*}
\Pr(E_2|E_1)\geq & 1-\frac{1}{\mu}-(m^+-m^-)\exp\left(-\frac{4}{3}\mu\right)\\
&-\sum_{j=\lceil m ^+\rceil}^t\exp\left(-\frac{4\delta^2}{2+\delta}\mu\right)\\
\geq & 1-\frac{1}{\mu}-2\exp\left(-\frac{4}{3}\mu\right)-t\exp\left(-\frac{4\delta^2}{2+\delta}\mu\right)
\end{align*}

Note we have $\mu\geq 3\log n$ and $\obsTime\leq\frac{\log n}{(1+\alpha)\log\mu},$ therefore, for $n$ sufficiently large, we have
\[
\Pr(E_2|E_1)\geq 1-\epsilon.
\]
\end{proof}

\subsection{Bounds on the Number of Infected Nodes}
\begin{lemma}\label{lem:E3}
Assume the conditions in Theorem \ref{thm:sourceIsJordanER} hold, for any $\epsilon>0,$ we have
\[
\Pr\left(E_3|E_1\right)\geq 1-\epsilon
\]
for sufficiently large $n.$
\end{lemma}
\begin{proof}
\begin{itemize}
\item We first show that for any $\epsilon>0,$
\[
\Pr(Z^1_1\geq(1-\delta)^2\mu q|E_1)\geq 1-\epsilon
\]
for sufficient large $n.$
$Z^1_1$ is lower bounded by a binomial distribution $\hbox{Bi}((1-\delta)\mu, q).$ Hence using Chernoff bound, we have
\[
\Pr(Z^1_1\geq(1-\delta)(1-\delta)\mu q|E_1)\geq 1-\exp\left(-\frac{\delta^2}{2}(1-\delta)\mu q\right)
\]
Note $\mu\rightarrow \infty$ while all other parameters are constants, for any $\epsilon>0,$ we have
\[
\Pr(Z^1_1\geq(1-\delta)^2\mu q|E_1)\geq 1-\epsilon
\]

\item We show that for any $\epsilon>0,$
\[
\Pr\left(\{\forall v \in {\cal Z}^1_1, Z^{\obsTime}_{\obsTime}(v)\geq [(1-\delta)^2\mu q]^{\obsTime-1} \}|E_1\right)\geq 1-\epsilon
\]
for sufficiently large $n.$
Define
\[
E_4=\{(1-\delta)^2\mu q\leq Z^1_1\leq(1+\delta)\mu\}
\]
Note when $n$ is sufficiently large, the following holds
\[
\Pr(Z^1_1\geq(1-\delta)^2\mu q)\geq 1-\frac{\epsilon}{4}.
\]
When $E_1$ occurs, we have $Z^1_1\leq(1+\delta)\mu.$ Therefore, we have
\[
\Pr(E_4|E_1)\geq 1-\frac{\epsilon}{4}
\]

Define
\begin{align*}
&{\cal S}^t_2(v)=\{({\cal Z}^2_2(v),\cdots,{\cal Z}^t_t(v))| \\
&\cap_{i=2}^t Z^i_i(v)\geq (1-\tilde{\delta} )(1-\delta)\mu qZ^{i-1}_{i-1}(v).\}
\end{align*}
We have
\begin{align}
&\Pr\left(\cap_{i=2}^t Z^{i}_i(v)\geq (1-\tilde{\delta} )(1-\delta)\mu qZ^{i-1}_{i-1}(v)|E_4,E_1\right)\\
=&\Pr\big( Z^{t}_t(v)\geq (1-\tilde{\delta} )(1-\delta)\mu qZ^{t-1}_{t-1}(v),\\
&\cap_{i=2}^{t-1} Z^{i}_i(v)\geq (1-\tilde{\delta} )(1-\delta)\mu qZ^{i-1}_{i-1}(v)|E_4,E_1\big)\\
&= \sum_{{\cal Z}^2_2(v),\cdots,{\cal Z}^{t-1}_{t-1}(v)\in {\cal S}^{t-1}_2(v)}\Pr\big( Z^{t}_t(v)\geq \\
&(1-\tilde{\delta} )(1-\delta)\mu qZ^{t-1}_{t-1}(v)|{\cal Z}^{2}_2(v),\cdots, {\cal Z}^{t-1}_{t-1}(v), E_4,E_1\big)\\
&\times \Pr({\cal Z}^{2}_2(v),\cdots, {\cal Z}^{t-1}_{t-1}(v)| E_4,E_1)\label{eqn:ZiiIterative}
\end{align}

Conditioned on $E_1$ and $E_4,$ we have $Z_1^1\neq 0.$
For any $v\in {\cal Z}^1_1,$ $Z^i_i(v)$ stochastically dominates $\hbox{Bi}((1-\delta)\mu Z^{i-1}_{i-1}(v),q)$ given ${\cal Z}^{i-1}_{i-1}(v).$ Therefore, denote by $\tilde{\delta} \in(0,1),$ we have
\begin{align*}
&\Pr\big( Z^{t}_t(v)\geq (1-\tilde{\delta} )(1-\delta)\mu qZ^{t-1}_{t-1}(v)\\
&|{\cal Z}^{2}_2(v),\cdots, {\cal Z}^{t-1}_{t-1}(v), E_4,E_1\big)\\
\geq &\Pr(Z^t_t(v)\geq (1-\tilde{\delta} )(1-\delta)\mu qZ^{t-1}_{t-1}(v)|{\cal Z}^{t-1}_{t-1}(v),E_4,E_1)\\
\geq & 1-\exp\left(-\frac{\tilde{\delta} ^2(1-\delta)\mu q Z^{t-1}_{t-1}(v)}{2+\tilde{\delta} }\right)\\
\geq & \exp\left(-2\exp\left(-\frac{\tilde{\delta} ^2(1-\delta)\mu q Z^{t-1}_{t-1}(v)}{2+\tilde{\delta} }\right)\right)
\end{align*}
Since we have ${\cal Z}^2_2(v),\cdots,{\cal Z}^{t-1}_{t-1}(v)\in {\cal S}^{t-1}_2(v),$ therefore,
\[
 Z^{t-1}_{t-1}(v)\geq [(1-\tilde{\delta} )(1-\delta)\mu q]^{t-2}
\]
Hence, we have
\begin{align*}
&\Pr\big( Z^{t}_t(v)\geq (1-\tilde{\delta} )(1-\delta)\mu qZ^{t-1}_{t-1}(v)\\
&|{\cal Z}^{2}_2(v),\cdots, {\cal Z}^{t-1}_{t-1}(v), E_4,E_1\big)\\
\geq & \exp\left(-2\exp\left(-\frac{\tilde{\delta} ^2(1-\delta)\mu q [(1-\tilde{\delta} )(1-\delta)\mu q]^{t-2}}{2+\tilde{\delta} }\right)\right)\\
\geq & \exp\left(-2\exp\left(-\frac{\tilde{\delta} ^2[(1-\tilde{\delta} )(1-\delta)\mu q]^{t-1}}{(2+\tilde{\delta}) (1-\tilde{\delta} )}\right)\right)
\end{align*}
Substituting back to Equation (\ref{eqn:ZiiIterative}), we obtain \begin{align}
&\Pr\left(\cap_{i=2}^t Z^{i}_i(v)\geq (1-\tilde{\delta} )(1-\delta)\mu qZ^{i-1}_{i-1}(v)|E_4,E_1\right)\\
&\geq \sum_{{\cal Z}^2_2(v),\cdots,{\cal Z}^{t-1}_{t-1}(v)\in {\cal S}^{t-1}_2(v)}\exp\Big(-2\exp\big(\\
&-\frac{\tilde{\delta} ^2[(1-\tilde{\delta} )(1-\delta)\mu q]^{t-1}}{(2+\tilde{\delta}) (1-\tilde{\delta} )}\big)\Big)\\
&\times \Pr({\cal Z}^{2}_2(v),\cdots, {\cal Z}^{t-1}_{t-1}(v)| E_4,E_1)\\
&=\exp\left(-2\exp\left(-\frac{\tilde{\delta} ^2[(1-\tilde{\delta} )(1-\delta)\mu q]^{t-1}}{(2+\tilde{\delta}) (1-\tilde{\delta} )}\right)\right)\\
&\times\sum_{{\cal Z}^2_2(v),\cdots,{\cal Z}^{t-1}_{t-1}(v)\in {\cal S}^{t-1}_2(v)}\Pr({\cal Z}^{2}_2(v),\cdots, {\cal Z}^{t-1}_{t-1}(v)| E_4,E_1)\\
&=\exp\left(-2\exp\left(-\frac{\tilde{\delta} ^2[(1-\tilde{\delta} )(1-\delta)\mu q]^{t-1}}{(2+\tilde{\delta}) (1-\tilde{\delta} )}\right)\right)\\
&\times \Pr(\cap_{i=2}^{t-1} Z^{i}_i(v)\geq (1-\tilde{\delta} )(1-\delta)\mu qZ^{i-1}_{i-1}(v)| E_4,E_1)
\label{eqn:ZiiIterative2}
\end{align}
Use Equation (\ref{eqn:ZiiIterative2}) iteratively on all levels, we obtain
\begin{align*}
&\Pr\left(\cap_{i=2}^t Z^{i}_i(v)\geq (1-\tilde{\delta} )(1-\delta)\mu qZ^{i-1}_{i-1}(v)|E_4,E_1\right)\\
\geq & \prod_{i=2}^t\exp\left(-2\exp\left(-\frac{\tilde{\delta} ^2[(1-\tilde{\delta} )(1-\delta)\mu q]^{i-1}}{(2+\tilde{\delta}) (1-\tilde{\delta} )}\right)\right)\\
=&  \exp\left(-2\sum_{i=2}^t\exp\left(-\frac{\tilde{\delta} ^2[(1-\tilde{\delta} )(1-\delta)\mu q]^{i-1}}{(2+\tilde{\delta}) (1-\tilde{\delta} )}\right)\right)\\
\geq& \exp\left(-2(t-1)\exp\left(-\frac{\tilde{\delta} ^2(1-\delta )\mu q}{2+\tilde{\delta} }\right)\right)\\
\geq & 1-2(t-1)\exp\left(-\frac{\tilde{\delta} ^2(1-\delta )\mu q}{2+\tilde{\delta} }\right)
\end{align*}
Using the union bound for all $v\in{\cal Z}^1_1,$ we have
\begin{align*}
&\Pr\bigg(\forall v\in{\cal Z}^1_1, \cap_{i=2}^t Z^{i}_i(v)\\
&\geq (1-\tilde{\delta} )(1-\delta)\mu qZ^{i-1}_{i-1}(v)|E_4,E_1\bigg)\\
\geq &1-2(1+\delta)\mu t\exp\left(-\frac{\tilde{\delta} ^2(1-\delta )\mu q}{2+\tilde{\delta} }\right)
\end{align*}
Note $t\leq \frac{\log n }{(1+\alpha)\log \mu}$ and $\mu>3\log n.$ We have $t\leq \log n \leq \mu,$ and

\begin{align*}
&\Pr\bigg(\forall v\in{\cal Z}^1_1, \cap_{i=1}^t Z^{i}_i(v)\\
&\geq (1-\tilde{\delta} )(1-\delta)\mu qZ^{i-1}_{i-1}(v)|E_4,E_1\bigg)\\
&\geq 1-2(1+\delta)\mu^2\exp\left(-\frac{\tilde{\delta} ^2(1-\delta )\mu q}{2+\tilde{\delta} }\right)\\
&\geq 1-\frac{\epsilon}{2}
\end{align*}
for sufficiently large $n.$

Define
\[
E_5 =  \{\forall v\in{\cal Z}^1_1, \cap_{i=2}^t Z^{i}_i(v)\geq (1-\delta)^2\mu qZ^{i-1}_{i-1}(v)\}
\]
We further have,
\begin{align*}
&\Pr\left(E_5|E_1\right)\\
\geq &\Pr\left(E_5|E_4,E_1\right)\\
\times & \Pr(E_4|E_1)\\
\geq & (1-\frac{\epsilon}{4})(1-\frac{\epsilon}{4})\\
\geq & 1-\frac{\epsilon}{2}.
\end{align*}
Choosing $\tilde{\delta}=\delta,$ we have
\begin{align*}
&\Pr\left(\forall v\in{\cal Z}^1_1, \cap_{i=2}^t Z^{i}_i(v)\geq (1-\delta)^2\mu qZ^{i-1}_{i-1}(v)|E_1\right)\\
\geq  & 1-\frac{\epsilon}{2}.
\end{align*}
\end{itemize}
Note $E_3 = E_4\cap E_5.$
We have
\begin{align*}
\Pr(E_3|E_1) &= \Pr( E_4\cap E_5|E_1)\\
&\geq 1-  \Pr( \bar{E}_4|E_1)- \Pr(\bar{E}_5|E_1) \\
&\geq 1-\epsilon.
\end{align*}
where $\bar{E}_4, \bar{E}_5$ are the complement of event $E_4,E_5.$

\end{proof}

\begin{lemma}\label{prop:collisionNodesUpperBounds}
When $E_1$ and $E_2$ occurs, if $\lfloor m^-\rfloor<t\leq\frac{\log n}{(1+\alpha)\log \mu},$ we have
\[
H\leq c[(1+\delta)\mu]^{\frac{3}{4}t+\frac{1}{2}}+c[(1+\delta)\mu]^{(\frac{5}{4}-\frac{\alpha}{2})t+2},
\]
where $c$ is a constant.
\end{lemma}
\begin{figure}[h!]
  \centering
    \includegraphics[width=\columnwidth]{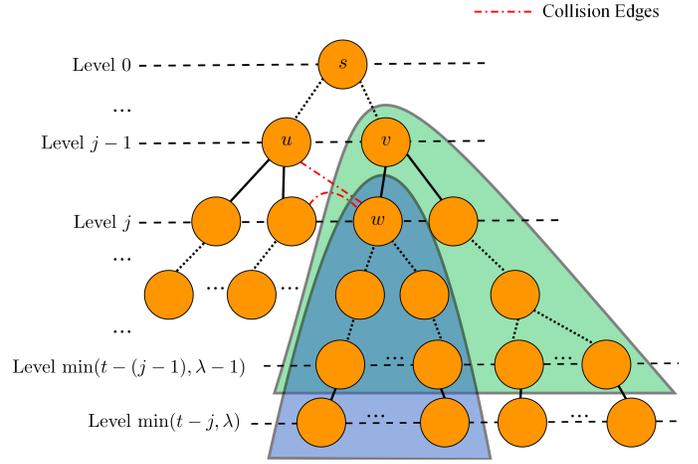}
   \caption{A pictorial example of upper bounds of $H$ }\label{fig:Hupperbound}
\end{figure}
\begin{proof}

Define a \emph{collision removed} breadth-first search tree to be a BFS tree on the graph with all collision nodes removed. Denote by ${\cal U}^{h}(v)$ the set of nodes in the {\bf collision removed} BFS tree from node $v$ with $h$ hops in $g_t$ and $U^h(v)=|{\cal U}^{h}(v)|.$ Denote by $\tilde{{\cal U}}^{h}(v)$ the set of nodes that are within $h$ hops from node $v$ on $g_t.$ Recall that a node $v$ is said to locate on level $j$ if $\dist_{\source v}=j.$ Denote by ${\cal C}_j$ the set of collision nodes on level $j$ in $g_t.$ Therefore, we have
\[
H\leq \left|\cup_{j=0}^t\cup_{v\in {\cal C}_j} \tilde{{\cal U}}^{t-d_{\source' v}}(v)\right|,
\]
where $\tilde{\cal U}^{t-d_{s'v}}(v)$ is the set of nodes that can be reached from $s'$ within $t$ hops via the collision node $v.$
Next, we will prove that
 \[
 \cup_{j=0}^t\cup_{v\in {\cal C}_j} \tilde{{\cal U}}^{t-d_{\source' v}}(v) = \cup_{j=0}^t\cup_{v\in {\cal C}_j} {\cal U}^{t-d_{\source' v}}(v)
 \]
 Since ${\cal U}^{t-d_{\source' v}}(v)\subset \tilde{{\cal U}}^{t-d_{\source' v}}(v),$ we have
 \[
 \cup_{j=0}^t\cup_{v\in {\cal C}_j} \tilde{{\cal U}}^{t-d_{\source' v}}(v) \supseteq \cup_{j=0}^t\cup_{v\in {\cal C}_j} {\cal U}^{t-d_{\source' v}}(v).
 \]
 We only need to show that
 \[
 \cup_{j=0}^t\cup_{v\in {\cal C}_j} \tilde{{\cal U}}^{t-d_{\source' v}}(v) \subseteq \cup_{j=0}^t\cup_{v\in {\cal C}_j} {\cal U}^{t-d_{\source' v}}(v).
 \]
  For any node $w\in \cup_{j=0}^t\cup_{v\in {\cal C}_j} \tilde{{\cal U}}^{t-d_{\source' v}}(v),$ we consider the following cases.
 \begin{itemize}
 \item If $w$ is a collision node, we have $w\in {\cal U}^{t-d_{\source' w}}(w).$ Hence
 $w\in \cup_{j=0}^t\cup_{v\in {\cal C}_j} {\cal U}^{t-d_{\source' v}}(v).$
 \item If $w$ is not a collision node, there exists $v'$ such that $w \in \tilde{{\cal U}}^{t-d_{\source' v'}}(v').$
 \begin{itemize}
 \item  If the shortest path from $w$ to $v'$ does not contain any other collision nodes, we have $w\in {\cal U}^{t-d_{\source' v'}}(v').$
 \item If the shortest path from $w$ to $v'$ contains other collision nodes, denote by $u$ the collision node on that path which is the closest to node $w.$ Therefore, there is no collision node on the shortest path from node $u$ to node $w.$ We have
 \[
 d_{uw}=d_{wv'}-d_{uv'}
 \]
 Note we have $w\in \tilde{{\cal U}}^{t-d_{\source' v'}}(v'),$ therefore $d_{wv'}\leq t-d_{\source' v'}.$ Hence, we have
 \[
 d_{uw}=d_{wv'}-d_{uv'}\leq t-d_{\source' v'}-d_{uv'}\leq_{(a)} t-d_{\source'u},
 \]
 where $(a)$ is due to the triangle inequality. Therefore, we have
 \[
 d_{uw}\leq t-d_{\source'u},
 \]
 and the shortest path from node $u$ to node $w$ contains no collision nodes.
 Hence,
 \[
 w\in {\cal U}^{t-d_{\source' u}}(u)
 \]

 \end{itemize}
 \end{itemize}

As a summary, we proved
\[
 \cup_{j=0}^t\cup_{v\in {\cal C}_j} \tilde{{\cal U}}^{t-d_{\source' v}}(v) = \cup_{j=0}^t\cup_{v\in {\cal C}_j} {\cal U}^{t-d_{\source' v}}(v)
\]
Now we can use the collision removed BFS tree to bound $H$ since the branch of collision nodes of traditional BFS tree are already counted in the collision removed BFS tree rooted at these collision nodes. For example, consider the collision removed BFS tree from node $w$ in Figure \ref{fig:Hupperbound}. We ignore the presence of node $u$ since the branch of node $u$ are already considered in the collision removed BFS tree rooted at node $u.$

Hence, we have
\begin{align*}
H&\leq \left|\cup_{j=0}^t\cup_{v\in {\cal C}_j} \tilde{{\cal U}}^{t-d_{\source' v}}(v)\right| \\
&=\left|\cup_{j=0}^t\cup_{v\in {\cal C}_j} {\cal U}^{t-d_{\source' v}}(v)\right|\\
&\leq \sum_{j=0}^t\sum_{v \in {\cal C}_j}U^{t-d_{\source' v}}(v)
\end{align*}

Since $U^{\lambda}(u)$ is an increasing function of $\lambda$ and $d_{\source'v}\geq |k-j|$ for node $v$ on level $j,$ we have
\[
H\leq \sum_{j=0}^t\sum_{v \in {\cal  C}_j}U^{t-|k-j|}(v).
\]
We next establish the lemma using the following steps.
\begin{itemize}
\item {\bf Step 1: Upper bound on $U^{t-|k-j|}(w).$}
Denote by $\hbox{par}(w)$ the parent of $w$ on the BFS tree from the source and denote by $\hbox{par}^i(w)$ the $i$th ancestor of $w$ on the BFS tree from the source. For example, $v$ is the first ancestor of $w$ ($\hbox{par}(w) = v$) and $\hbox{par}(v)$ is the second ancestor of $w$ ($\hbox{par}^2(w)= \hbox{par}(v)$) as shown in Figure \ref{fig:Hupperbound}. Denote by $\sigma^h(v)$ the number of nodes in the collision removed BFS subtree rooted in node $v$ with height $h$ without branch of $\hbox{par}(v).$

Consider node $w$ in Figure \ref{fig:Hupperbound} and ignore the presence of the collision nodes. Firstly, we remove the branch of the parent of $w.$ The remaining nodes on the tree are below the level of $w$. This is because the level of a node $w'$s neighbor can differ from node $w'$s level by at most one and those neighbors that are at the same or higher levels must be collision nodes. The height of the tree is no larger than the total number of hops $\lambda.$ On the other hand, the tree only contains the nodes within $t$ hops from the actual source $\source$  since the tree is based on the infection subnetwork.  Since $w$ locates on level $j,$ the height of the tree must be no larger than $t-j.$ Therefore, the maximum height of the tree is $\min(t-j,\lambda)$ and denote by $\sigma^{\min(t-j,\lambda)}(w)$ the total number of nodes in the tree as shown in Figure \ref{fig:Hupperbound}.

Next, we consider the branch of the parent of $w$ ($v=\hbox{par}(w)$) in Figure \ref{fig:Hupperbound}. Note $w$ has only one parent $v$ (all other parents are collision nodes thus removed). Since we considered the $\lambda$ hops of the removed collision BFS tree rooted at $w$ and it takes one hop from $w$ to $v,$ the branch of node $v$ in the collision removed BFS tree is contained in ${\cal U}^{\lambda-1}(v).$

Therefore, we have
\begin{align}
U^{\lambda}(w) \leq& \sigma^{\min(t-j,\lambda)}(w)+U^{\lambda-1}(v)\\
=&\sigma^{\min(t-j,\lambda)}(w)+U^{\lambda-1}(\hbox{par}(w))\label{eqn:Urecursive}
\end{align}
Repeatedly using Equation (\ref{eqn:Urecursive}), we have
\begin{align}
U^{\lambda}(w)\leq\sum_{i=0}^{\min(\lambda,j)}\sigma^{\min(t-(j-i),\lambda-i)}(\hbox{par}^i(w)).\label{eqn:Urepeat}
\end{align}
Note the maximum number of hops upward is no larger than $\lambda$ and the total number of levels above $w$ is no larger than $j.$ Therefore, we only need to consider $\min(\lambda,j)$ levels above $w$ in Equation (\ref{eqn:Urepeat}).

Intuitively, the upper bound on $U^{\lambda}(w)$ is a collection of trees rooted at level $j-i$ with height $\min(t-(j-i),\lambda-i),\forall i\leq j.$ For example, in Figure \ref{fig:Hupperbound}, the blue area shows the tree rooted in level $j$ with height $\min(t-j,\lambda)$ and the green area shows in tree rooted in level $j-1$ with height $\min(t-(j-1),\lambda-1).$  In this example, we consider the removed collision BFS tree rooted at $w.$ The blue area is the collision removed BFS tree from $w$ after further removing the branch from $v$ and the green area is the collision removed BFS tree from $v$ by further removing the branch from $\hbox{par}(v).$ The height is no larger than $t-(j-1)$ since node $v$ locates on level $j-1$ and we consider the $t$ hop neighborhood of $s.$ In addition, the height is no larger than $\lambda-1$ since it takes one hop from node $w$ to node $v$ and the total number of possible hops from node $w$ is $\lambda.$

According to $E_1,$ we have
\[
\sigma^{l}(v)\leq\sum_{h=0}^{l}[(1+\delta)\mu]^h,\forall v \in \nodes(g_t).
\]
Hence,
\begin{align}
U^{\lambda}(w)\leq
&\sum_{i=0}^{\min(\lambda,j)}\sum_{h=0}^{\min(\lambda-i,t-(j-i))}[(1+\delta)\mu]^{h}.
\end{align}
Consider $\lambda\in(0,t)$ and $j\in[1,t].$ We obtain an upper bound on $U^{\lambda}(w)$ by analyzing different ranges of $\lambda.$
\begin{itemize}
\item $\lambda < j.$ We have $\min(\lambda,j)=\lambda.$
\begin{itemize}
\item $\lambda < t-j.$ We have $\min(\lambda-i,t-(j-i))=\lambda-i.$
Hence,
\begin{align}
U^{\lambda}(w)\leq
&\sum_{i=0}^{\lambda}\sum_{h=0}^{\lambda-i}[(1+\delta)\mu]^{h}\\
\leq & \sum_{i=0}^{\lambda}2[(1+\delta)\mu]^{\lambda-i}\\
\leq & 4[(1+\delta)\mu]^{\lambda}.
\end{align}
\item $\lambda\geq t-j.$ When $i\leq \frac{\lambda-t+j}{2},$ we have $\lambda-i>t-(j-i).$ Therefore, $\min(\lambda-i,t-(j-i))=t-(j-i).$ Hence,
\begin{align}
&U^{\lambda}(w)\\
\leq
&\sum_{i=0}^{\left\lfloor\frac{\lambda-t+j}{2}\right\rfloor}\sum_{h=0}^{t-j+i}[(1+\delta)\mu]^{h} \\
&+ \sum_{i=\left\lceil\frac{\lambda-t+j}{2}\right\rceil}^{\lambda}\sum_{h=0}^{\lambda -i}[(1+\delta)\mu]^{h}\\
\leq & \sum_{i=0}^{\left\lfloor\frac{\lambda-t+j}{2}\right\rfloor}2[(1+\delta)\mu]^{t-j+i} \\
&+ \sum_{i=\left\lceil\frac{\lambda-t+j}{2}\right\rceil}^{\lambda}2[(1+\delta)\mu]^{\lambda -i }\\
\leq & 4[(1+\delta)\mu]^{t-j+\left\lfloor\frac{\lambda-t+j}{2}\right\rfloor}\\
&+4[(1+\delta)\mu]^{\lambda -\left\lceil\frac{\lambda-t+j}{2}\right\rceil }\\
\leq & 4[(1+\delta)\mu]^{t-j+\frac{\lambda-t+j}{2}}+4[(1+\delta)\mu]^{\lambda -\frac{\lambda-t+j}{2} }\\
\leq & 8[(1+\delta)\mu]^{\frac{\lambda+t-j}{2}}
\end{align}
\end{itemize}
\item $\lambda \geq j.$ We have $\min(\lambda,j)=j.$
\begin{itemize}
\item $\lambda < t-j.$ We have $\min(\lambda-i,t-(j-i))=\lambda-i.$
Hence,
\begin{align}
U^{\lambda}(w)\leq
&\sum_{i=0}^{j}\sum_{h=0}^{\lambda-i}[(1+\delta)\mu]^{h}\\
\leq & \sum_{i=0}^{j}2[(1+\delta)\mu]^{\lambda-i}\\
\leq & 4[(1+\delta)\mu]^{j}.
\end{align}
\item $\lambda\geq t-j.$ When $i\leq \frac{\lambda-t+j}{2},$ we have $\lambda-i>t-(j-i).$ Hence,
\begin{align}
&U^{\lambda}(w)\\
\leq
&\sum_{i=0}^{\left\lfloor\frac{\lambda-t+j}{2}\right\rfloor}\sum_{h=0}^{t-j+i}[(1+\delta)\mu]^{h} \\
&+ \sum_{i=\left\lceil\frac{\lambda-t+j}{2}\right\rceil}^{j}\sum_{h=0}^{\lambda -i}[(1+\delta)\mu]^{h}\\
\leq & \sum_{i=0}^{\left\lfloor\frac{\lambda-t+j}{2}\right\rfloor}2[(1+\delta)\mu]^{t-j+i} \\
&+ \sum_{i=\left\lceil\frac{\lambda-t+j}{2}\right\rceil}^{j}2[(1+\delta)\mu]^{\lambda -i }\\
\leq & 4[(1+\delta)\mu]^{t-j+\left\lfloor\frac{\lambda-t+j}{2}\right\rfloor}+4[(1+\delta)\mu]^{\lambda -\left\lceil\frac{\lambda-t+j}{2}\right\rceil }\\
\leq & 4[(1+\delta)\mu]^{t-j+\frac{\lambda-t+j}{2}}+4[(1+\delta)\mu]^{\lambda -\frac{\lambda-t+j}{2} }\\
\leq & 8[(1+\delta)\mu]^{\frac{\lambda+t-j}{2}}
\end{align}
\end{itemize}
\end{itemize}
As a summary, we have
\[
U^{\lambda}(w)\leq \begin{cases} 4[(1+\delta)\mu]^{\min(\lambda,j)}. & \mbox{if } \lambda< t-j \\
8[(1+\delta)\mu]^{\frac{\lambda+t-j}{2}}& \mbox{if } \lambda\geq t-j.
\end{cases}
\]
For simplicity, we define $U^{\lambda}_j$ to be the upper bound on $U^{\lambda}(w)$  for $w$ on level $j.$ We have
\begin{align}
U^{\lambda}_j= \begin{cases} 4[(1+\delta)\mu]^{\min(\lambda,j)}. & \mbox{if } \lambda< t-j \\
8[(1+\delta)\mu]^{\frac{\lambda+t-j}{2}}& \mbox{if } \lambda\geq t-j.
\end{cases}\label{eqn:upperboundU}
\end{align}
and $U^{\lambda}(w)\leq U^{\lambda}_j$ where the subscript means the level of the nodes and the superscript means the number of hops.

Hence, we have
\[
H\leq  \sum_{j=0}^t\sum_{v \in {\cal  C}_j}U^{t-|k-j|}(v)\leq \sum_{j=0}^t|{\cal C}_j| U^{t-|k-j|}_j
\]
\item{\bf Step 2: Upper bound on $|{\cal C}_j|.$} Recall that ${\cal C}_j$ is the set of collision nodes on level $j.$  Note one collision edge may connect two nodes on the same level or connect one node on level $j$ and one node on level $j-1.$ Therefore, we have
\[
|{\cal C}_j|\leq 2\collisionsize_{j+1},\quad \forall j\leq t-1
\]
For $j=t,$ since we only consider the $t$ hop neighborhood of the actual source $s,$ we have
\[
|{\cal C}_t|\leq 2\collisionsize_{t}
\]
Therefore, we have
\begin{align}
H&\leq 2\collisionsize_{t}U^{t-|k-t|}_{t}+\sum_{j=0}^{t-1}2\collisionsize_{j+1}U^{t-|k-j|}_{j}\\
&=\underbrace{2\collisionsize_{t}U^{k}_{t}}_{(a)}+\underbrace{ \sum_{j=1}^{t}2\collisionsize_{j}U^{t-|k-(j-1)|}_{j-1}}_{(b)}\label{eqn:Hupperbound2}
\end{align}
\item{\bf Step 3} We analyze part $(a)$ and part $(b)$ in equation (\ref{eqn:Hupperbound2}) separately.
\begin{itemize}
\item {\bf Step 3.a: Upper bound on part $(a)$ in Equation (\ref{eqn:Hupperbound2})}

Define
\[
\alpha' = \frac{\alpha}{2}+\frac{1}{4}.
\]
and we have $\alpha'\in(1/2, 3/4).$
Since $t\leq \frac{\log n}{(1+\alpha)\log \mu},$ $\alpha<1$ and $\delta<1,$ we have when $n$ is sufficiently large,
\begin{align}
&t\leq \frac{\log n}{(1+\alpha)\log \mu}\\
&t \leq \frac{\log n}{(1+\alpha')\log [(1+\delta)\mu]}\\
&[(1+\delta)\mu]^{(1+\alpha')t}\leq n,\label{eqn:lowerBoundOnN}
\end{align}

According to Equation (\ref{eqn:upperboundU}), since $k\geq t-t =0,$ we have
\[
U_{t}^k =8[(1+\delta)\mu]^{\frac{k+t-t}{2}} = 8[(1+\delta)\mu]^{\frac{k}{2}}
\]

Based on event $E_2,$  we have
\begin{align*}
2\collisionsize_{t}U^{k}_{t}\leq& \frac{64}{n}[(1+\delta)\mu]^{2t+1+\frac{k}{2}}
\end{align*}
Since $[(1+\delta)\mu]^{(1+\alpha')t}\leq n,$
\[
2\collisionsize_{t}U^{k}_{t}\leq 64[(1+\delta)\mu]^{(1-\alpha')t+1+\frac{k}{2}}
\]
Since $k\leq t,$ we have
\begin{align}
2\collisionsize_{t}U^{k}_{t}\leq 64[(1+\delta)\mu]^{(\frac{3}{2}-\alpha')t+1} \label{eqn:HupperboundpartA}
\end{align}

\item {\bf Step 3.b: Upper bound on part $(b)$ in Equation (\ref{eqn:Hupperbound2})}

Based on event $E_3,$ we have
\[
\sum_{j=1}^{t}2\collisionsize_jU_{j-1}^{t-|k-(j-1)|} = \sum_{j=\lfloor m^-\rfloor+1}^{t}2\collisionsize_jU_{j-1}^{t-|k-(j-1)|}
\]

Therefore, we only consider the cases when $j\geq \lfloor m^-\rfloor+1.$ We will show that
\[
t-|k-(j-1)| \geq t-(j-1),
\]
when $j\geq  \lfloor m^-\rfloor+1.$
As a consequence, we have
\begin{align*}
U_{j-1}^{t-|k-(j-1)|}&=8[(1+\delta)\mu]^{\frac{t-|k-(j-1)|+t-(j-1)}{2}}\\
&=8[(1+\delta)\mu]^{t-\frac{(j-1)+|k-(j-1)|}{2}}
\end{align*}

Based on $t\leq \frac{\log n}{(1+\alpha)\log \mu},$ we have
\begin{align}
\frac{t}{2}&\leq \frac{\log n}{2(1+\alpha)\log \mu}\\
&\leq \lfloor m^-\rfloor \label{eqn:m-lowerbound}
\end{align}
for sufficiently large $n.$ Therefore, we have $j-1\geq \lfloor m^-\rfloor>\frac{t}{2}.$

When $j-1\leq k\leq t,$ we have $|k-(j-1)|\leq \frac{t}{2}\leq j-1.$ When $0\leq k< j-1,$ we have $|k-(j-1)|=(j-1)-k\leq j-1.$ We have
\[
|k-(j-1)| \leq j-1, \forall k\in[0,t].
\]
Hence,
\[
t-|k-(j-1)| \geq t-(j-1).
\]
Therefore, for all the discussions in Step 3.b, based on Equation \ref{eqn:Hupperbound2}, we have
\[
U_{j-1}^{t-|k-(j-1)|}=8[(1+\delta)\mu]^{t-\frac{(j-1)+|k-(j-1)|}{2}}
\]

Based on event $E_2,$  we have
\begin{align*}
&\sum_{j=1}^{t}2\collisionsize_jU_{j-1}^{t-|k-(j-1)|}\\
\leq&128\mu\sum_{j= \lfloor m^-\rfloor+1}^{\lceil m^+\rceil-1}  [(1+\delta)\mu]^{t-\frac{j-1+|k-(j-1)|}{2}}\\
 +& \frac{64}{n}\sum_{j=\lceil m^+\rceil}^{t}[(1+\delta)\mu]^{2j+1+t-\frac{j-1+|k-(j-1)|}{2}}
\end{align*}
Since $[(1+\delta)\mu]^{(1+\alpha')t}\leq n,$
\begin{align*}
&\sum_{j=1}^{t}2\collisionsize_jU_{j-1}^{t-|k-(j-1)|}\\
\leq & 128\mu\sum_{j= \lfloor m^-\rfloor+1}^{\lceil m^+\rceil-1}  [(1+\delta)\mu]^{t-\frac{j-1+|k-(j-1)|}{2}} \\
&+ 64\sum_{j=\lceil m^+\rceil}^{t}[(1+\delta)\mu]^{2j+1-\alpha' t-\frac{j-1+|k-(j-1)|}{2}}\\
\end{align*}
Next, we discuss the upper bounds for different $k$ values. We first show several necessary inequalities.
We have
\begin{align}
m^+-m^-=\frac{2\log \mu+\log 8}{2\log[(1+\delta)\mu]}\leq 1+\frac{\log 8}{2\log \mu}< 2.\label{eqn:m+-inequality}
\end{align}
Recall, we consider the case where $t>\lfloor m^-\rfloor.$ As shown in \ref{eqn:m-lowerbound}, we have
\begin{align}
\lfloor m^-\rfloor\in\left[\frac{t}{2},t\right)\label{eqn:m-allbounds}
\end{align}
Hence, we have
\begin{align}
\lceil m^+\rceil \in\left[\frac{t}{2},t+3\right)\label{eqn:m+allbounds}
\end{align}
Then, we consider the following cases for different $k$ values. Recall that $k$ is the level of node $s'.$

\begin{itemize}
\item $k\leq \lfloor m^-\rfloor.$
We have
\begin{align*}
&\sum_{j=1}^{t}2\collisionsize_jU_{j-1}^{t-|k-(j-1)|}\\
\leq & 128\mu\sum_{j= \lfloor m^-\rfloor+1}^{\lceil m^+\rceil-1}  [(1+\delta)\mu]^{t-\frac{j-1+|k-(j-1)|}{2}} \\
&+ 64\sum_{j=\lceil m^+\rceil}^{t}[(1+\delta)\mu]^{2j+1-\alpha' t-\frac{j-1+|k-(j-1)|}{2}}\\
= & 128\mu\sum_{j= \lfloor m^-\rfloor+1}^{\lceil m^+\rceil-1}  [(1+\delta)\mu]^{t-\frac{j-1-(k-(j-1))}{2}} \\
&+ 64\sum_{j=\lceil m^+\rceil}^{t}[(1+\delta)\mu]^{2j+1-\alpha' t-\frac{j-1-(k-(j-1))}{2}}\\
= & 128\mu\sum_{j= \lfloor m^-\rfloor+1}^{\lceil m^+\rceil-1}  [(1+\delta)\mu]^{t-j+1+\frac{k}{2}} \\
&+ 64\sum_{j=\lceil m^+\rceil}^{t}[(1+\delta)\mu]^{-\alpha' t+j+2+\frac{k}{2}}\\
\leq & 256\mu [(1+\delta)\mu]^{t-\lfloor m^-\rfloor+\frac{k}{2}} \\
&+ 128[(1+\delta)\mu]^{(1-\alpha')t+2+\frac{k}{2}}\\
\leq_{(a)} & 256 [(1+\delta)\mu]^{\frac{3}{4}t}+ 128[(1+\delta)\mu]^{(\frac{3}{2}-\alpha')t+2}
\end{align*}
where $(a)$ is due to $k\leq \lfloor m^-\rfloor$ and $ \frac{t}{2}\leq\lfloor m^-\rfloor< t.$
\item $\lfloor m^-\rfloor+1\leq k\leq \lceil m^+\rceil -2.$ In this case, we have
\[
k\in[\frac{t}{2}+1,t+1),
\]
according to Inequalties (\ref{eqn:m+allbounds}) and (\ref{eqn:m-allbounds}).
Hence, we have
\begin{align*}
&\sum_{j=1}^{t}2\collisionsize_jU_{j-1}^{t-|k-(j-1)|}\\
\leq & 128\mu\sum_{j=\lfloor m^-\rfloor+1}^{k}  [(1+\delta)\mu]^{t-\frac{j-1+|k-(j-1)|}{2}}\\
&+128\mu\sum_{j=k+1}^{\lceil m^+\rceil -1}  [(1+\delta)\mu]^{t-\frac{j-1+|k-(j-1)|}{2}}\\
& + 64\sum_{j=\lceil m^+\rceil}^{t}[(1+\delta)\mu]^{2j+1-\alpha' t-\frac{j-1+|k-(j-1)|}{2}}\\
= & 128\mu\sum_{j=\lfloor m^-\rfloor+1}^{k}  [(1+\delta)\mu]^{t-\frac{j-1+(k-(j-1))}{2}}\\
&+128\mu\sum_{j=k+1}^{\lceil m^+\rceil -1}  [(1+\delta)\mu]^{t-\frac{j-1-(k-(j-1))}{2}} \\
&+ 64\sum_{j=\lceil m^+\rceil}^{t}[(1+\delta)\mu]^{2j+1-\alpha' t-\frac{j-1-(k-(j-1))}{2}}\\
= & 128\mu\sum_{j=\lfloor m^-\rfloor+1}^{k}  [(1+\delta)\mu]^{t-\frac{k}{2}}\\
&+128\mu\sum_{j=k+1}^{\lceil m^+\rceil -1}  [(1+\delta)\mu]^{t-j+1+\frac{k}{2}}\\
& + 64\sum_{j=\lceil m^+\rceil}^{t}[(1+\delta)\mu]^{-\alpha' t+j+2+\frac{k}{2}}\\
\leq_{(a)} & 256\mu[(1+\delta)\mu]^{t-\frac{k}{2}}\\
&+128[(1+\delta)\mu]^{(1-\alpha')t+2+\frac{k}{2}}\\
\leq & 256 [(1+\delta)\mu]^{\frac{3}{4}t+\frac{1}{2}}+ 128[(1+\delta)\mu]^{(\frac{3}{2}-\alpha')t+2},
\end{align*}
where $(a)$ is due to $m^+-m^-<2$ which we proved in Inequality (\ref{eqn:m+-inequality}) and $k\in[\frac{t}{2}+1,t+1).$
\item $k\geq \lceil m^+\rceil-1.$ In this case,
we have
\[
k\in[\frac{t}{2}-1,t],
\]
according to Inequalties (\ref{eqn:m+allbounds}) and (\ref{eqn:m-allbounds}).
 Hence, we have
\begin{align*}
&\sum_{j=1}^{t}2\collisionsize_jU_{j-1}^{t-|k-(j-1)|}\\
\leq & 128\mu\sum_{j= \lfloor m^-\rfloor+1}^{\lceil m^+\rceil-1}  [(1+\delta)\mu]^{t-\frac{j-1+|k-(j-1)|}{2}} \\
&+ 64\sum_{j=\lceil m^+\rceil}^{k}[(1+\delta)\mu]^{2j+1-\alpha' t-\frac{j-1+|k-(j-1)|}{2}}\\
&+64\sum_{j=k+1}^{t}[(1+\delta)\mu]^{2j+1-\alpha' t-\frac{j-1+|k-(j-1)|}{2}}\\
= & 128\mu\sum_{j= \lfloor m^-\rfloor+1}^{\lceil m^+\rceil-1}  [(1+\delta)\mu]^{t-\frac{j-1+(k-(j-1))}{2}} \\
&+ 64\sum_{j=\lceil m^+\rceil}^{k}[(1+\delta)\mu]^{2j+1-\alpha' t-\frac{j-1+(k-(j-1))}{2}}\\
&+64\sum_{j=k+1}^{t}[(1+\delta)\mu]^{2j+1-\alpha' t-\frac{j-1-(k-(j-1))}{2}}\\
= & 128\mu\sum_{j= \lfloor m^-\rfloor+1}^{\lceil m^+\rceil-1}  [(1+\delta)\mu]^{t-\frac{k}{2}} \\
&+ 64\sum_{j=\lceil m^+\rceil}^{k}[(1+\delta)\mu]^{2j+1-\alpha' t-\frac{k}{2}}\\
&+64\sum_{j=k+1}^{t}[(1+\delta)\mu]^{-\alpha' t+j+2+\frac{k}{2}}\\
\leq_{(a)} & 256\mu[(1+\delta)\mu]^{t-\frac{k}{2}}+128[(1+\delta)\mu]^{\frac{3k}{2}-\alpha' t+1}\\
&+128[(1+\delta)\mu]^{(1-\alpha')t+2+\frac{k}{2}}\\
\leq & 256[(1+\delta)\mu]^{\frac{3}{4}t+\frac{1}{2}}+128[(1+\delta)\mu]^{(\frac{3}{2}-\alpha')t+1}\\
&+128[(1+\delta)\mu]^{(\frac{3}{2}-\alpha')t+2}\\
\leq & 256[(1+\delta)\mu]^{\frac{3}{4}t+\frac{1}{2}}+256[(1+\delta)\mu]^{(\frac{3}{2}-\alpha')t+2}\\
\end{align*}
$(a)$ is due to $m^+-m^-<2$  which we proved in Inequality \ref{eqn:m+-inequality} and $k\in[\frac{t}{2}-1,t].$
\end{itemize}
Therefore, we obtain a universal bound for different $k.$
\begin{align}
\sum_{j=1}^{t}2\collisionsize_jU_{j-1}^{t-|k-(j-1)|}&\leq c'[(1+\delta)\mu]^{\frac{3}{4}t+\frac{1}{2}}\\
&+c'[(1+\delta)\mu]^{(\frac{3}{2}-\alpha')t+2},\label{eqn:HupperboundpartB}
\end{align}
for all $k\in[1,t],$ where $c'\geq 256.$
\end{itemize}

\end{itemize}

As a summary, based on Equations (\ref{eqn:Hupperbound2}),(\ref{eqn:HupperboundpartA}) and (\ref{eqn:HupperboundpartB})
\begin{align}
H&\leq c[(1+\delta)\mu]^{\frac{3}{4}t+\frac{1}{2}}+c[(1+\delta)\mu]^{(\frac{3}{2}-\alpha')t+2}\\
&= c[(1+\delta)\mu]^{\frac{3}{4}t+\frac{1}{2}}+c[(1+\delta)\mu]^{(\frac{5}{4}-\frac{\alpha}{2})t+2}
\end{align}
for all $k\in[1,t],$ where $c\geq 257.$
\end{proof}

\section{Proof of Theorem \ref{thm:approximateRatio}}\label{sec:approximateRate}

\begin{proof}
Similar to the proof of Theorem \ref{thm:sourceIsJordanER}, we assume $E_1,E_2$ and $E_3$ occur. The BND one node can have is bounded by the number of all infected nodes. Therefore, the upper bound on BND is the sum of degree of all infected nodes. The edges of one infected node compose three parts: (1) the edge which infects the node; (2) the edges between the node and its offsprings; (3) the collision edges attaching to the node. Therefore, the total degree of all the infected nodes is upper bounded by
\[
\sum_{i=0}^t Z^{\leq t}_i+\sum_{i=0}^t Z^{\leq t}_i[(1+\delta)\mu]+2\collisionsize_{t+1}
\]
To use $\sum_{i=0}^t Z^{\leq t}_i[(1+\delta)\mu]$ above as the upper bound on offsprings, we need to extend $E_1$ to the range of $\levelset_{t}.$ It is easy to check that
\[
E'_1=\{\forall v \in \levelset_{\obsTime}, \phi'(v)\in((1-\delta)\mu,(1+\delta)\mu)\}.
\]
happens with a high probability with the same proof of Lemma \ref{lem:ERAllOffspring}.

A lower bound on BND of the actual source is
\[
Z^t_t[(1-\delta)\mu]
\]
according to $E'_1.$
Therefore, we have
\begin{align}
&\frac{Z^t_t[(1-\delta)\mu]}{\sum_{i=0}^t Z^{\leq t}_i+\sum_{i=0}^t Z^{\leq t}_i[(1+\delta)\mu]+2\collisionsize_{t+1}}\\
&= \frac{(1-\delta)\mu}{\frac{\sum_{i=0}^t Z^{\leq t}_i}{Z^t_t}(1+(1+\delta)\mu)+\frac{2\collisionsize_{t+1}}{Z^t_t}}\\
&\geq_{(a)} \frac{(1-\delta)\mu}{\frac{1+(1+\delta)\mu}{1-\epsilon'}+\frac{2\collisionsize_{t+1}}{Z^t_t}}.\\
&\geq \frac{(1-\delta)}{\frac{\frac{1}{\mu}+(1+\delta)}{1-\epsilon'}+\frac{2\collisionsize_{t+1}}{Z^t_t\mu}}\\
&\geq_{(b)}\frac{(1-\delta)}{\frac{\delta''+(1+\delta)}{1-\epsilon'}+\delta'}\\
&\geq \frac{1-\delta}{1+1.1\delta},\label{eqn:lowerBoundOnDegreeCount}
\end{align}
where inequality (a) holds due to Lemma \ref{lem:infectedNodeRatio1}, inequality (b) is based on Lemma \ref{lem:infectedNodeRatio2} and $\delta'',\delta',\epsilon'$ can be arbitrarily small when $n\rightarrow \infty$.

In the proof of Lemma \ref{lem:ERAllOffspring}, we need $\mu>\frac{2+\delta}{\delta^2}\log n$.
Hence, we have
\[
\frac{Z^t_t[(1-\delta)\mu]}{\sum_{i=0}^t Z^{\leq t}_i+\sum_{i=0}^t Z^{\leq t}_i[(1+\delta)\mu]+2\collisionsize_{t+1}}\geq \frac{1-\delta}{1+1.1\delta}
\]

when $\mu> \frac{2+\delta}{\delta^2}\log n.$

Assume we want the ratio to be $\geq 1-x,$ we have
\[
\frac{1-\delta}{1+1.1\delta} = 1-x.
\]
Therefore
\[
\delta = \frac{x}{2.1-1.1x}
\]
Hence,
\[
\frac{2+\delta}{\delta^2}=\frac{1.32x^2-7.14x+8.82}{x^2}
\]
Therefore, when
\[
\mu > \frac{9}{x^2}\log n>\frac{1.32x^2-7.14x+8.82}{x^2}\log n
\]
we have the ratio $\geq 1-x.$

Note, the condition that $\alpha>\frac{1}{2}$ is not used in the high probability result of $E_1,E_2,E_3.$ Therefore, we only need $\alpha\in(0,1)$ for this theorem. Therefore, the theorem is proved.

\end{proof}
\begin{lemma}\label{lem:infectedNodeRatio1}
If the conditions in Theorem \ref{thm:approximateRatio} hold and events $E_1,$ $E_2$ and $E_3$ occur, we have given any $\epsilon>0$, for sufficiently large $n,$ the following inequality holds
\[
\frac{Z^t_t}{\sum_{i=0}^t Z^{\leq t}_i}\geq 1-\epsilon.
\]
\end{lemma}
\begin{proof}
For any $\epsilon>0,$ we use induction and assume that
\[
Z_{m-1}^{m-1}\geq (1-\epsilon)\sum_{i=0}^{m-1} Z^{\leq m-1}_i
\]
Consider time slot $m,$ we have
\begin{align*}
\frac{Z^m_m}{\sum_{i=0}^m Z^{\leq m}_i}&=\frac{Z^m_m}{\sum_{i=0}^{m} Z^{\leq m-1}_i+\sum_{i=0}^m Z^m_i}\\
&=_{(a)}\frac{Z^m_m}{\sum_{i=0}^{m-1} Z^{\leq m-1}_i+\sum_{i=0}^{m-1} Z^m_i+Z_m^m}\\
&=\frac{1}{\underbrace{\frac{\sum_{i=0}^{m-1} Z^{\leq m-1}_i}{Z_m^m}}_{\hbox{A}}+\underbrace{\frac{\sum_{i=0}^{m-1}Z_i^m}{Z^m_m}}_{\hbox{B}}+1}
\end{align*}
In (a) we use that $Z_m^{\leq m-1}=0$ and $Z_i^{\leq m}=Z_i^{\leq m-1}+Z_i^m.$

Use induction assumption for part A, we have
\begin{align}
\frac{Z^m_m}{\sum_{i=0}^m Z^{\leq m}_i}\geq & \frac{1}{\frac{Z^{m-1}_{m-1}}{(1-\epsilon) Z_m^m}+\underbrace{\frac{\sum_{i=0}^{m-1}Z_i^m}{Z^m_m}}_{\hbox{B}}+1}\label{eqn:induction}
\end{align}
Based on event $E_3$, we have
\[
\frac{Z^{m-1}_{m-1}}{Z^m_m}\leq \frac{1}{(1-\delta)^2\mu q}
\]
and
\[
Z^m_m\geq [(1-\delta)^2\mu q]^m.
\]
Next, we establish an upper bound on $\sum_{i=0}^{m-1}Z_i^m.$ Note $\sum_{i=0}^{m-1}Z_i^m$ represents the number of nodes which are from level $0$ to level $m-1$ and are infected at time $m.$ Denote by $C({\cal V})$ the set of offsprings of node set ${\cal V}$ who are not collision nodes and were infected by node ${\cal V}.$ Define $C^2({\cal V}) = C(C({\cal V})).$ Recall the number of collsion nodes from level $0$ to level $m-1$ is upper bounded by $2\collisionsize_m.$ We establish an upper bound as following
\[
\sum_{i=0}^{m-1}Z_i^m\leq 2\collisionsize_m+\left| C\left(\cup_{i=0}^{m-2}{\cal Z}_i^{m-1}\right)\right|.
\]
Similarly, we have
\[
\left|\cup_{i=0}^{m-2}{\cal Z}_i^{m-1}\right|\leq 2\collisionsize_{m-1}+\left| C\left(\cup_{i=0}^{m-3}{\cal Z}_i^{m-2}\right)\right|.
\]
Based on $E_1,$ we have
\begin{align*}
\sum_{i=0}^{m-1}Z_i^m\leq 2\collisionsize_m+ 2\collisionsize_{m-1}[(1+\delta)\mu]+ \left| C^2\left(\cup_{i=0}^{m-3}{\cal Z}_i^{m-2}\right)\right|
\end{align*}
Repeating the step above, we have

\[
\sum_{i=0}^{m-1}Z_i^m\leq\sum_{j=0}^{m}2\collisionsize_{j} [(1+\delta)\mu]^{m-j}
\]
Based on $E_2,$ we evaluate the upper bound for different values of $m.$
\begin{itemize}
\item $0<m\leq \lfloor m^-\rfloor.$  We have
\[
\sum_{i=0}^{m-1}Z_i^m=0.
\]
Hence,
\begin{align*}
\frac{Z^m_m}{\sum_{i=0}^m Z^{\leq m}_i}\geq & \frac{1}{\frac{Z^{m-1}_{m-1}}{(1-\epsilon) Z_m^m}+1}\\
\geq &\frac{1}{\frac{1}{(1-\epsilon) (1-\delta)^2\mu q}+1}
\end{align*}
For any $\epsilon>0,$ we have
\[
\frac{1}{(1-\epsilon) (1-\delta)^2\mu q}\leq \epsilon
\]
for sufficiently large $n.$

Therefore, we have
\[
\frac{Z^m_m}{\sum_{i=0}^m Z^{\leq m}_i}\geq (1-\epsilon).
\]

\item $\lfloor m^-\rfloor<m<\lceil m^+\rceil.$ We have
\begin{align*}
\sum_{i=0}^{m-1}Z_i^m\leq& \sum_{j=\lfloor m^-\rfloor+1}^{\lceil m^+\rceil-1} 2\collisionsize_j[(1+\delta)\mu]^{m-j}\\
\leq_{(a)}& 32\mu [(1+\delta)\mu]^{m-\lfloor m^-\rfloor-1}
\end{align*}
$(a)$ is based on the fact that $\collisionsize_j\leq 8\mu$ and $\lceil m^+\rceil-\lfloor m^-\rfloor\leq 2.$

Hence, we have
\begin{align*}
&\frac{\sum_{i=0}^{m-1}Z_i^m}{Z^m_m}\\
\leq &\frac{32\mu [(1+\delta)\mu]^{m-\lfloor m^-\rfloor-1}}{[(1-\delta)^2\mu q]^m}\\
\leq & \frac{32}{\mu} \frac{(1+\delta)^{m-\lfloor m^-\rfloor-1}}{(1-\delta)^{2m}q^m\mu^{\lfloor m^-\rfloor-1}}\\
=& \frac{32}{\mu} \left(\frac{(1+\delta)^{1-\frac{\lfloor m^-\rfloor-1}{m}}}{(1-\delta)^2q\mu^{\frac{\lfloor m^-\rfloor-1}{m}}}\right)^m\\
\leq & \frac{32}{\mu} \left(\frac{(1+\delta)^{1-\frac{\lfloor m^-\rfloor-1}{\lceil m^+\rceil}}}{(1-\delta)^2q\mu^{\frac{\lfloor m^-\rfloor-1}{\lceil m^+\rceil}}}\right)^m
\end{align*}
Note $\lceil m^+\rceil<\lfloor m^-\rfloor+2.$ For sufficiently large $n,$ we have
\[
\frac{\lfloor m^-\rfloor-1}{\lceil m^+\rceil}\geq 1-\frac{3}{\lfloor m^-\rfloor}\geq \frac{1}{2}.
\]
Hence we have
\[
\frac{\sum_{i=0}^{m-1}Z_i^m}{Z^m_m}\leq \frac{32}{\mu} \left(\frac{(1+\delta)^{\frac{1}{2}}}{(1-\delta)^2q\mu^{\frac{1}{2}}}\right)^m\leq \frac{\epsilon}{2}\]
for sufficiently large $n.$

In addition, we have
\[
\frac{1}{(1-\delta)^2q\mu}\leq \frac{\epsilon}{2}
\]
for sufficiently large $n.$
\begin{align}
\frac{Z^m_m}{\sum_{i=0}^m Z^{\leq m}_i}\geq & \frac{1}{\frac{Z^{m-1}_{m-1}}{(1-\epsilon) Z_m^m}+\frac{\sum_{i=0}^{m-1}Z_i^m}{Z^m_m}+1}\\
\geq &\frac{1}{
\frac{\epsilon}{2(1-\epsilon)}+\frac{\epsilon}{2}+1}\geq(1-\epsilon).
\end{align}
\item $\lceil m^+\rceil\leq m\leq t.$ Define
\[
\alpha' = \frac{\alpha}{2}+\frac{1}{4}.
\]
and we have $\alpha'\in(1/4, 3/4).$ Follow the same argument in Equation (\ref{eqn:lowerBoundOnN}), we obtain that
\begin{align}
[(1+\delta)\mu]^{(1+\alpha')t}\leq n.\label{eqn:LowerBoundOnN2}
\end{align}
We have
\begin{align*}
\sum_{i=0}^{m-1}Z_i^m\leq& \sum_{j=\lfloor m^-\rfloor+1}^{\lceil m^+\rceil-1} 2\collisionsize_j[(1+\delta)\mu]^{m-j}\\
&+\sum_{j=\lceil m^+\rceil}^{m-1} 2\collisionsize_j[(1+\delta)\mu]^{m-j}\\
\leq& 32\mu [(1+\delta)\mu]^{m-\lfloor m^-\rfloor-1}\\
&+\frac{8}{n}\sum_{j=\lceil m^+\rceil}^{m-1}[(1+\delta)\mu]^{m+j+1}\\
\leq & 32\mu [(1+\delta)\mu]^{m-\lfloor m^-\rfloor-1}+\frac{16}{n}[(1+\delta)\mu]^{2m}\\
\leq & 32\mu [(1+\delta)\mu]^{m-\lfloor m^-\rfloor-1}\\
&+16[(1+\delta)\mu]^{2m-(1+\alpha')t}
\end{align*}
The last inequality holds based on Inequality (\ref{eqn:LowerBoundOnN2}).

Hence, we have
\begin{align*}
&\frac{\sum_{i=0}^{m-1}Z_i^m}{Z^m_m}\\
\leq &\underbrace{\frac{32\mu [(1+\delta)\mu]^{m-\lfloor m^-\rfloor-1}}{[(1-\delta)^2 q \mu]^m}}_{(A)}+\underbrace{\frac{16[(1+\delta)\mu]^{2m-(1+\alpha')t}}{[(1-\delta)^2 q \mu]^m}}_{(B)}
\end{align*}
$(A)$ has been handled in the previous case. For sufficiently large $n$ we have $(A)\leq \frac{\epsilon}{4}.$

Next, we focus on $(B).$ Since $m\leq t$, for sufficiently large $n,$ we have
\[
\frac{m+1}{t}\leq 1+\frac{\alpha'}{2}.
\]
Hence, for sufficiently large $n$
\begin{align*}
&\frac{16[(1+\delta)\mu]^{2m-(1+\alpha')t}}{[(1-\delta)^2 q \mu]^m}\\
= & \frac{16}{\mu}\left(\frac{(1+\delta)^{\frac{2m}{t}-1-\alpha'}}{(1-\delta)^{2\frac{m}{t}}q^{\frac{m}{t}}}\mu^{\frac{m+1}{t}-1-\alpha'}\right)^t\\
\leq & \frac{16}{\mu}\left(\frac{(1+\delta)^{2-\alpha'}}{(1-\delta)^2 q}\mu^{-\alpha'/2}\right)^t\\
\leq & \frac{\epsilon}{4}.
\end{align*}
Hence, we have
\[
\frac{\sum_{i=0}^{m-1}Z_i^m}{Z^m_m}\leq \frac{\epsilon}{2}
\]
Following the analysis in the previous case, we have
\[
\frac{Z^m_m}{\sum_{i=0}^m Z^{\leq m}_i}\geq 1-\epsilon.
\]
\end{itemize}

As a summary, we proved that
\[
\frac{Z^t_t}{\sum_{i=0}^t Z^{\leq t}_i}\geq 1-\epsilon.
\]
\end{proof}

\begin{lemma}\label{lem:infectedNodeRatio2}
If the conditions in Theorem \ref{thm:approximateRatio} hold and events $E_1,$ $E_2$ and $E_3$ occur, we have given any $\epsilon>0$, for sufficiently large $n,$ the following inequality holds
\begin{align}
\frac{2\collisionsize_{t+1}}{Z^t_t\mu}\leq \epsilon \label{eqn:ratio}
\end{align}
\end{lemma}
\begin{proof}

Note the upper bound of $\collisionsize_{t+1}$ can be obtained by a same proof of Lemma \ref{lem:E1} and the conclusions are the same when we extend the range from $t$ to $t+1$. Based on Lemma \ref{lem:E3}, we have
\[
Z^t_t\geq [(1-\delta)^2\mu]^t.
\]

When $t< \lceil m^+\rceil,$ Inequality \ref{eqn:ratio} trivially holds.

For $t\geq \lceil m^+\rceil,$ we have
\begin{align}
&\frac{2\collisionsize_{t+1}}{Z^t_t\mu}\\
\leq & \frac{8[(1+\delta)\mu]^{2t+3}}{n\mu[(1-\delta)^2\mu]^t}\\
= & \frac{8(1+\delta)^{2t+3}}{(1-\delta)^{2t}}\frac{1}{\mu^{\alpha t -2}}\\
= & \frac{8}{\mu}\left(\frac{1+\delta}{(1-\delta)^2} \times\frac{1}{\mu^{\frac{\alpha t -3}{2t+3}}}\right)^{2t+3}
\end{align}
Note $\frac{\alpha t -3}{2t+3}>0$ for sufficiently large $t$. Therefore, we have
\[
\frac{2\collisionsize_{t+1}}{Z^t_t\mu}\leq \epsilon.
\]
For sufficiently large $n.$
\end{proof}

\section{Proof of Lemma \ref{thm:diameterER}}\label{sec:proofImpossibility}
We present the proof of Theorem 4.2 from \cite{DraMas_10} with some minor changes to provide a more specific lower bound on $\mu q.$ This proof is included for the sake of completeness and is not a contribution of this paper.
\begin{proof}
Given some $\epsilon > 0,$ define
\[
d_{j}^{\pm}=\begin{cases} (1\pm\epsilon)^j\mu^j &\hbox{if } j =1,2, \\
(1\pm\epsilon)^2(1\pm \frac{\epsilon}{\mu})^{j-2}\mu^j & \hbox{if } j=3,\cdots,D'.
\end{cases}
\]
where $D'= \left\lceil\frac{\log n}{ 2\log \mu}\right\rceil.$
Define
\[
\Gamma_i(u) = \{v: d^g_{uv} = i\},
\]
and
\[
d_i(u)=|\Gamma_i(u)|.
\]
We first prove the following lemma.
\begin{lemma}
Let $\epsilon>0$ be fixed. Define for all $u\in\{1,\cdots,n\}$ and all $i=1,\cdots,D',$ the event $E_i(u)$ by
\[
E_i(u)=\{d_i^-\leq d_i(u)\leq d_i^+\}.
\]
Assumes that $\gamma_l\log n\leq\mu<< \sqrt{n},$ for large enough $n$,we have
\[
\Pr(E_i(u))\geq 1-D'n^{-\frac{\gamma_l\epsilon^2}{2+\epsilon}},~ u\in\{1,\cdots,n\},~i=1,\cdots, D'.
\]
\end{lemma}
\begin{proof}
Let $u\in\{1,\cdots,n\}$ and $i\in\{1,\cdots,D'\}$ be fixed. Note that, conditional on $d_1(u),\cdots,d_{i-1}(u),d_i(u)$ admits a binomial distribution with parameters
\begin{align*}
&{\cal L}(d_i(u)|d_1(u),\cdots,d_{i-1}(u))\\
=&\hbox{Bi}(n-1-d_1(u)-\cdots-d_{i-1}(u),1-(1-p)^{d_{i-1}(u)})
\end{align*}
where ${\cal L}(X|{\cal F})$ is the distribution of the random variable $X$ conditional on the event ${\cal F}.$ Denote by $\bar{E}_i(u)$ the complement of ${E}_i(u).$ It readily follows that
\begin{align*}
&\Pr(\bar{E}_i(u)|{E}_1(u),\cdots,{E}_{i-1}(u))\\
\leq&\Pr(\hbox{Bi}(n,1-(1-p)^{d_{i-1}^+})\geq d_{i}^+)\\
+&\Pr(\hbox{Bi}(n-1-d_1^+-\cdots-d_{i-1}^+,1-(1-p)^{d_{i-1}^-})\leq d_{i}^-).
\end{align*}
Note that, for all $j<D',$ one has
\begin{align*}
& d_j^-\leq d_j^+\leq d_{D'-1}^+ \\
&= (1+\epsilon)^2\left(1+\frac{\epsilon}{\mu}\right)^{D'-3}\mu^{D'-1}\\
&=\left(\frac{\mu(1+\epsilon)}{\mu+\epsilon}\right)^2(\mu+\epsilon)^{D'-1}\\
&\leq \left(\frac{\mu(1+\epsilon)}{\mu+\epsilon}\right)^2(\mu+\epsilon)^{\frac{\log n}{2\log \mu}}\\
& = \left(\frac{\mu(1+\epsilon)}{\mu+\epsilon}\right)^2\exp\left(\frac{\log n}{2}\frac{\log(\mu+\epsilon)}{\log\mu}\right)\\
& = \left(\frac{\mu(1+\epsilon)}{\mu+\epsilon}\right)^2\exp\left(\frac{\log n}{2}\right)\\
&\times  \exp\left(\frac{\log n}{2}\left(\frac{\log(\mu+\epsilon)}{\log\mu}-1\right)\right)\\
&= \left(\frac{\mu(1+\epsilon)}{\mu+\epsilon}\right)^2\exp\left(\frac{\log n}{2}\right)\\
&\times \exp\left(\frac{\log n}{2}\frac{\log(\mu+\epsilon)-\log\mu}{\log\mu}\right)\\
&= \left(\frac{\mu(1+\epsilon)}{\mu+\epsilon}\right)^2\\
\exp\left(\frac{\log n}{2}\right)&\times \exp\left(\frac{\log n}{2}\frac{\log(1+\epsilon/\mu)}{\log\mu}\right)
\end{align*}
Note $\log(1+x)\leq x$ for $x\geq 0.$
We have
\begin{align*}
& \leq \left(\frac{\mu(1+\epsilon)}{\mu+\epsilon}\right)^2\sqrt{n}\exp\left(\frac{\log n}{2}\frac{\epsilon}{\mu\log\mu}\right)\\
&\leq (1+\epsilon)^3\sqrt{n}
\end{align*}
Since  $\mu\geq \gamma_1\log n$, we have
\begin{align*}
 &\leq \left(\frac{\mu(1+\epsilon)}{\mu+\epsilon}\right)^2\sqrt{n}\exp\left(\frac{1}{2}\frac{\epsilon}{\gamma_l\log\mu}\right)\\
 &\leq (1+\epsilon)^2\sqrt{n}\exp\left(\frac{1}{2}\frac{\epsilon}{\gamma_l\log\mu}\right)
\end{align*}
For sufficiently large $n,$ we have $\mu$ is sufficiently large and $\exp\left(\frac{1}{2}\frac{\epsilon}{\gamma_l\log\mu}\right)\rightarrow 1$ as $n\rightarrow \infty.$ Hence, we have
\[
=(1+\epsilon)^2(1+o(1))\sqrt{n}
\]
Next, we compute the mean of $\hbox{Bi}(n,1-(1-p)^{d_{i-1}^+})$ and $\hbox{Bi}(n-1-d_1^+-\cdots-d_{i-1}^+,1-(1-p)^{d_{i-1}^-}).$

Since $i-1\leq D'-1,$ we have $d^+_{i-1}p=d^+_{i-1}\frac{\mu}{n}\rightarrow 0$ as $n\rightarrow \infty.$
Based on Taylor expansion, we have
we have
\[
(1-p)^{d^+_{i-1}} = 1-d^+_{i-1}p + o(d^+_{i-1}p)
\]
Hence,
\begin{align*}
&n(1-(1-p)^{d^+_{i-1}}) \\
=& d^+_{i-1}pn-o(d^+_{i-1}pn) = (1-o(1))d^+_{i-1}\mu
\end{align*}
Note
\begin{align*}
&n-1-d^+_1-\cdots-d^+_{i-1} \geq n-D'(1+\epsilon)^2(1+o(1))\sqrt{n} \\
\geq &n-(1+\epsilon)^2(1+o(1))\log n\sqrt{n}
\end{align*}
Therefore
\begin{align*}
&(n-1-d^+_1-\cdots-d^+_{i-1})(1-(1-p)^{d^-_{i-1}})\\
&\geq (n-(1+\epsilon)^2(1+o(1)\log n\sqrt{n})(d^-_{i-1}p - o(d^-_{i-1}p))\\
&= d^-_{i-1}pn-o(d^-_{i-1}pn)-d^-_{i-1}p(1+\epsilon)^2(1+o(1)\log n\sqrt{n}\\
&+o(d^-_{i-1}p(1+\epsilon)^2(1+o(1)\log n\sqrt{n})\\
&\geq (1-o(1))d^-_{i-1}\mu
\end{align*}
Using the Chernoff bound, we have
\begin{align*}
&\Pr\left(\hbox{Bi}(n,1-(1-p)^{d^+_{i-1}})\geq d^+_i\right)\\
&\leq \exp\left(-\frac{\xi^2}{2+\xi}n\left(1-(1-p)^{d^+_{i-1}}\right)\right)
\end{align*}
where
\[
(1+\xi)n\left(1-(1-p)^{d^+_{i-1}}\right) = d^+_{i}
\]
Therefore,
\[
\xi = \frac{d^+_{i}}{(1-o(1))d^+_{i-1}\mu}-1 \geq \begin{cases} \epsilon &\hbox{if } j =1,2, \\
\frac{\epsilon}{\mu} & \hbox{if } j=3,\cdots,D'.
\end{cases}
\]
Therefore, when $i=1,2$ we have
\begin{align*}
&\Pr\left(\hbox{Bi}(n,1-(1-p)^{d^+_{i-1}})\geq d^+_i\right)\\
&\leq \exp\left(-\frac{\xi^2}{2+\xi}(1-o(1))d^+_{i-1}\mu\right)\\
&\leq \exp\left(-\frac{\xi^2}{2+\xi}(1-o(1))\mu\right)\\
&\leq \exp\left(-\frac{\xi^2}{2+\xi}(1-o(1))\gamma_l \log n\right)\\
& \leq n^{-\gamma_l(1-o(1))\frac{\epsilon^2}{2+\epsilon}}
\end{align*}
when $i>2,$ we have
\begin{align*}
&\Pr\left(\hbox{Bi}(n,1-(1-p)^{d^+_{i-1}})\geq d^+_i\right)\\
&\leq \exp\left(-\frac{\xi^2}{2+\xi}(1-o(1))d^+_{i-1}\mu\right)\\
\end{align*}
Note, since $i>2,$ we have $d^+_{i-1}\geq (1+\epsilon)^2\mu^2.$ Hence, we have
\begin{align*}
&\leq \exp\left(-\frac{\xi^2}{2+\xi}(1-o(1)(1+\epsilon)^2\mu^3\right)\\
&\leq \exp\left(-\epsilon^2(1-o(1))(1+\epsilon)\mu\right)\\
&\leq n^{-\epsilon^2(1+\epsilon)(1-o(1)\gamma_l}
\end{align*}
Therefore, we have for all $i\leq D',$
\[
\Pr\left(\hbox{Bi}(n,1-(1-p)^{d^+_{i-1}})\geq d^+_i\right)\leq n^{-\gamma_l(1-o(1))\frac{\epsilon^2}{2+\epsilon}}
\]

Similarly, using the Chernoff bound, we have
\begin{align*}
&\Pr(\hbox{Bi}(n-1-d_1^+-\cdots-d_{i-1}^+,1-(1-p)^{d_{i-1}^-})\leq d_{i}^-)\\
&\leq \exp\left(-\frac{\xi'^2}{2}(n-1-d_1^+-\cdots-d_{i-1}^+)(1-(1-p)^{d_{i-1}^-})\right)
\end{align*}
where
\[
(1-\xi')(n-1-d_1^+-\cdots-d_{i-1}^+)(1-(1-p)^{d_{i-1}^-}) = d^-_{i}
\]
Therefore,
\[
\xi' = 1-\frac{d^-_{i}}{(1-o(1))d^-_{i-1}\delta}\geq \begin{cases} (1-\delta)\epsilon &\hbox{if } j =1,2, \\
(1-\delta)\frac{\epsilon}{\mu} & \hbox{if } j=3,\cdots,D'.
\end{cases}
\]
for any fixed $\delta \in(0,1)$ when $n$ is sufficiently large.
Therefore, when $i=1,2$ we have
\begin{align*}
&\Pr(\hbox{Bi}(n-1-d_1^+-\cdots-d_{i-1}^+,1-(1-p)^{d_{i-1}^-})\\
&\leq \exp\left(-\frac{\xi'^2}{2}(1-o(1)d^-_{i-1}\mu\right)\\
&\leq \exp\left(-\frac{\xi'^2}{2}(1-o(1)\mu\right)\\
&\leq \exp\left(-\frac{\xi'^2}{2}(1-o(1)\gamma_l \log n\right)\\
& \leq n^{-\gamma_l(1-o(1))\frac{(1-\delta)^2\epsilon^2}{2}}
\end{align*}
when $i>2,$ we have
\begin{align*}
&\Pr(\hbox{Bi}(n-1-d_1^+-\cdots-d_{i-1}^+,1-(1-p)^{d_{i-1}^-})\\
&\leq \exp\left(-\frac{\xi'^2}{2}(1-o(1))d^-_{i-1}\mu\right)\\
\end{align*}
Note, since $i>2,$ we have $d^-_{i-1}\geq (1-\epsilon)^2\mu^2.$ Hence, we have
\begin{align*}
&\leq \exp\left(-\frac{\xi'^2}{2}(1-o(1)(1-\epsilon)^2\mu^3\right)\\
&\leq \exp\left(-\frac{1}{2}(1-\delta)^2\epsilon^2(1-o(1))(1+\epsilon)\mu\right)\\
&\leq n^{-\frac{1}{2}(1-\delta)^2\epsilon^2(1-o(1))(1+\epsilon)\gamma_l}
\end{align*}
Therefore, we have for all $i\leq D',$
\begin{align*}
&\Pr(\hbox{Bi}(n-1-d_1^+-\cdots-d_{i-1}^+,1-(1-p)^{d_{i-1}^-})\\
&\leq n^{-\frac{\gamma_l(1-o(1))(1-\delta)^2\epsilon^2}{2}}
\end{align*}
Next, using union bounds, we have
\begin{align*}
&\Pr(E_i(u))\\
&\geq \Pr(E_1(u),\cdots,E_i(u))\\
&\geq \Pr(E_1(u),\cdots,E_{i-1}(u))\\
&-\Pr(\bar{E}_i(u)|E_1(u),\cdots,E_{i-1}(u))\\
&\geq 1-\sum_{j=1}^{i}\Pr(\bar{E}_j(u)|E_1(u),\cdots,E_{j-1}(u))\\
&\geq 1-D'n^{-\frac{\gamma_l(1-o(1))(1-\delta)^2\epsilon^2}{2}}-D'n^{-\gamma_l(1-o(1))\frac{\epsilon^2}{2+\epsilon}}\\
& \geq 1-D'n^{-\frac{\gamma_l\epsilon^2}{2+\epsilon}}
\end{align*}
for sufficiently large $n.$
\end{proof}

Next, we consider the upper bound of the diameter.

For any arbitrary nodes $u,v$, note that
\begin{align*}
&\Pr(d^g_{uv}>2D'+1|\Gamma_1(u),\cdots,\Gamma_{D'}(u),\Gamma_1(v),\cdots,\Gamma_{D'}(v))\\
&\leq (1-p)^{d_{D'}(u)d_{D'}(v)}
\end{align*}
Note if their $D'$ neighborhood has non-empty intersection, we have $d^g_{uv}\leq 2D'.$
Therefore, we obtain that
\[
\Pr(d^g_{uv}>2D'+1)\leq \Pr(\bar{E}_{D'}(u))+\Pr(\bar{E}_{D'}(v))+(1-p)^{(d^-_{D'})^2}
\]
The last term is evaluated as follows:
\begin{align*}
&(1-p)^{(d_{D'}^-)^2}\\
&\leq \exp(-p(d_{D'}^-)^2)\\
&=\exp\left(-p\left((1-\epsilon)^2(1-\frac{\epsilon}{\mu})^{D'-2}\mu^{D'}\right)^2\right)\\
&=\exp\left(-p\left(\frac{1-\epsilon}{1-\frac{\epsilon}{\mu}}\right)^4(\mu -\epsilon)^{2D'}\right)\\
&\leq \exp\left(-p\left(\frac{1-\epsilon}{1-\frac{\epsilon}{\mu}}\right)^4(\mu -\epsilon)^{\log n /\log \mu}\right)\\
& = \exp\left(-pn^{\frac{\log(\mu -\epsilon)}{\log \mu}}\left(\frac{1-\epsilon}{1-\frac{\epsilon}{\mu}}\right)^4\right)\\
& = \exp\left(-pn^{1+\frac{\log(1-\epsilon/\mu)}{\log \mu}}\left(\frac{1-\epsilon}{1-\frac{\epsilon}{\mu}}\right)^4\right)\\
& = \exp\left(-\mu e^{\log n\frac{\log(1-\epsilon/\mu)}{\log \mu}}\left(\frac{1-\epsilon}{1-\frac{\epsilon}{\mu}}\right)^4\right)\\
&\leq  \exp\left(-\mu(1-\frac{\epsilon}{\mu})^{-2}\left(1-\epsilon\right)^4\right)\\
&\leq\exp\left(-\mu \left(1-\epsilon\right)^4\right)\\
&\leq n^{-\gamma_l\left(1-\epsilon\right)^4}
\end{align*}

Therefore, we have
\[
\Pr(d^g_{uv}>2D'+1)\leq n^{-\gamma_l\left(1-\epsilon\right)^4}+2D'n^{-\frac{\gamma_l\epsilon^2}{2+\epsilon}}
\]
Finally, we have
\begin{align*}
&\Pr(\hbox{Diamter}>2D'+1)\leq \sum_{u\neq v}\Pr(d^g_{uv}>2D'+1)\\
&\leq n^2 \times\left(n^{-\gamma_l\left(1-\epsilon\right)^4}+2D'n^{-\frac{\gamma_l\epsilon^2}{2+\epsilon}}\right)
\end{align*}
Therefore, we have
\[
\gamma_l> \max\left(\frac{2}{(1-\epsilon)^4},\frac{2(2+\epsilon)}{\epsilon^2}\right)
\]
Note $\frac{2}{(1-\epsilon)^4}-\frac{2(2+\epsilon)}{\epsilon^2}$ is a increasing function for $\epsilon(0,1)$ and $\max\left(\frac{2}{(1-\epsilon)^4},\frac{2(2+\epsilon)}{\epsilon^2}\right) \geq 23.35.$ The optimal value is when $\epsilon = 0.459.$
Therefore, when
\[
\gamma_l>24.
\]
we have
\begin{align*}
\Pr(\hbox{Diamter}>2D'+1)&\leq \sum_{u\neq v}\Pr(d^g_{uv}>2D'+1)\\
&\leq n^{-\delta_2}+2D'n^{-\delta_1}
\end{align*}
where $\delta_1$ and $\delta_2$ is fixed positive constant. Note $D'\leq \log n.$ We have
\[
\lim_{n}\Pr(\hbox{Diamter}>2D'+1) = 0.
\]
Note $2\lceil x/2 \rceil\leq \lceil x\rceil.$ Hence
\[
2D'+1\leq D+2
\]
Hence, we have
\[
\lim_{n}\Pr(\hbox{Diamter}\leq D+2) = 1.
\]
The theorem is proved.
\end{proof}

\section{Necessary Inequalities}
We use the following Chernoff bounds.
\begin{lemma}\label{lem:chernoff}
Let $X_1,X_2,\cdots,X_n$ be i.i.d Poisson trials such that $\Pr(X_i)=p_i$. Let $X=\sum_{i=1}^n X_i$ and $E(X)=\mu.$ For any $\delta>0,$ we have
\[
\Pr(X\geq(1+\delta)\mu)\leq \left(\frac{e^\delta}{(1+\delta)^{1+\delta}}\right)^\mu\leq \exp\left(-\frac{\delta^2\mu}{2+\delta}\right)
\]
and for $\delta\in(0,1)$
\[
\Pr(X\leq(1-\delta)\mu)\leq \exp\left(-\frac{\delta^2\mu}{2}\right)
\]
\end{lemma}
\begin{proof}
We only need to prove $\left(\frac{e^\delta}{(1+\delta)^{1+\delta}}\right)^\mu\leq \exp\left(-\frac{\delta^2\mu}{2+\delta}\right).$ All other bounds are proved in \cite{MitUpf_05}.
We need to show
\begin{align*}
\left(\frac{e^\delta}{(1+\delta)^{1+\delta}}\right)^\mu&\leq \exp\left(-\frac{\delta^2\mu}{2+\delta}\right)\\
\mu\left(\delta-(1+\delta)\log(1+\delta)\right)&\leq -\frac{\delta^2\mu}{2+\delta}\\
(2+\delta)\delta-(1+\delta)(2+\delta)\log(1+\delta)+\delta^2&\leq 0\\
(2\delta-(2+\delta)\log(1+\delta))(1+\delta)&\leq 0\\
2\delta-(2+\delta)\log(1+\delta)&\leq 0
\end{align*}
Denote by $f(\delta)=2\delta-(2+\delta)\log(1+\delta).$ We have
\[
f'(\delta)=2-\log(1+\delta)-\frac{2+\delta}{1+\delta}=1-\log(1+\delta)-\frac{1}{1+\delta}
\]
\[
f''(\delta)=-\frac{1}{1+\delta}+\frac{1}{(1+\delta)^2}=\frac{1}{1+\delta}(\frac{1}{1+\delta}-1)\leq 0
\]
Hence, $f'(\delta)\leq f'(0)=0.$ Therefore, we have
\[
f(\delta)\leq f(0)=0
\]
Hence we prove the lemma.
\end{proof}
We need the following bounds
\begin{lemma}\label{lem:exponentialBounds}
When $x>0,$ we have
\[
1-x\leq e^{-x}
\]
and when $x\in(0,\frac{\log 2}{2}),$
\[
1-x\geq e^{-2x}.
\]
\end{lemma}
\begin{proof}
Let $f_1(x)=1-x-e^{-x}.$ We have
\[
f'_1(x)=-1+e^{-x}<0
\]
when $x>0.$
Hence, $f_1(x)\leq f_1(0)=0.$ Therefore, we have $1-x\leq e^{-x}.$

Let $f_2(x)=1-x-e^{-2x}.$ We have
\[
f'_2(x)=-1+2e^{-2x}.
\]
When $x<\frac{\log 2}{2},$ we have $f'_2(x)>0.$ Therefore $f_2(x)\geq f_2(0)=0.$ We have $1-x\geq e^{-2x}.$
\end{proof}
We obtain the following bound using the similar proof procedures. $\forall x>0, 1-\frac{1}{x}\leq \log(x)\leq x-1.$

\begin{lemma}
For $x\geq 2$ and integer $n\geq 0$ we have
\[
x^n \leq \sum_{i=0}^n x^i\leq 2x^n
\]
\end{lemma}
\begin{proof}
\begin{align*}
\sum_{i=0}^n x^n-2x^n&=\frac{x^{n+1}-1}{x-1}-2x^n\\
&=\frac{x^{n+1}-1-2x^{n+1}+2x^n}{x-1}\\
&=\frac{2x^n-1-x^{n+1}}{x-1}\\
&=\frac{x^n(2-\frac{1}{x^n}-x)}{x-1}\\
&\leq \frac{x^n(2-\frac{1}{x^n}-2)}{x-1}\\
&\leq 0
\end{align*}
Hence, we obtain the inequality in the lemma.
\end{proof}

\end{document}